\documentclass[11pt]{article}
\usepackage{amssymb}
\usepackage{amsthm}
\usepackage{amsmath}
\usepackage{mathtools}
\usepackage{graphicx,color,xcolor} 
\usepackage{amsfonts}
\usepackage[margin=1in]{geometry}
\usepackage{url}
\usepackage{natbib}
\usepackage[colorlinks,citecolor=blue!80!black,linkcolor=black,urlcolor=blue]{hyperref} 
\usepackage{subcaption}
\usepackage{bbm}
\usepackage{float}

\usepackage{authblk}
\graphicspath{ {plots_exploration/} }

\newcommand{\RNum}[1]{\uppercase\expandafter{\romannumeral #1\relax}}

\newcommand{\norm}[1]{\left\lVert#1\right\rVert} 

\newtheorem{theorem}{Theorem}

\newtheorem{setting}{Setting}

\newtheorem{lemma}[theorem]{Lemma}

\newtheorem{proposition}[theorem]{Proposition}
\newtheorem{remark}{Remark}

\newcommand{\indep}{\perp\!\!\!\perp}

\newcommand{\summ}[2]{\overset{#2}{\underset{#1}{\sum}}}

\graphicspath{{./plots_new/}}

\title{\textbf{Selective inference for multiple pairs of clusters} \\ \textbf{after $K$-means clustering}} 
\author[1]{Youngjoo Yun}
\author[1]{Yinqiu He}
\affil[1]{Department of Statistics, University of Wisconsin-Madison}
\date{}

\begin{document}
\sloppy
\maketitle
\begin{abstract}
If the same data is used for both clustering and for testing a null hypothesis that is formulated in terms of the estimated clusters, then the traditional hypothesis testing  framework often fails to control the Type I error. \cite{gao2022selective} and \cite{chen2023selective} provide selective inference frameworks for testing if a pair of estimated clusters indeed stem from underlying differences, for the case where hierarchical clustering and $K$-means clustering, respectively, are used to define the clusters. In applications, however, it is often of interest to test for 
multiple pairs of clusters. In our work, we extend the pairwise test of \cite{chen2023selective} to a test for multiple pairs of clusters, where the cluster assignments are produced by $K$-means clustering. We further develop an analogous test for the setting where the variance is unknown, building on the work of \cite{yun2023selective} that extends \citet{gao2022selective}'s pairwise test to the case of unknown variance. For both known and unknown variance settings, we present methods that address certain forms of data-dependence in the choice of pairs of clusters to test for. We show that our proposed tests control the Type I error, both theoretically and empirically, and provide a numerical study of their empirical powers under various settings. 
\end{abstract}

\section{Introduction}
Clustering is a widely used tool for studying unlabeled data that works by dividing a given data into groups based on certain similarity measures. The number of clusters to expect from a given data may not always be apparent, but several popular clustering algorithms require that it be specified, including $K$-means clustering. In such cases, 
researchers might try running the algorithm with a relatively large number of clusters and then examine if there  exists any difference among several of the estimated clusters. It could then be of interest
to do a statistical test to investigate if there is indeed an underlying group structure among these clusters. In this work, we propose methods for testing the null hypothesis that states there is no difference in the means of the cluster centers of an arbitrary number of clusters, where the cluster assignments are produced by the $K$-means algorithm run on the same data as the one used for inference; here, and in the rest of the paper, the cluster center of a cluster denotes the sample mean of the observations in the cluster.  

Since our null hypothesis is formulated in terms of the clusters defined by $K$-means clustering, it is inevitably data-dependent and breaches the assumption of the traditional hypothesis testing framework that all aspects of the inference procedure are determined independently of the data. As a result, traditional inference procedures may no longer be valid for such a setting---they could invalidate the p-value obtained, as well as undermine replicability, as discussed in \citet{benjamini2020selective}. While sample splitting is a widely used method for addressing issues of data-dependence, it is not applicable to statistical inferences involving clusters, as discussed in 
\citet{gao2022selective} and \citet{neufeld2024inference}. An alternative approach that is relevant to such settings is a framework called selective inference, which, first proposed in \citet{fithian2014optimal}, allows for valid inference in the presence of data-dependence in the inference procedure. It works by accounting for the event that leads to this dependence---more specifically, \citet{fithian2014optimal} account for this event by considering the selective Type I error, the Type I error conditioned on the selection event. 

To address the challenges of statistical inference for data-dependent clusters, 
various selective inference  procedures have been developed; examples include
\citet{gao2022selective} for hierarchical clustering algorithms, \citet{chen2023selective} for $K$-means clustering, \citet{bachoc2023selective} for convex clustering, and \citet{watanabe2021selective} for latent block models whose structure is determined by a clustering algorithm. 

In particular, 
\citet{gao2022selective} and \citet{chen2023selective} provide selective inference frameworks for testing for 
the difference in means between the cluster centers of a pair of clusters, which are chosen independently of the data from the $K$ estimated clusters. They consider the data generating distribution
where $n$ observations  $ X_i\in \mathbb{R}^q$ for $ i\in\{1,\ldots, n \}$ are generated independently as
\begin{align}\label{data_gen}
	X_i \sim  \mathcal{N}\left(\mu_i, \sigma^2 I_q\right)
\end{align}
with unknown $\mu_i\in \mathbb{R}^q$ and known $\sigma>0$. Let $\{\mathcal{C}_1,...,\mathcal{C}_K\}$ be a partition of \{1,...,n\}, where $\mathcal{C}_k\subset\{1,...,n\}$ denotes the set of indices of the observations $X_i$s that are assigned to the $k$th cluster. 
\cite{gao2022selective} and \cite{chen2023selective} 
 consider the null hypothesis that states 
\begin{align}\label{null:gao_chen}
        H_{0,\, \{k,k'\}}:\ \  \bar{\mu}_{\, \mathcal{C}_k}=\bar{\mu}_{\, \mathcal{C}_{k'}},
 \end{align}
where $k\neq k'\in \{1,...,K\},$ and $\bar{\mu}_{\mathcal{C}_k} = \sum_{i\in \mathcal{C}_k}\mu_i/|\mathcal{C}_k|$ denotes the mean of the cluster center of the $k$th cluster.

Several works have extended these works to other related settings. \cite{yun2023selective} study the case where the parameter $\sigma$ in \eqref{data_gen} is unknown, and \citet{gonzalez2023post} relax the assumption on the data generating process in \eqref{data_gen} by allowing for dependence across the observations. Furthermore, \citet{chen2023testing} develop a method for testing the null hypothesis analogous to \eqref{null:gao_chen} for a single feature, which is also studied in \citet{hivert2022post} but under different assumptions on the data. \citet{hivert2022post} additionally propose a method for testing for clusters that are in a specific arrangement with respect to each other. In this setting, they provide a method that combines the p-values produced by the test of \citet{gao2022selective} for the pairs of clusters involved. While this work also considers a test for multiple clusters, it differs from our proposed tests in that the latter can be applied to any collection of clusters.

In this work, we extend the work of both \cite{chen2023selective} and \cite{yun2023selective} by developing tests for multiple pairs of clusters for both known and unknown variances, where the clusters are produced by $K$-means clustering.

\subsection{Global null hypothesis for multiple pairs of clusters} 

Let $\mathcal{V}_{\mathrm{all}}=\{(k,k'): k< k'\in [K] \}$ denote the set of all possible index pairs out of integers 1 through $K,$ the number of clusters outputted by $K$-means clustering. We consider the null hypothesis that states  
\begin{align}\label{global_null}
    H_{0,\, \mathcal{V}}: \, \bar{\mu}_{\,\mathcal{C}_k} = \bar{\mu}_{\,\mathcal{C}_{k'}} \ \text{ for all }\,   (k,k')\in \mathcal{V}, 
\end{align}  
where $\mathcal{V}\subseteq \mathcal{V}_{\mathrm{all}}$ is defined as
the set of index pairs corresponding to the pairs of clusters to test for. Note that if $\mathcal{V}=\mathcal{V}_{\mathrm{all}},$ then the null hypothesis states that the cluster centers of all $K$ clusters have equal means, i.e., $H_{0,\mathcal{V}}:\bar{\mu}_{\mathcal{C}_1}=\cdots=\bar{\mu}_{\mathcal{C}_K}.$

The set $\mathcal{V}$ can be any subset of $\mathcal{V}_{\mathrm{all}},$ and the way in which it is chosen determines the corresponding inference procedure. 
\begin{itemize}
    \item 
If $\mathcal{V}$ is chosen from $\mathcal{V}_{\mathrm{all}}$ independently of the data, then the data-dependence in the null hypothesis exists only through the clustering procedure. In this setting, note that if $\mathcal{V}$ contains a single index pair $(k,k')$, then the null hypothesis in \eqref{global_null} reduces to the null hypothesis in \eqref{null:gao_chen} that tests for the pair of clusters $\mathcal{C}_k$ and $\mathcal{C}_{k'}.$ 
\item If the choice of $\mathcal{V}$ is data-dependent, then it introduces an additional selection event, which may lead to a lack of Type I error control if not accounted for. An example of such a data-dependent choice includes the selection of the pair of clusters  
that are the closest to each other among all pairs.  
\end{itemize} 
Our contributions  include the following:  
\begin{itemize}
    \item when the noise level $\sigma$ in \eqref{data_gen} is known, we provide a selective inference procedure for testing the null hypothesis in \eqref{global_null} for pre-specified $\mathcal{V},$ along with methods that address certain data-dependent choices of $\mathcal{V}$, and
    \item we develop analogous procedures for the case where $\sigma$ is unknown. 
\end{itemize} 
For each of the tests that we propose, we provide a p-value that can be computed exactly. 

An immediate test for the null hypothesis in \eqref{global_null} would be to combine the pairwise test of \citet{chen2023selective} with a correction for multiple comparisons. One such method for multiplicity adjustment is the Bonferroni correction, which is applicable to many settings due to the lack of distributional assumptions it makes, especially in this setting where the selective p-values may have complicated dependence structures; further discussion on the use of the Bonferroni correction in the context of testing for the null hypothesis in \eqref{global_null} can be found in Section \ref{sec:bonfcorr}. Our method differs from this testing procedure in that it is based on a single test statistic that combines signals across all pairs of clusters of interest. With both being valid tests that control the Type I error, it would be interesting to compare how the two methods perform in terms of power---we provide a simulation study in Section \ref{sec:simul_known_pre} that explores their empirical powers in different settings. 

The rest of the paper is organized as follows. In Section \ref{sec:Bon}, 
we review the pairwise test of \cite{chen2023selective}. In Sections \ref{sec:jointtest} and \ref{sec:group_dep}, we discuss the proposed tests for the null hypothesis in \eqref{global_null}, for the cases where the set $\mathcal{V}$ is pre-specified and chosen in a data-dependent way, respectively. In Section \ref{sec:var}, we propose analogous tests for the case of unknown $\sigma.$ Section \ref{sec:simul} presents a simulation study on the Type I error control and empirical powers of the proposed tests, followed by an application to a real data in Section \ref{sec:data}. All of the proofs can be found in the Appendix, and the codes for reproducing the empirical results are available at \url{https://github.com/yjyun97/cluster_inf_multiple}. 

Throughout the paper, we let $\|A\|_F$  denote the Frobenius norm of a matrix $A$, and let $\|\mathbf{a}\|_2$  denote the $\ell_2$-norm of a vector $\mathbf{a}$. For $r>0,$ we let $\chi_r$ denote the distribution of a random variable $\sqrt{Y}$ where $Y\sim \chi_r^2,$ the chi-squared distribution with $r$ degrees of freedom. For a null hypothesis $H_0$ and an event $\mathcal{A},$ $\mathbb{P}_{H_0}(\mathcal{A})$ denotes the probability of the event under $H_0$. {\color{black} In the case where a null hypothesis $H_0(f(X))$ depends on $f(X)$ for some function $f$ and random variable $X,$ $\mathbb{P}_{H_0(f(X))}(\mathcal{A}\mid f(X))$ denotes the conditional probability of an event $\mathcal{A}$ given $f(X)$ under $H_0(f(X))$---this notation appears in the statements of the theorems throughout this paper, which follow the style of those in \citet{yun2023selective}.} For a set $\mathcal{A}\subset \{1,...,m\}$ for some $m\in\mathbbm{N},$ we let $\mathbf{1}_{\mathcal{A}}\in\mathbb{R}^m$ denote a vector whose $i$th entry is 1 if $i\in \mathcal{A}$ and $0$ otherwise; likewise, $I_{\mathcal{A}}$ denotes a diagonal matrix whose $i$th diagonal entry is 1 if $i\in \mathcal{A}$ and 0 otherwise. For any $m\in \mathbb{N},$ $I_m$ and $\mathbf{1}_m$ denote an identity matrix in $\mathbb{R}^{m\times m}$ and a vector of 1s in $\mathbb{R}^m,$ respectively, and for $m,n\in \mathbb{N},$ $0_{m\times n}$ denotes a matrix of 0s and $0_m$ a vector of 0s. For a set $\mathcal{A},$ $|\mathcal{A}|$ denotes its cardinality, and for a matrix $A,$ $A_i$ denotes its $i$th row and $A^j$ its jth column. Finally, given a positive integer $m$, we define $[m]=\{1,\ldots,m\}$. 

\section{Pairwise test of \texorpdfstring{\cite{chen2023selective}}{chen}}\label{sec:Bon}
We first review the method of \cite{chen2023selective} for testing the null hypothesis in \eqref{null:gao_chen}, where $K$-means clustering is used for generating the cluster assignments. Let $X=[X_1\cdots X_n]^{\top} \in \mathbb{R}^{n\times q}$ denote the matrix consisting of the $n$ observations and
 $\bar{X}_{\mathcal{C}_k} \coloneqq \sum_{i\in \mathcal{C}_k}X_i/|\mathcal{C}_k|$ the cluster center of the $k$th cluster. Further define $\mathbf{v}_{k,k'} \coloneqq\frac{1}{|\mathcal{C}_{k}|}\mathbf{1}_{\mathcal{C}_{k}}-\frac{1}{|\mathcal{C}_{k'}|}\mathbf{1}_{\mathcal{C}_{k'}}$ and  $P_{\mathbf{v}_{k,k'}}:= \frac{\mathbf{v}_{k,k'}\mathbf{v}_{k,k'}^\top}{\|\mathbf{v}_{k,k'}\|_2^2},$ which
denotes the projection matrix that projects onto the span of $\mathbf{v}_{k,k'}.$ \cite{chen2023selective} propose the test statistic 
\begin{align}\label{eq:tkkprime}
	T_{\sigma, k,k'}\coloneqq \left\|\bar{X}_{\mathcal{C}_k}- \bar{X}_{\mathcal{C}_{k'}}\right\|_2 /\|\mathbf{v}_{k,k'}\|_2\sigma =\|P_{\mathbf{v}_{k,k'}}X\|_F/\sigma,
\end{align}  
which is the scaled distance between the cluster centers of the $k$th and $k'$th clusters. 

To account for the selection event, they condition on the clustering outcome---henceforth, let $\mathcal{C}(A)$ denote the outcome of $K$-means clustering run on the rows of a matrix $A\in \mathbb{R}^{n\times q},$ where the outcome refers to the cluster assignments generated in every iteration of the algorithm. They further condition on components of $X$ that are independent of $T_{\sigma, k, k'}$ 
in order to derive a p-value that can be computed. They thus provide the selective p-value  
$p_{\sigma, k,k'} = 1 - Q_{\sigma, k,k'} (T_{\sigma, k,k'}),$
where $ Q_{\sigma,k,k'}$ is the conditional CDF (cumulative distribution function) of 
$T_{\sigma,k,k'}$ given
\begin{align*}
    \mathcal{C}(X)\hspace{1em}\text{and}\hspace{1em}Z_{k,k'}\coloneqq\left(\frac{P_{\mathbf{v}_{k,k'}}X}{\|P_{\mathbf{v}_{k,k'}}X\|_F},\ \ P_{\mathbf{v}_{k,k'}}^\perp X\right),
\end{align*}
where $P_{\mathbf{v}_{k,k'}}^{\perp}:=I_n-P_{\mathbf{v}_{k,k'}}.$ They show that the p-value $p_{\sigma, k,k'},$ conditioned on $\mathcal{C}(X),$ is uniformly distributed under the null hypothesis in \eqref{null:gao_chen}; specifically, they derive that the distribution function $Q_{\sigma,k,k'}$ equals $F_{\chi_q}(\cdot; \mathcal{S}_{\sigma,k,k'}),$
which represents the CDF of the $\chi_q$ distribution truncated to the set 
\begin{align*}
    \mathcal{S}_{\sigma, k,k'}\coloneqq\{\psi \geq 0: \mathcal{C}(X)=\mathcal{C}(x_{\sigma,k,k'}(\psi))\},
\end{align*}
where  $x_{\sigma,k,k'}: [0, \infty) \rightarrow \mathbb{R}^{n\times q},$ 
\begin{align}\label{eq:decomp}
	x_{\sigma,k,k'}(\psi)=\psi \frac{\sigma P_{\mathbf{v}_{k,k'}}X}{\|P_{\mathbf{v}_{k,k'}}X\|_F}+P^\perp_{\mathbf{v}_{k,k'}}X. 
\end{align} 
In words, $\mathcal{S}_{\sigma, k,k'}$ is the set of realizations of the test statistic for which the $K$-means algorithm yields identical outcomes in every step of the algorithm to those of the algorithm run on $X.$  

\cite{chen2023selective} derive an explicit characterization of the truncation set $\mathcal{S}_{\sigma, k,k'}$ for the clusters produced by the $K$-means algorithm, specifically the standard Lloyd's algorithm (\citet{lloyd1982least}). Details of the algorithm can be found in \citet[Section 2.1]{chen2023selective}. 

To present \cite{chen2023selective}'s characterization of $\mathcal{S}_{\sigma, k,k'},$ we first define relevant notations that have been modified from those of \cite{chen2023selective}. Suppose that $K$-means clustering is applied to the rows of $A\in \mathbb{R}^{n\times q}$ for $J$ iterations. We let $m_l^{(0)}(A)$ denote the $l$th cluster center determined at the initialization step of the $K$-means algorithm, and for each $j\in \{0\}\cup [J]$, where $j=0$ represents the initialization step, we let $c_i^{(j)}(A)\in [K]$ denote the cluster assignment of the $i$th row of $A$ at $j$th step of the algorithm. Further define  
\begin{align*}
    \hspace{1em}w_{l,s}^{(j-1)}=\frac{\mathbf{1}\left\{c_s^{(j-1)}(X) =l\right\}}{\sum_{i=1}^n\mathbf{1}\left\{c_{i}^{(j-1)}(X) =l\right\}}
    \hspace{1em}\text{and}\hspace{1em}M^{(j-1)}_l (A)= \sum_{s=1}^n w_{l,s}^{(j-1)}A_{s}
\end{align*}
for $l\in[K]$ and $j\in[J],$ where $M_l^{(j-1)}(A)$ represents the average of the rows of $A$ in the $l$th cluster, where the clusters are defined by the output of the $j-1$th iteration of $K$-means clustering run on $X.$ For each $i\in[n],$ $l\in[K],$ and $j\in\{0\}\cup [J],$ define the function $\mathcal{S}_{i,l,j}:\mathcal{F}([0,\infty),\mathbb{R}^{n\times q})\rightarrow \mathcal{P}([0,\infty)),$ where $\mathcal{F}([0,\infty),\mathbb{R}^{n\times q})$ and $\mathcal{P}([0,\infty))$ denote the set of functions from $[0,\infty)$ to $\mathbb{R}^{n\times q}$ and  the power set of $[0,\infty),$ respectively, 
\begin{align*}
\mathcal{S}_{ i, l, j}(y)=\begin{cases}\{\psi\geq 0:\, \| [y(\psi)]_i-m^{(0)}_{c_i^{(0)}(X)}(y(\psi))\|^2_2\leq \|[y(\psi)]_i-m^{(0)}_{l}(y(\psi))\|^2_2\}&\text{ if }j=0\\[15pt]
\{\psi\geq 0:\, \|[y(\psi)]_i- M^{(j-1)}_{c_i^{(j)}(X)}(y(\psi))\|^2 _2 
\leq \|[y(\psi)]_i- M^{(j-1)}_l(y(\psi)) \|^2_2\}&\text{ if }j\in [J].
\end{cases} 
\end{align*}
Here, and throughout the rest of the paper, $[A]_i$ for a matrix $A$ and $i\in[n]$ is used interchangeably with $A_i$ for the clarity of notations. \citet{chen2023selective} show that 
\begin{align} \label{eq:charac} 
\mathcal{S}_{\sigma, k,k'}=\bigcap_{j=0}^J\bigcap_{l=1}^K \bigcap_{i=1}^n 	\mathcal{S}_{ i, l, j} (x_{\sigma,k,k'}). 
\end{align}
The authors further show that  
each $\mathcal{S}_{ i, l, j} (x_{\sigma, k,k'})$ is the set of solutions to a quadratic inequality, and thus $\mathcal{S}_{\sigma, k,k'}$ in \eqref{eq:charac} is the set of solutions to a system of 
$nK(J+1)$ quadratic inequalities.  

\section{Proposed test for $H_{0,\mathcal{V}}$ with pre-specified $\mathcal{V}$}\label{sec:jointtest}

Building on the work of \cite{chen2023selective}, we develop a test for the null hypothesis in \eqref{global_null}. Throughout this section, we consider the setting where $\sigma$ is known and $\mathcal{V}$ is pre-specified.

\subsection{Baseline testing procedure for $H_{0,\mathcal{V}}$}\label{sec:bonfcorr} 

Before presenting the proposed method, we first discuss a baseline test for the null hypothesis in \eqref{global_null} that combines the pairwise test of \cite{chen2023selective} reviewed in Section \ref{sec:Bon} with the Bonferroni correction. We consider the Bonferroni adjusted p-value given by
$p_{\sigma, \textnormal{Bon}} =\min_{(k,k')\in \mathcal{V}}\{|\mathcal{V}|\cdot p_{\sigma, k,k'},\ 1\},$ which, conditioned on $\mathcal{C}(X),$ follows the super-uniform distribution. As an illustration of the test's Type I error control, Figure \ref{figure_1} presents a QQ plot of a sample of $p_{\sigma,\mathrm{Bon}}$ generated under the null hypothesis in \eqref{global_null} with $\mathcal{V}=\mathcal{V}_{\mathrm{all}},$ i.e.,  $H_{0,\mathcal{V}}:\bar{\mu}_{\mathcal{C}_1}=\cdots = \bar{\mu}_{\mathcal{C}_K},$ for $K\in\{3, 5,7\}.$ The test, as expected, controls the Type I error, and we see that it becomes more conservative as $K$ increases. More details on the simulation settings of Figure \ref{figure_1} can be found in Section \ref{sec:simul_known_pre}.
\begin{figure}[ht]
    \centering
    {{\includegraphics[width=0.4\textwidth]{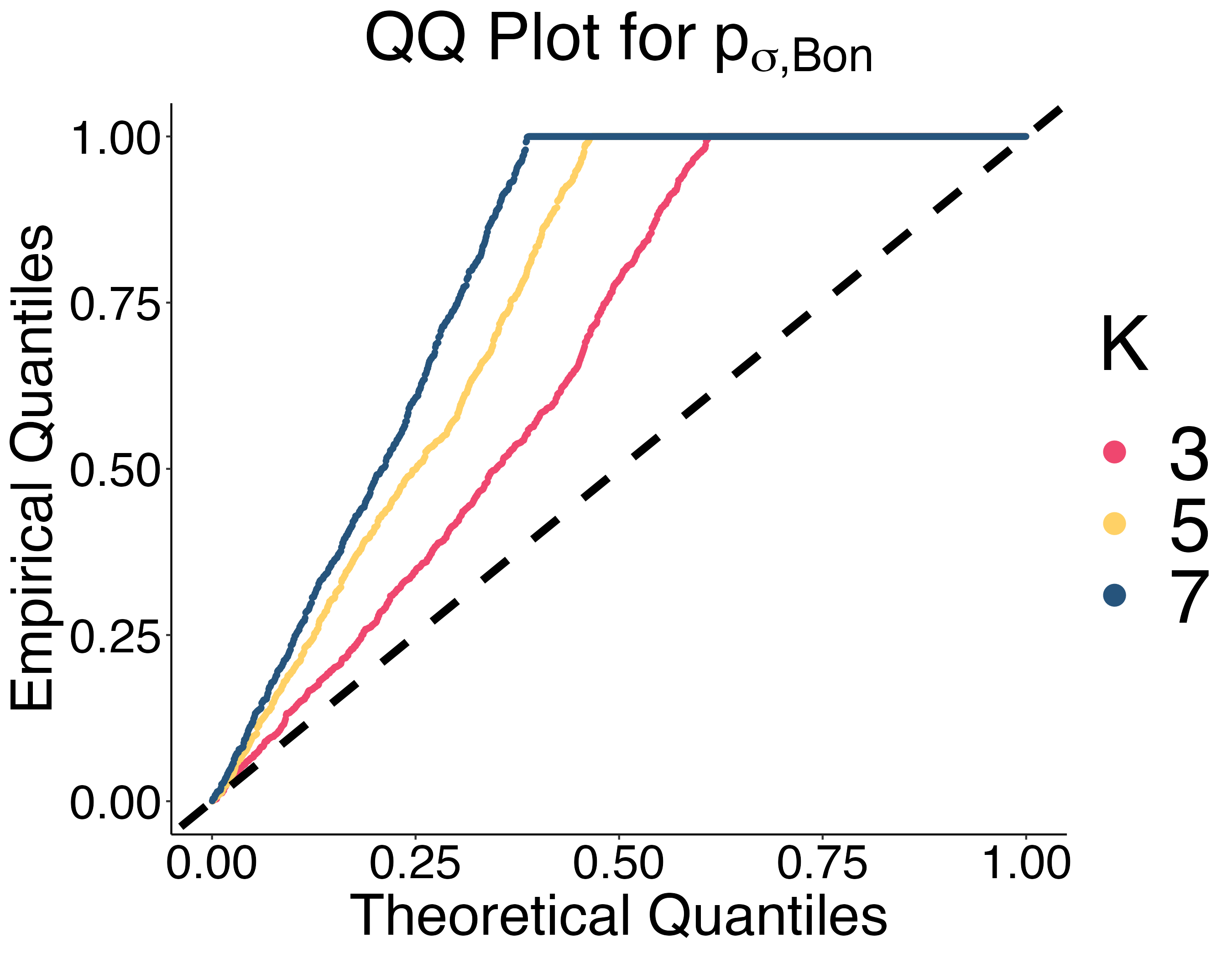}}}
    \caption{QQ plot of a 
sample of $p_{\sigma, \mathrm{Bon}}$ generated under the null hypothesis in \eqref{global_null} with $\mathcal{V}=\mathcal{V}_{\mathrm{all}}.$ }\label{figure_1}
\end{figure}

\subsection{Proposed method}\label{sec:proposed} 

We now present the proposed test for the null hypothesis in \eqref{global_null}. Define $P_{\mathcal{E}}$ as the matrix of the orthogonal projection that projects onto the space $\mathcal{E}\coloneqq \mathrm{span}\{\mathbf{v}_{k,k'}: \ (k,k')  \in \mathcal{V} \}.$ We propose the test statistic
\begin{align}\label{eq:tsigma}
T_{\sigma}=\frac{\|P_{\mathcal{E}}X\|_F}{\sigma}, 
\end{align}
which captures signals across all pairs of clusters of interest. Building on the work of \cite{gao2022selective} and \cite{chen2023selective}, we consider the decomposition of $X,$
\begin{align*}
X=T_{\sigma}\cdot \frac{\sigma P_{\mathcal{E}}X}{\|P_{\mathcal{E}}X\|_F}+ P_{\mathcal{E}}^\perp X,
\end{align*}
where $P_{\mathcal{E}}^{\perp}:= I_n - P_{\mathcal{E}}.$ With this decomposition in mind, we define the p-value as 
 $p_{\sigma } = 1-Q_{\sigma}(T_{\sigma }),$ where $Q_{\sigma}$ is the conditional CDF of $T_{\sigma}$ given 
\begin{align} 
\mathcal{C}(X)\hspace{1em}\text{and}\hspace{1em} Z\coloneqq\left(\frac{P_{\mathcal{E}}X}{\|P_{\mathcal{E}}X\|_F},\ \ P_{\mathcal{E}}^{\perp}X\right).\label{eq:condsetg}
\end{align}
Theorem \ref{thm:main_thm} provides an explicit form of the distribution function $Q_{\sigma}.$

\begin{theorem}\label{thm:main_thm}
Under the null hypothesis $H_{0, \mathcal{V}}$ in \eqref{global_null}, 
\begin{align*}
Q_{\sigma}=F_{\chi_{d}}\!\left(\cdot ;\,  \mathcal{S}_\sigma  \right),
\end{align*}
where $d= q\cdot \mathrm{dim}(\mathcal{E})$, and $F_{\chi_{d}}(\cdot;\,  \mathcal{S}_\sigma )$ 
 represents the CDF of the
 $\chi_{d}$ distribution truncated to the set 
\begin{align*}
	\mathcal{S}_\sigma \coloneqq \{\psi \geq 0: \mathcal{C}(X)=\mathcal{C}(x_\sigma (\psi))\}, \notag
\end{align*} 
where $x_\sigma:[0,\infty)\rightarrow \mathbb{R}^{n\times q},$
\begin{align*}
	x_\sigma(\psi)=\psi \, \cdot \, \frac{\sigma P_{\mathcal{E}}X}{\|P_{\mathcal{E}}X\|_F}+ P_{\mathcal{E}}^\perp X. 
\end{align*}  
Furthermore, $\mathbb{P}_{H_{0,\mathcal{V}}}(p_\sigma\leq \alpha \mid \mathcal{C}(X))=\alpha.$
\end{theorem}
It remains to characterize the selection event $\mathcal{S}_\sigma$ in Theorem \ref{thm:main_thm}. It follows from the characterization of $S_{\sigma, k,k'}$ of \citet{chen2023selective} discussed in Section \ref{sec:Bon} that 
\begin{align} \label{eq:charac2}
	\mathcal{S}_\sigma = \bigcap_{j=0}^J\bigcap_{l=1}^K\bigcap_{i=1}^n  	\mathcal{S}_{ i, l, j} (x_\sigma).
\end{align} 
By the definitions of the functions $\mathcal{S}_{ i, l, j}$ for $i\in[n],l\in[K],$ and $j\in\{0\}\cup [J]$ given above \eqref{eq:charac}, each 
 $\mathcal{S}_{i,l,j}(x_{\sigma})$ is characterized by inequalities of the form \begin{align}\label{eq:form1}
      \|{ [x_{\sigma}(\psi)]_i-[x_{\sigma}(\psi)]_{i'}}\|^2_2\leq \|{[x_{\sigma}(\psi)]_i-[x_{\sigma}(\psi)]_{i^{''}}}\|^2_2
 \end{align} 
if $j=0$ and those of the form 
\begin{align}\label{eq:form2}
    \|[x_{\sigma}(\psi)]_i-M_{l}^{(j-1)}(x_{\sigma}(\psi) ) \|_2^2 \leq \|[x_{\sigma}(\psi)]_i-M_{l'}^{(j-1)}(x_{\sigma}(\psi)) \|_2^2
\end{align} 
if $j\in [J],$
where $i',i''\in [n]$ and $l'\in[K].$
Proposition \ref{prop:quadratic} below, analogous to \citet[Lemma 2]{gao2022selective} and \citet[Lemma 2]{chen2023selective}, states that these inequalities can be written as quadratic inequalities in $\psi.$
\begin{proposition}
\label{prop:quadratic}
Let $D=\sigma{P_\mathcal{E}X}/{\|P_\mathcal{E}X\|_F}$ and $E=P_\mathcal{E}^\perp X.$ For all $i,i'\in [n],$ $l\in[K],$ and $j\in [J],$ 
\begin{align*}
    \norm{[x_\sigma(\psi)]_i-[x_\sigma(\psi)]_{i'}}_2^2=&~\lambda_{ii', 1}\psi^2+\lambda_{ii', 2}\psi+\lambda_{ii', 3}\text{ and} \\
    \norm{[x_\sigma(\psi)]_i-M_l^{(j-1)}(x_\sigma(\psi)) }_2^2=&~\lambda_{ilj,1}\psi^2+\lambda_{ilj,2}\psi+\lambda_{ilj,3}, 
\end{align*} 
with coefficients 
\begin{align*}
    \lambda_{ii',1}&=\|\mathbf{d}_{ii'}\|_2^2,  & \lambda_{ii', 2}&=2\langle{\mathbf{d}_{ii'}, \mathbf{e}_{ii'}}\rangle,  &   \lambda_{ii',3}&=\|\mathbf{e}_{ii'}\|_2^2, \\ 
    \lambda_{ilj,1}&=\|\mathbf{d}_{ilj}\|_2^2, &
    \lambda_{ilj,2}&=2\langle{\mathbf{d}_{ilj}, \mathbf{e}_{ilj}}\rangle, 
  & \lambda_{ilj,3}&=\|\mathbf{e}_{ilj}\|_2^2,
\end{align*} 
where $\mathbf{d}_{ii'}=D_i-D_{i'},$ $\mathbf{e}_{ii'}=E_i-E_{i'},$ $\mathbf{d}_{ilj}=D_i-M_l^{(j-1)}(D),$ and $\mathbf{e}_{ilj}=E_i-M_l^{(j-1)}(E).$
\end{proposition} 
It then follows from Theorem \ref{thm:main_thm} and Proposition \ref{prop:quadratic} that $p_\sigma$ can be computed exactly. 

\begin{remark}
The null hypothesis $H_{0,\mathcal{V}}$ reduces to the null hypothesis $H_{0,\{k,k'\}}$ if  $\mathcal{V}=\{(k,k')\}.$ In this setting, $p_\sigma,$ $T_{\sigma},$ and $S_\sigma$ reduce to $p_{\sigma, k,k'},$ $T_{\sigma, k,k'},$ and $S_{\sigma, k,k'},$ respectively, and thus 
the proposed test generalizes the pairwise test of \citet{chen2023selective}.
\end{remark}

\begin{remark}\label{rem:invariance}
Different choices of $\mathcal{V}$ can lead to equivalent statements of the null hypothesis $H_{0,\mathcal{V}}.$ For example, two different choices $\mathcal{V}_1$ and $\mathcal{V}_2$ lead to the same statement of $H_{0,\mathcal{V}}$ if the spans of $\{\mathbf{v}_{k,k'}: \ (k,k')  \in \mathcal{V}_1 \}$ and $\{\mathbf{v}_{k,k'}: \ (k,k') \in \mathcal{V}_2\}$ are equal. Under this condition, the proposed test remains the same since $T_\sigma$ is invariant to different specifications of $\mathcal{V}$ that correspond to the same set $\mathcal{E}.$ The baseline testing procedure discussed in Section \ref{sec:bonfcorr}, however, differs since different choices of $\mathcal{V}$ lead to different tests. The reason is that $p_{\sigma, \textnormal{Bon}}$ depends on both the cardinality of $\mathcal{V}$ and the specific pairs of clusters chosen for the pairwise test of \cite{chen2023selective}---Section \ref{sec:simul_known_pre} presents a comparison of empirical power of the baseline testing procedure for different specifications of $\mathcal{V}$ associated with the same set $\mathcal{E}.$
\end{remark}

\section{Proposed test for $H_{0,\mathcal{V}}$ with data-dependent choice of $\mathcal{V}$} \label{sec:group_dep}

In many applications, researchers may use the clustering outcomes to choose the pairs of clusters to test for. For instance, one may choose the index pair $(k,k')$ that corresponds to the two clusters whose cluster centers are the farthest (or the closest) from each other (or to each other) and then test for the null hypothesis in \eqref{global_null} with $\mathcal{V}=\{(k,k')\}$
to decide if the clusters $\mathcal{C}_k$ and $\mathcal{C}_{k'}$ indeed stem from underlying differences. In such cases, $\mathcal{V}\subset \mathcal{V}_{\mathrm{all}}$ is chosen in a data-dependent way, and the way in which it is chosen defines the additional selection event. Depending on how $\mathcal{V}$ is chosen, not accounting for the selection event may lead to a lack of Type I error control. The left-hand plot of Figure \ref{figure_2} presents a QQ plot of a sample of $p_{\sigma}$ generated under the null hypothesis in \eqref{global_null} with $\mathcal{V}=\{(k,k')\}$ that corresponds to the pair of clusters that are the farthest apart among the $K=20$ clusters. It illustrates that the distribution of $p_{\sigma}$ deviates from the uniform distribution, in such a way that the test based on $p_\sigma$ fails to control the Type I error. However, not accounting for the data-dependence in the selection of $\mathcal{V}$ may not always lead to a lack of Type I error control, as we see in the right-hand plot of Figure \ref{figure_2}, whose settings are analogous to those of the left-hand plot except that the pair of clusters that are the closest to each other is chosen. More discussion on this phenomenon can be found in Remark \ref{remark_dep}, and additional details on the simulation settings of Figure \ref{figure_2} can be found in Section \ref{sec:simul_known_dep}.

\begin{figure}[ht]
    \centering
    {{\includegraphics[width=0.6\textwidth]{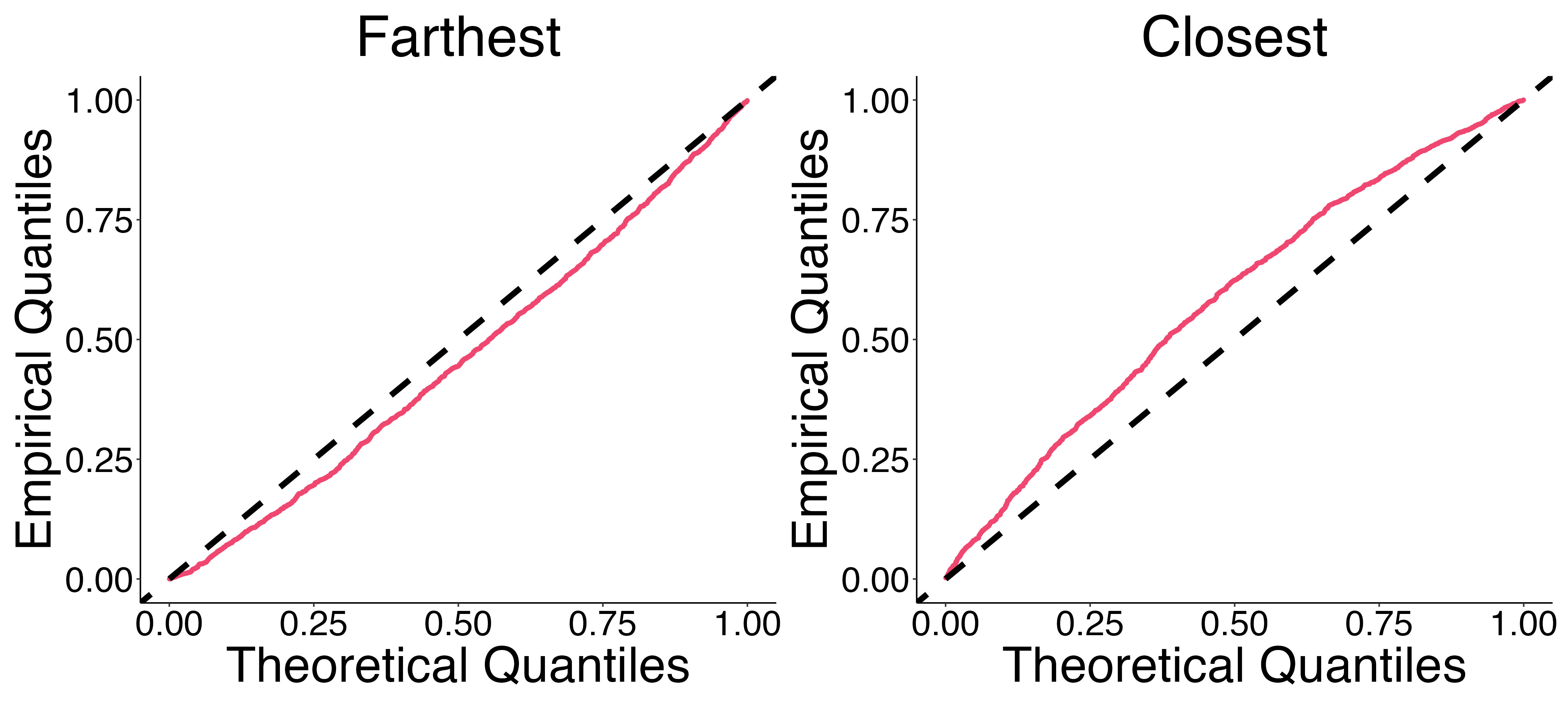}}}
    \caption{QQ plots of samples of $p_{\sigma}$ generated under the null hypothesis in \eqref{global_null} with two different data-dependent choices of $\mathcal{V}.$ }\label{figure_2}
\end{figure}

In this section, we propose a test for the null hypothesis in \eqref{global_null} where $\mathcal{V}$ is chosen in a data-dependent way, for specific forms of data dependence. We adapt the test presented in Section \ref{sec:proposed} to additionally account for the data-dependence in the choice of $\mathcal{V}.$ Henceforth, let $\mathcal{V}(A)$ for a matrix $A\in\mathbb{R}^{n\times q}$ denote the outcome of the procedure that selects the set $\mathcal{V}$ based on $A.$ We consider the same test statistic $T_\sigma$ and define the p-value as $p_{\sigma, J} = 1-Q_{\sigma,J}(T_{\sigma})$, where  $ Q_{\sigma,J}$ is the conditional CDF of $T_{\sigma}$ given 
\begin{align*}
    \mathcal{C}(X),\hspace{1em}\mathcal{V}(X),\hspace{1em}\text{and}\hspace{1em}Z,
\end{align*}
where $Z$ is as defined in \eqref{eq:condsetg}. Theorem \ref{thm:group_dep} below, analogous to Theorem \ref{thm:main_thm}, provides an explicit form of the distribution function $Q_{\sigma,J}.$ 

\begin{theorem}\label{thm:group_dep}
Under the null hypothesis $H_{0,\mathcal{V}}$ in \eqref{global_null}, 
\begin{align*}
Q_{\sigma,J}=F_{\chi_{d}}\!\left(\cdot ;\,  \mathcal{S}_{\sigma, J}  \right),
\end{align*}
 where $d$ is as defined in Section \ref{sec:proposed}, 
and $F_{\chi_{d}}(\cdot;\,  \mathcal{S}_{\sigma, J} )$ 
 represents the CDF of the 
 $\chi_{d}$ distribution truncated to the set 
 $\mathcal{S}_{\sigma, J} = \mathcal{S}_{\sigma} \cap \mathcal{S}_{\sigma, \mathcal{V}},$ where
 \begin{align*}
     \mathcal{S}_{\sigma, \mathcal{V}}\coloneqq   \{\psi \geq 0: \, \mathcal{V}(X)=\mathcal{V}(x_\sigma(\psi))\}. 
 \end{align*}
Furthermore, $\mathbb{P}_{H_{0,\mathcal{V}}}(p_{\sigma,J}\leq \alpha \mid \mathcal{C}(X),\mathcal{V}(X))=\alpha.$
\end{theorem}

The characterization of the set $\mathcal{S}_{\sigma, \mathcal{V}}$ depends on the way in which $\mathcal{V}$ is chosen. We consider three settings where the choice of $\mathcal{V}$ depends on the data $X$ only through the $\ell_2$ distances between cluster centers, after the data has been divided into $K$ clusters. We show that $\mathcal{S}_{\sigma,\mathcal{V}}$ in each of the settings can be written as a set of solutions to a system of inequalities. 

\begin{setting}\label{ex:1}
Let $t_{g}$ denote the $g$-th largest among all between-cluster distances $\|\bar{X}_{\mathcal{C}_k}-\bar{X}_{\mathcal{C}_{k'}} \|=\|X^{\top}\mathbf{v}_{k,k'}\|_2$ for $(k,k')\in \mathcal{V}_{\mathrm{all}},$ where $g$ is pre-specified. If we choose the pairs of clusters whose corresponding distances are among the $g$-th largest, then 
$
    \mathcal{V} = \{(k,k')\in  \mathcal{V}_{\mathrm{all}}: \, \|X^{\top}\mathbf{v}_{k,k'}\|_2 \geqslant t_g \}.
$ The corresponding truncation set takes the form 
 \begin{align*} 
    \mathcal{S}_{\sigma,\mathcal{V}} = \bigcap_{\mathbf{v}\in \mathcal{V} } \bigcap_{\mathbf{v}'\in \mathcal{V}_{\mathrm{all}}\backslash \mathcal{V} } \left\{\psi\geqslant 0:\,  \big\|[x_{\sigma}(\psi)]^{\top}\mathbf{v}   \big\|_2^2> \big\|[x_{\sigma}(\psi)]^{\top}\mathbf{v}'  \big\|_2^2 \right\},  
\end{align*}
where $\mathcal{V}_{\mathrm{all}}\backslash \mathcal{V}$ represents the relative complement of $ \mathcal{V}$ with respect to $\mathcal{V}_{\mathrm{all}}$. If $g=1,$ $\mathcal{V}$ consists of a single index pair corresponding to the pair of clusters whose cluster centers are the farthest apart among all pairs of clusters, which coincides with the setting of the left-hand plot of Figure \ref{figure_2}.
\end{setting}

\begin{setting}\label{ex:2}
Let $t_{g}$ be the $g$-th smallest among all between-cluster distances $\|X^{\top}\mathbf{v}_{k,k'}\|_2$ for $(k,k')\in \mathcal{V}_{\mathrm{all}},$ where $g$ is pre-specified. If we choose the pairs of clusters whose corresponding distances are among the $g$-th smallest, then $\mathcal{V} = \{(k,k')\in \mathcal{V}_{\mathrm{all}}: \, \|X^{\top}\mathbf{v}_{k,k'}\|_2 \leqslant t_g \}.$ The corresponding truncation set takes the form 
\begin{align*} 
    \mathcal{S}_{\sigma,\mathcal{V}} = \bigcap_{\mathbf{v}\in \mathcal{V}_{\mathrm{all}} } \bigcap_{\mathbf{v}'\in \mathcal{V}_{\mathrm{all}}\backslash \mathcal{V} } \left\{\psi\geqslant 0:\,  \big\|[x_{\sigma}(\psi)]^{\top}\mathbf{v}   \big\|_2^2< \big\|[x_{\sigma}(\psi)]^{\top}\mathbf{v}'  \big\|_2^2 \right\}.
\end{align*} 
If $g=1$, $\mathcal{V}$ consists of a single index pair corresponding to the pair of clusters whose cluster centers are the closest to each other among all pairs of clusters, which coincides with the setting of the right-hand plot of Figure \ref{figure_2}.
\end{setting}

\begin{setting}\label{ex:3}
Suppose $t_g>0$ is a pre-specified threshold. If we choose the pairs of clusters where the distances between the cluster centers are no larger than $t_g,$ then 
$
    \mathcal{V} = \{(k,k')\in \mathcal{V}_{\mathrm{all}}: \, \|X^{\top}\mathbf{v}_{k,k'}\|_2 \leqslant t_g \}.
$ The corresponding truncation set takes the form 
\begin{align*}
      \mathcal{S}_{\sigma,\mathcal{V}} = \bigcap_{\mathbf{v}\in \mathcal{V} } \left\{\psi\geqslant 0:\,  \big\|[x_{\sigma}(\psi)]^{\top}\mathbf{v}  \big\|_2^2\leqslant t_g \right\} ~\bigcap~ \bigcap_{\mathbf{v}'\in \mathcal{V}_{\mathrm{all}}\backslash \mathcal{V} } \left\{\psi\geqslant 0:\,  \big\|[x_{\sigma}(\psi)]^{\top}\mathbf{v}'   \big\|_2^2>t_g \right\} .  
\end{align*}
Similarly, if $
    \mathcal{V} = \{(k,k')\in \mathcal{V}: \, \|X^{\top}\mathbf{v}_{k,k'}\|_2 \geqslant t_g \},$
    then $\mathcal{S}_{\sigma,\mathcal{V}}$ takes a similar form, but with the direction of the inequalities reversed.
\end{setting}

Note that the set $\mathcal{S}_{\sigma,\mathcal{V}}$ is defined by quantities of the form $\|[x_{\sigma}(\psi)]^{\top}\mathbf{v}_{k,k'}\|^2$ in each of the settings above. Proposition \ref{prop:group} shows that $\|[x_{\sigma}(\psi)]^{\top}\mathbf{v}\|^2_2$ for any vector ${\mathbf{v}\in\mathbb{R}^n}$ can be written as a quadratic expression in $\psi.$

\begin{proposition}\label{prop:group}
Let $D=\sigma{P_{\mathcal{E}}X}/{\|P_{\mathcal{E}}X\|_F}$ and $E=P_{\mathcal{E}}^\perp X.$ For any vector $\mathbf{v}\in \mathbb{R}^{n}$, 
$$  \|[x_{\sigma}(\psi)]^\top \mathbf{v} \|_2^2=\lambda_{{\mathbf{v}},1}\psi^2+\lambda_{{\mathbf{v}},2}\psi+\lambda_{{\mathbf{v}},3},$$ 
with coefficients 
\begin{align*}
\lambda_{{\mathbf{v}},1}&=\|\mathbf{d}_{\mathbf{v}}\|_2^2,  & \lambda_{{\mathbf{v}}, 2}&=2\langle \mathbf{d}_{\mathbf{v}}, \mathbf{e}_{\mathbf{v}} \rangle,  &   \lambda_{{\mathbf{v}},3}&= \|\mathbf{e}_{\mathbf{v}}\|_2^2,
\end{align*} 
where $\mathbf{d}_{\mathbf{v}}=D^\top \mathbf{v}$ and $\mathbf{e}_{\mathbf{v}}=E^\top \mathbf{v}.$
\end{proposition} 

Thus, $\mathcal{S}_{\sigma,\mathcal{V}}$ in each of the three settings is a set of solutions to a system of quadratic inequalities, as is the case for $\mathcal{S}_{\sigma}$. It follows that $p_{\sigma, J}$ can be computed exactly by Theorem \ref{thm:group_dep} and Proposition \ref{prop:group}. 

\begin{remark}\label{remark_dep} We see in the right-hand plot of Figure \ref{figure_2} that the test based on $p_{\sigma},$ which does not account for the data-dependence in the choice of $\mathcal{V},$ still controls the Type I error under Setting \ref{ex:2}. Intuitively, $H_{0,\mathcal{V}}$ is more likely to happen under the selection event of Setting \ref{ex:2}, so, in this setting, it may be natural to expect the test based on $p_{\sigma}$ to be more conservative than the test based on $p_{\sigma, J}.$ It could then be of interest to ask if the latter has any advantage over the former---in Section \ref{sec:simul_known_dep}, we compare the empirical powers of the two tests to address this question, and we observe that the test based on $p_{\sigma, J}$ leads to a higher empirical power. Another example of such a phenomenon in a different context is discussed in \citet{saha2024inferring}.
\end{remark} 

\section{Unknown variance}\label{sec:var}
In practice, the noise level $\sigma$ in \eqref{data_gen} is often unknown. In Section \ref{sec:unknown_exist},
we give a review of existing methods that test for a pair of clusters under this setting. Then, in Section \ref{sec:ftest}, we present our proposed test that considers multiple pairs of clusters.

\subsection{Existing methods}\label{sec:unknown_exist} 
We start by discussing two existing approaches for testing for a pair of clusters $C_k$ and $C_{k'},$ which are assumed to be chosen independently of the data.

\subsubsection{Variance estimators of \citet{gao2022selective} and \citet{chen2023selective}}
\citet{gao2022selective} and \citet{chen2023selective} provide tests for the null hypothesis in \eqref{null:gao_chen} that are based on plug-in estimators for $\sigma.$ They show that if the estimators satisfy certain properties, then the tests control the Type I error asymptotically in one of the two dimensions. \citet{gao2022selective} provide the estimator
\begin{align}\label{est:sample}
\hat{\sigma}^2_{\textnormal{sample}}=\frac{1}{(n-1)q}\sum_{i=1}^n \norm{X_i-\bar{X}}^2_2,\hspace{1.5em} \text{where} \hspace{1.5em} \bar{X}=\frac{1}{n}\sum_{i=1}^n X_i, 
\end{align} 
which, as they show, asymptotically over-estimates $\sigma$ under certain conditions. \citet{chen2023selective} provide the median-based estimator 
\begin{align}\label{chen_estimator}
\hat{\sigma}^2_{\textnormal{med}}=\frac{1}{\textnormal{Median} (\chi_1^2)} \cdot\mathrm{Median}\left(\{(X_{ij}-\textnormal{Median}(\{X_{sj}\}_{s\in [n]}))^2\}_{i\in[n],j\in[q]}\right),
\end{align} 
where $\textnormal{Median} (\cdot)$ represents the median of a distribution or the median of a set of values depending on the input, and show that a closely related estimator to that of \eqref{chen_estimator} 
is consistent in one of the two dimensions under certain conditions. In terms of the performance of the two estimators, \citet{chen2023selective} numerically illustrate that their proposed test achieves a higher empirical power when applied with $\hat{\sigma}_{\textnormal{med}}.$ 

In Section \ref{sec:simul_unknown_pre}, we show empirically that the proposed test of Section \ref{sec:proposed} controls the Type I error when applied with each of the two estimators.

\subsubsection{Pairwise test of \cite{yun2023selective}}
Alternatively, \cite{yun2023selective} propose a selective inference procedure that avoids the use of a plug-in estimator. They consider a stronger null hypothesis that states 
\begin{align}\label{eq:null_yun}
    H_{0,\{k,k'\}}^*:\ \mu_i = \mu_{i'} \quad \text{for all } i,i'\in \mathcal{C}_k\cup \mathcal{C}_{k'}, 
\end{align}
where $k\neq k'\in[K]$ are pre-specified. The null hypothesis
assumes that the means of all of the observations in the clusters $\mathcal{C}_k$ and $\mathcal{C}_{k'}$ are equal, thereby allowing for the derivation of a test statistic whose distribution is known under the null hypothesis; more discussion on the motivation behind the stronger null hypothesis can be found in \citet[Section 3]{yun2023selective}. They propose the test statistic 
\begin{align}\label{ts_Yun}
T_{k,k'}^*=\frac{\|P_{\mathbf{v}_{k,k'}}X\|^2_F\, /\, q}{\|P_{1,k,k'}X\|^2_F\, /\, d^*_{k,k'}},  
\end{align} 
where $P_{\mathbf{v}_{k,k'}}$ is as defined in Section \ref{sec:proposed}, $d^*_{k,k'}\coloneqq q\cdot (|\mathcal{C}_k|+|\mathcal{C}_{k'}|-2),$ and
\begin{align*}
   P_{1,k,k'} \coloneqq \left(I_{\mathcal{C}_k}-\frac{\mathbf{1}_{\mathcal{C}_k}\mathbf{1}_{\mathcal{C}_k}^\top}{|\mathcal{C}_k|}\right)+\left(I_{\mathcal{C}_{k'}}-\frac{\mathbf{1}_{\mathcal{C}_{k'}}\mathbf{1}_{\mathcal{C}_{k'}}^\top}{|\mathcal{C}_{k'}|}\right).
\end{align*} 
They provide the selective p-value $p^*_{k,k'}=1-Q_{k,k'}^*(T^*_{k,k'}),$ where $Q_{k,k'}^*$ is the conditional CDF of $T_{k,k'}^*$ given 
\begin{align*}
    \mathcal{C}(X)\text{ and }Z^*_{k,k'}\coloneqq \left(\frac{P_{\mathbf{v}_{k,k'}}X}{\|  P_{\mathbf{v}_{k,k'}}X\|_F},\; \frac{  P_{1,k,k'} X}{\|P_{1,k,k'} X\|_F}, \;\|  P_{\mathbf{v}_{k,k'}}X\|_F^2+\|  P_{1,k,k'} X\|_F^2,\; P_{2,k,k'} X \right),
\end{align*}
where $P_{2,k,k'}\coloneqq I_n-P_{\mathbf{v}_{k,k'}}-P_{1,k,k'}.$ They show that the p-value $p_{k,k'}^*,$ conditioned on $\mathcal{C}(X),$ follows the uniform distribution under $H_{0,\{k,k'\}}^*.$ Specifically, they derive that the distribution function $Q^*_{k,k'}$ equals $F_{F_{q,d_{k,k'}^*}}(\cdot; \mathcal{S}_{k,k'}^*),$
 which represents the CDF of the $F_{q,d_{k,k'}^*}$ distribution truncated to the set 
\begin{align*}
    \mathcal{S}^*_{k,k'}\coloneqq \{\psi \geq 0: \mathcal{C}(X)=\mathcal{C}(x_{k,k'}^*(\psi))\}, \notag
\end{align*}
where  
$x^*_{k,k'}: [0, \infty) \rightarrow \mathbb{R}^{n\times q},$
\begin{multline*}
    x^*_{k,k'}(\psi)= \sqrt{\|P_{\mathbf{v}_{k,k'}}X\|_F^2+\|  P_{1,k,k'} X\|_F^2 } \bigg( \sqrt{\frac{\psi}{\psi+r^*_{k,k'}}}  \frac{P_{\mathbf{v}_{k,k'}}X} {\|P_{\mathbf{v}_{k,k'}}X\|_F}+ \\\sqrt{\frac{r^*_{k,k'}}{\psi+ r^*_{k,k'}}}\frac{P_{1,k,k'}X}{\|P_{1,k,k'}X\|_F}  \bigg)
    + P_{2,k,k'}X,
\end{multline*}
where $r^*_{k,k'}=d^*_{k,k'}/q.$  \cite{yun2023selective} derive an explicit characterization of the truncation set $\mathcal{S}^*_{k,k'}$ for the case where $K=2$ clusters are produced by a clustering algorithm that is invariant to the scale and location of the data, which includes $K$-means clustering, and provide an importance sampling algorithm for the case where $K>2.$

\subsection{Proposed test for $H_{0,\mathcal{V}}^*$} \label{sec:ftest}
As discussed in Section \ref{sec:unknown_exist}, the pairwise tests of \citet{gao2022selective} and \citet{chen2023selective} for the null hypothesis in \eqref{null:gao_chen} that are based on plug-in estimators for $\sigma^2$ control the Type I error asymptotically. While the pairwise test of \citet{yun2023selective} controls the Type I error in finite samples, the computation of the associated p-value relies on a sampling algorithm in the presence of $K>2$ clusters. In this section, we provide a test for multiple pairs of clusters that avoid the use of estimators and produce exact p-values, building on the works of \cite{chen2023selective} and \cite{yun2023selective}.

Following the approach of \citet{yun2023selective}, we consider the null hypothesis that states 
\begin{align}\label{eq:null_strong}
    H_{0,\mathcal{V}}^*: \ \mu_i=\mu_{i'} \ \text{for all } i, i' \in \mathcal{C}_k\cup \mathcal{C}_{k'} \ \text{and } (k,k')\in \mathcal{V}.
\end{align}
Note \eqref{eq:null_strong} is stronger than \eqref{global_null} in that the former assumes all of the observations in the clusters of interest have the same mean, and it reduces to the null hypothesis in \eqref{eq:null_yun} if $\mathcal{V}=\{(k,k')\}.$

Let $\mathcal{K} = \{k\in [K]: \text{ there exists } k'\in [K] \text{ such that } (k,k')\in \mathcal{V}\text{ or }(k',k)\in \mathcal{V}\}$ denote the set of indices of all clusters of interest. Analogous to the test statistic in \eqref{ts_Yun}, we propose the test statistic 
\begin{align*}
T^*=\frac{\|P_{\mathcal{E}}X\|^2_F\, /\, d}{\|P_{1}X\|^2_F\, /\, d^*}, 
\end{align*} 
where $d$ and $P_{\mathcal{E}}$ are as defined in Section \ref{sec:proposed}, $d^*=q(\sum_{k\in \mathcal{K}}|\mathcal{C}_k|- |\mathcal{K}| )$, and 
\begin{align*}
   P_1 =\sum_{k\in \mathcal{K}}\left(I_{\mathcal{C}_k}-\frac{\mathbf{1}_{\mathcal{C}_k}\mathbf{1}_{\mathcal{C}_k}^\top}{|\mathcal{C}_k|}\right),
\end{align*} 
which generalizes $P_{1,k,k'}.$ Intuitively, $\|P_{\mathcal{E}}X\|^2_F$ and $\|P_1 X\|^2_F$ reflect the overall between-cluster differences and within-cluster variations, respectively, for the clusters whose indices are in $\mathcal{K}.$ Note that the test statistic $T^*$ is undefined when each of the clusters with indices in $\mathcal{K}$ consists of a single observation, which aligns with the intuition that there is not enough information to estimate the level of within-cluster variations in such a case. Building on the work of \citet{yun2023selective}, we consider the decomposition 
\begin{align}
    X= &~ P_{\mathcal{E}}X + P_1 X +P_2X \notag \\
    =&~ \sqrt{\|P_{\mathcal{E}}X\|_F^2+\|  P_1 X\|_F^2 } \left( \sqrt{\frac{T^*}{T^*+r^*}}  \frac{P_{\mathcal{E}}X} {\|P_{\mathcal{E}}X\|_F} + \sqrt{\frac{r^*}{T^*+ r^*}}\frac{P_1X}{\|P_1X\|_F}  \right) + P_2X, \label{eq:xdecomp2}
\end{align} 
where $r^* =d^*/d$ and $P_2=I_n-P_{\mathcal{E}}-P_1.$ With this decomposition in mind, we derive the selective p-value, separately for the cases where the set $\mathcal{V}$ is pre-specified and chosen in a data-dependent way.

\subsubsection{Pre-specified $\mathcal{V}$}

We first consider the case where $\mathcal{V}$ is pre-specified. In this setting, we define the selective p-value as $p^*=1-Q^*(T^*),$ where $Q^*$ denotes the conditional CDF of $T^*$ given 
\begin{align}\label{eq:condsetg_unknown}
\mathcal{C}(X)\hspace{1em}\text{and}\hspace{1em}Z^*:=\left(\frac{  P_{\mathcal{E}}X}{\|  P_{\mathcal{E}}X\|_F}, \; \frac{  {P_1} X}{\|P_1 X\|_F}, \;\|  P_{\mathcal{E}}X\|_F^2+\|  P_1 X\|_F^2,\;P_2 X\right).
\end{align}
Theorem \ref{thm:unknown_var} gives an explicit form of the distribution function $Q^*.$ 

\begin{theorem}\label{thm:unknown_var}
Under the null hypothesis $H_{0,\mathcal{V}}^*$ in \eqref{eq:null_strong}, 
\begin{align*}
    Q^*=1-F_{F_{d, d^*}}(\cdot ;\mathcal{S}^*), 
\end{align*}
where $F_{F_{d, d^*}}(\cdot;\mathcal{S}^*)$ 
represents the CDF of the $F_{d,d^*}$ distribution truncated to the set 
\begin{align*}
    \mathcal{S}^*\coloneqq \big\{\psi \geq 0: \mathcal{C}(X)=\mathcal{C}(x^* (\psi))\big\},
\end{align*}
where $x^*:[0, \infty)\rightarrow \mathbb{R}^{n\times q},$
\begin{align*}
    x^* (\psi) = \sqrt{\|P_{\mathcal{E}}X\|_F^2+\|  P_1 X\|_F^2 } \left( \sqrt{\frac{\psi}{\psi+r^*}}  \frac{P_{\mathcal{E}}X} {\|P_{\mathcal{E}}X\|_F} + \sqrt{\frac{r^*}{\psi+ r^*}}\frac{P_1X}{\|P_1X\|_F}  \right) + P_2X.
\end{align*}
Furthermore, $\mathbb{P}_{H^*_{0,\mathcal{V}}}(p^*\leq \alpha \mid \mathcal{C}(X))=\alpha.$
\end{theorem}

We next characterize the set $\mathcal{S}^*.$ Analogous to \eqref{eq:charac} and \eqref{eq:charac2}, 
we have 
\begin{align}\label{eq:charac3}
    \mathcal{S}^* =  \bigcap_{j=0}^J\bigcap_{l=1}^K \bigcap_{i=1}^n\mathcal{S}_{ i, l, j} (x^*),
\end{align}
where, by the definitions of the functions $\mathcal{S}_{ i, l, j}$ for $i\in[n],l\in[K],$ and $j\in\{0\}\cup [J]$ given above \eqref{eq:charac}, each 
$\mathcal{S}_{i,l,j}(x^*)$ is defined by inequalities of the form 
\begin{align}\label{eq:form1_unknown}
      \|{ [x^*(\psi)]_i-[x^*(\psi)]_{i'}}\|^2_2\leq \|{[x^*(\psi)]_i-[x^*(\psi)]_{i^{''}}}\|^2_2
 \end{align}   
if $j=0$ and  those of the form 
\begin{align}\label{eq:form2_unknown}
    \|[x^*(\psi)]_i-M_{l}^{(j-1)}(x^*(\psi) ) \|_2^2 \leq \|[x^*(\psi)]_i-M_{l'}^{(j-1)}(x^*(\psi) ) \|_2^2
\end{align} 
if $j\in [J],$ where $i',i''\in[n]$ and $l'\in [K].$ Therefore, to characterize $\mathcal{S}^*$,
it suffices to find the sets of solutions to a system of inequalities of the forms above. Proposition \ref{prop:unknown_cts} expresses these inequalities explicitly with respect to $\psi.$
\begin{proposition}
\label{prop:unknown_cts}
Let $A={P_\mathcal{E}X}/{\|P_\mathcal{E}X\|_F}$, $B= P_1 X/\|P_1X\|_{F}$, and $C=P_2 X/ \sqrt{\|P_\mathcal{E}X\|_F^2+ \|P_1X\|_{F}^2}.$ 
For all $i,i',i^{''}\in [n],$ $l,l'\in[K],$ and $j\in [J],$ 
the inequalities in \eqref{eq:form1_unknown} and \eqref{eq:form2_unknown} are equivalent to 
\begin{align}\label{eq:ineqf}
    f_{ii'}(\psi)\leqslant f_{ii^{''}}(\psi) \hspace{1em}\text{ and } \hspace{1em}{f}_{ijl}(\psi)\leqslant {f}_{ijl'}(\psi), 
\end{align}
 respectively, where $f_{ii'},f_{ijl}:[0,\infty)\rightarrow \mathbb{R},$ 
 \begin{align*}
   f_{ii'}(\psi)= &~\lambda_{ii',1} \psi+\lambda_{ii',2} \sqrt{\psi}+\lambda_{ii',3} \sqrt{\psi}\sqrt{\psi+r^*}+\lambda_{ii',4}  \sqrt{\psi+r^*}+\lambda_{ii',5}, \\
      f_{ijl}(\psi)=&~ \lambda_{ijl,1} \psi+\lambda_{ijl,2} \sqrt{\psi}+\lambda_{ijl,3} \sqrt{\psi}\sqrt{\psi+r^*}+\lambda_{ijl,4}  \sqrt{\psi+r^*}+\lambda_{ijl,5}, 
 \end{align*} 
with coefficients   
\begin{align*} 
 \lambda_{ii',1} = & ~\|\mathbf{a}_{ii'}\|_2^2+  \|\mathbf{c}_{ii'}\|_2^2, & 
 \lambda_{ii',2} =  & ~2 \langle \mathbf{a}_{ii'},  \mathbf{b}_{ii'} \rangle \sqrt{r^*} , &
  \lambda_{ii',3} =  &~  2 \langle  \mathbf{a}_{ii'}, \mathbf{c}_{ii'} \rangle, \\
 \lambda_{ii',4} = & ~ 2\langle \mathbf{b}_{ii'}, \mathbf{c}_{ii'}\rangle \sqrt{r^*}, & 
 \lambda_{ii',5}=   & ~  (\|\mathbf{b}_{ii'}\|_2^2+\|\mathbf{c}_{ii'}\|_2^2)r^*,  &  \\
    {\lambda}_{ijl,1} =  &  ~\|\mathbf{a}_{ijl}\|_2^2+  \|\mathbf{c}_{ijl}\|_2^2, & 
     {\lambda}_{ijl,2} =  & ~ 2\langle \mathbf{a}_{ijl}, \mathbf{b}_{ijl}\rangle \sqrt{r^*}, & 
     {\lambda}_{ijl,3} =  & ~ 2\langle \mathbf{a}_{ijl}, \mathbf{c}_{ijl}\rangle,  \\
       {\lambda}_{ijl,4} =  &~  2\langle \mathbf{b}_{ijl}, \mathbf{c}_{ijl}\rangle \sqrt{r^*} , & 
           {\lambda}_{ijl,5} =  & ~ (\|\mathbf{b}_{ijl}\|_2^2 + \|\mathbf{c}_{ijl} \|_2^2)r^*,  &  & 
\end{align*}
where $\mathbf{a}_{ii'}=A_i-A_{i'}$, $\mathbf{b}_{ii'}= B_i-B_{i'}$, $\mathbf{c}_{ii'}= C_i-C_{i'}$,
$\mathbf{a}_{ijl}=A_i-M_{l}^{(j-1)}(A)$,
$\mathbf{b}_{ijl}=B_i-M_{l}^{(j-1)}(B)$, and 
$\mathbf{c}_{ijl}=C_i-M_{l}^{(j-1)}(C)$. 
\end{proposition} 
It follows by Theorem \ref{thm:unknown_var} and Proposition \ref{prop:unknown_cts} that $p^*$ can be computed exactly. 
\begin{remark}\label{rem:unknown_comp}
Note that for any $i,i',i^{''}\in [n],$ $l,l'\in[K],$ and $j\in [J],$ $f_{ii'},f_{ii''},f_{ijl},$ and $f_{ijl'}$ are continuous functions, so $f_{ii'}-f_{ii''}$ and $f_{ijl}-f_{ijl'}$ are also continuous. Thus, Proposition  \ref{prop:unknown_cts} shows that the set of solutions to each of \eqref{eq:form1_unknown} and \eqref{eq:form2_unknown} is equivalent to the sub-level set of a continuous function at 0, which can be found by solving for the roots of the function and checking the signs of its values on each interval partitioned by the roots; further details on the implementation can be found in Appendix \ref{sec:imple_5.2}. 
\end{remark}

\subsubsection{Data-dependent choice of $\mathcal{V}$}

We next consider the setting where $\mathcal{V}$ is chosen in a data-dependent way and develop a test that is analogous to that of Section \ref{sec:group_dep}. In this setting, we define the selective p-value as 
$p_{J}^* = Q_{\mathcal{G}_J^*}^*(T^*)$, where  $Q^*_{\mathcal{G}_J^*}$ is the conditional CDF of $T^*$ given
\begin{align*}
\mathcal{C}(X),\hspace{1em}\mathcal{V}(X),\hspace{1em}\text{and}\hspace{1em} Z^*,
\end{align*}
where $Z^*$ is as defined in \eqref{eq:condsetg_unknown}. Theorem \ref{thm:unknownvar_datav} gives an explicit form of the distribution function $Q_{J}^*.$ 
\begin{theorem} \label{thm:unknownvar_datav}
Under the null hypothesis $H_{0,\mathcal{V}}^*$ in \eqref{eq:null_strong},
\begin{align*}
Q_{J}^*=F_{F_{d, d^*}}(\cdot;\mathcal{S}^*_J),
\end{align*}
where $F_{F_{d, d^*}}(\cdot;\mathcal{S}^*_J)$ 
represents the CDF of the $F_{d, d^*}$ distribution truncated to the set $$\mathcal{S}^*_J\coloneqq \mathcal{S}^*\cap \mathcal{S}_{\mathcal{V}}^*,$$ where  
$$  \mathcal{S}_{\mathcal{V}}^*:=   \{\psi \geq 0: \mathcal{C}(X)=\mathcal{C}(x^*(\psi))\}.$$
Furthermore, $\mathbb{P}_{H^*_{0,\mathcal{V}}}(p^*_J\leq \alpha \mid \mathcal{C}(X),\mathcal{V}(X))=\alpha.$
\end{theorem}

The characterization of the set $\mathcal{S}_{ \mathcal{V}}^*$ depends on the specific way in which $\mathcal{V}$ is chosen. We again consider Settings \ref{ex:1} through \ref{ex:3} of Section \ref{sec:group_dep}, with $x_\sigma$ replaced by $x^*.$ Then, the set $\mathcal{S}_{\mathcal{V}}^*$ in Settings \ref{ex:1} and \ref{ex:2} is defined by inequalities of the form 
\begin{align}\label{V_dep_unknown_1}
\|[x^*(\psi)]^\top \mathbf{v}\|^2_2<\|[x^*(\psi)]^\top \mathbf{v}'\|^2_2, 
\end{align}
where $\mathbf{v},\mathbf{v}'\in \mathbb{R}^n,$
and $\mathcal{S}_{\mathcal{V}}^*$ in Setting \ref{ex:3} is defined by inequalities of the form 
\begin{align}\label{V_dep_unknown_2}
\|[x^*(\psi)]^\top \mathbf{v}\|^2_2<t,
\end{align}
where $\mathbf{v}\in \mathbb{R}^n$ and $t>0.$ Proposition \ref{V_dep_unknown} expresses these inequalities explicitly with respect to $\psi.$

\begin{proposition}\label{V_dep_unknown}
Let $A,$ $B,$ and $C$ be defined as in Proposition \ref{prop:unknown_cts}. For any two vectors $\mathbf{v},\mathbf{v}'\in \mathbb{R}^n$ and for any $t>0,$ the inequalities in \eqref{V_dep_unknown_1} and \eqref{V_dep_unknown_2} are equivalent to 
\begin{align}\label{eq:prop8}
    h_\mathbf{v}(\psi)\leqslant h_{\mathbf{v}'}(\psi)\hspace{1em}\text{ and } \hspace{1em}h_\mathbf{v}(\psi)\leqslant h(\psi), 
\end{align}
respectively, 
where $h_\mathbf{v},h:[0,\infty)\rightarrow \mathbb{R},$
\begin{align*}
    h_\mathbf{v}(\psi)&=\lambda_{\mathbf{v},1}\psi + \lambda_{\mathbf{v},2}\sqrt{\psi}+\lambda_{\mathbf{v},3}\sqrt{\psi}\sqrt{\psi+r^*}+\lambda_{\mathbf{v},4}\sqrt{\psi+r^*}+\lambda_{\mathbf{v},5},\\
    h(\psi)&=\gamma_1\psi + \gamma_2,
\end{align*}
with coefficients
\begin{align*}
    \lambda_{\mathbf{v},1} = & ~\|\mathbf{a}_{\mathbf{v}}\|_2^2+  \|\mathbf{c}_{\mathbf{v}}\|_2^2, & 
 \lambda_{\mathbf{v},2} =  & ~2 \langle \mathbf{a}_{\mathbf{v}},  \mathbf{b}_{\mathbf{v}} \rangle \sqrt{r^*} , &
  \lambda_{\mathbf{v},3} =  &~  2 \langle  \mathbf{a}_{\mathbf{v}}, \mathbf{c}_{\mathbf{v}} \rangle, \\
 \lambda_{\mathbf{v},4} = & ~ 2\langle \mathbf{b}_{\mathbf{v}}, \mathbf{c}_{\mathbf{v}}\rangle \sqrt{r^*}, & 
 \lambda_{\mathbf{v},5}=   & ~  (\|\mathbf{b}_{\mathbf{v}}\|_2^2+\|\mathbf{c}_{\mathbf{v}}\|_2^2)r^*,  &  \\
    {\gamma}_{1} =  &  ~t/(\|P_{\mathcal{E}}X\|_F^2+\|P_{1}X\|_F^2), & 
     {\gamma}_{2} =  & ~ r^*t/(\|P_{\mathcal{E}}X\|_F^2+\|P_{1}X\|_F^2),
\end{align*}
where $\mathbf{a}_{\mathbf{v}}=A^\top \mathbf{v}$, $\mathbf{b}_{\mathbf{v}}= B^\top \mathbf{v}$, and $\mathbf{c}_{\mathbf{v}}=C^\top \mathbf{v}.$
\end{proposition}
The solution sets to the inequalities in \eqref{eq:prop8} can be computed by a method analogous to that discussed in Remark \ref{rem:unknown_comp}. It then follows by Theorem \ref{thm:unknownvar_datav} and Proposition \ref{V_dep_unknown} that $p_{J}^*$ can be computed exactly. 

\section{Simulations}\label{sec:simul}
We now present a simulation study of the proposed tests, exploring their Type I error control and empirical powers. 
Sections \ref{sec:simul_known} and \ref{sec:simul_unknown} provide empirical results for the settings where the noise level $\sigma$ is known and unknown, respectively. Sections \ref{sec:simul_known_pre} and \ref{sec:simul_unknown_pre} consider the case where $\mathcal{V}$ is pre-specified, and Sections 
\ref{sec:simul_known_dep} and \ref{sec:simul_unknown_dep} address the case where the choice of $\mathcal{V}$ is data-dependent. Specifically, Sections \ref{sec:simul_known_pre}, \ref{sec:simul_unknown_pre}, \ref{sec:simul_known_dep}, and \ref{sec:simul_unknown_dep} demonstrate the numerical performance of the  proposed tests based on the p-values $p_\sigma,$ $p_{\sigma,J},$ $p^*,$ and $p^*_J,$ respectively. 
In relevant contexts, we make comparisons with the baseline testing procedure of Section \ref{sec:bonfcorr} as well as the tests based on the plug-in estimators $\hat{\sigma}_{\textrm{sample}}$ and $\hat{\sigma}_{\textrm{med}}.$ For the simplicity of writing, we henceforth denote the tests based on $p_{\hat{\sigma}},$  $p_{\hat{\sigma},J},$ $p^*,$ $p^*_J,$ and $p_{\hat{\sigma},\mathrm{Bon}}$ as  $\phi_{\hat{\sigma}},$  $\phi_{\hat{\sigma},J},$ $\phi^*,$ $\phi^*_J,$ and $\phi_{\hat{\sigma},\mathrm{Bon}},$ respectively, where $\hat{\sigma}\in\{\sigma, \hat{\sigma}_{\mathrm{sample}},\hat{\sigma}_{\mathrm{med}}\}.$

Before presenting the empirical results, we first discuss the simulation settings that are common to all experiments in this section. For each test,  
we sample the associated p-value as follows. We draw each data  from the distribution in \eqref{data_gen} with $n=120$ and $\sigma=1,$ and run $K$-means clustering to divide the observations into $K$ clusters, where the values for $\mu,$ $q,$ 
and $K$ are specified in the respective sections. 
We then compute the selective p-value. We discuss below the ways in which the Type I error control and the empirical power of a test are illustrated in each section. 

\paragraph{Type I error} For experiments under the null hypothesis, we set $\mu$ in \eqref{data_gen} as $0_{n\times q}.$ To explore the Type I error control of a test, we sample the associated p-value 1500 times. We then present a QQ plot that compares the empirical quantiles of the 1500 p-values against the theoretical quantiles of the uniform distribution supported on $[0,1].$

\paragraph{Empirical power} For experiments under the alternative hypothesis, we consider the setting where there are $K_0\in\mathbb{N}$ different values of $\mu_i$s, and we henceforth refer to the sets of indices of the observations partitioned by their means as true clusters. Additionally, we set $K=K_0,$ i.e., the number of clusters produced by $K$-means clustering is equal to the number of true clusters. We consider two different specifications of $\mu$ in \eqref{data_gen} for studying the empirical power of a test. For any $\delta\geq0$ and each $k\in \{0\}\cup [K-1],$ define
\begin{itemize}
    \item (Horizontal)  $\mu_H^{(k)}=(k\cdot \delta, \underbrace{0,...,0}_{q-1})^\top$ and
    \item ($K$-gon)
    $\mu_K^{(k)}=(\delta \cdot \mathrm{cos}(2\pi k/K),\delta\cdot \mathrm{sin}(2\pi k/K),\underbrace{0,...,0}_{q-2})^\top.$ 
\end{itemize}
Note that if $q=2,$ $\mu_H^{(k)}$ for $k\in \{0\}\cup [K-1]$ are arranged horizontally, and $\mu_K^{(k)}$ for $k\in \{0\}\cup [K-1]$ form the $K$ vertices of a polygon, hence the labels. We set each of the $K$ clusters to contain an equal number of observations; specifically, there are $n/K$ $\mu_i$s that are set to $\mu_H^{(k)}$ for each $k\in \{0\}\cup [K-1]$ for the Horizontal case, and similarly for the $K$-gon case---we only consider values of $K$ that divide $n.$ Visualizations of data generated from the two specifications of $\mu$ along with the clustering outcomes for the case where $\delta=6,$ $q=2,$ and $K=3$ are shown in Figure \ref{figure_4}---in visualizations of clustering outcomes in this figure and in the rest of this paper, true clusters are differentiated by shape and  (estimated) clusters by color. 

\begin{figure}[ht]
    \centering
    {{\includegraphics[width=0.6\textwidth]{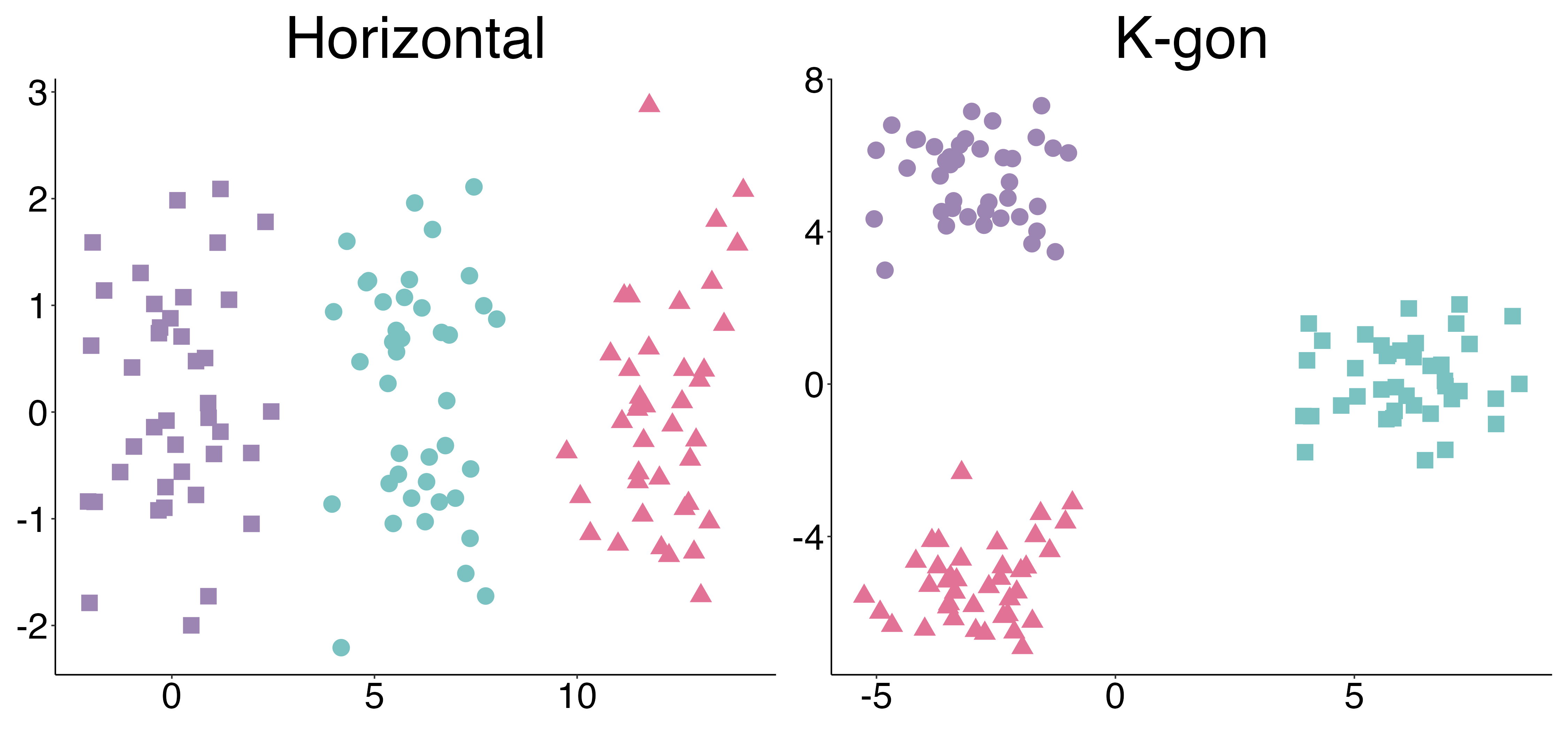}}}
    \caption{Visualizations of simulated data consistent with the alternative hypothesis for $\delta=6.$}\label{figure_4}
\end{figure} 

For each signal strength $\delta\in \{0,0.5,1,...,5.5,6\}$ and $\mu$ set according to either of the two specifications above, we 
sample 1500 p-values of the test of interest. Then, for each value of $\delta,$ we compute the empirical power of the test as the proportion of null hypotheses that are rejected (at the significance level of 0.05) among those that are false, and present a plot of empirical powers against $\delta.$ 

\subsection{Testing for $H_{0,\mathcal{V}}$ for known $\sigma$}\label{sec:simul_known}
We start by presenting the Type I error control and empirical powers of the proposed tests for the null hypothesis in \eqref{global_null}, where $\sigma$ is assumed to be known. The performance of the tests $\phi_{\sigma}$ and $\phi_{\sigma,J}$ are illustrated in Sections \ref{sec:simul_known_pre} and \ref{sec:simul_known_dep}, respectively. Unless specified otherwise, we set $q=2$ throughout the experiments in this section. 

\subsubsection{Pre-specified $\mathcal{V}$} \label{sec:simul_known_pre}
For the setting where $\mathcal{V}$ is pre-specified, we test for the global null hypothesis $H_{0,\mathcal{V}}:\bar{\mu}_1=\cdots\bar{\mu}_K,$ where we consider three different choices of  $\mathcal{V}$:
\begin{itemize}
    \item $\mathcal{V}_1=\{(k,k'):k<k'\in[K]\},$
    \item $\mathcal{V}_2=\{(k,k+1):k\in[K-1]\},$ and
    \item $\mathcal{V}_3=\{(1,k):k\in[K]\backslash\{1\}\}.$
\end{itemize}
As the corresponding spaces $\mathrm{span}\{\mathbf{v}_{k,k'}:(k,k')\in \mathcal{V}_m\}$ for $m\in[3]$ are identical, the proposed test $\phi_{\sigma}$ remains invariant to the different choices of $\mathcal{V}$, while the baseline testing procedure $\phi_{\sigma, \mathrm{Bon}}$ does not; see Remark \ref{rem:invariance} for further discussion. 

\paragraph{Type I error} Figure \ref{figure_3} presents QQ plots of 
samples of $p_{\sigma}$ generated under the null hypothesis. The plots show that $p_{\sigma}$ is uniformly distributed, which is consistent with Theorem \ref{thm:main_thm}. Analogous plots for $p_{\sigma,\mathrm{Bon}}$ for the setting where $q=2$ and $\mathcal{V}=\mathcal{V}_1$ was presented in Figure \ref{figure_1}.

\begin{figure}[ht]
    \centering
    {{\includegraphics[width=0.8\textwidth]{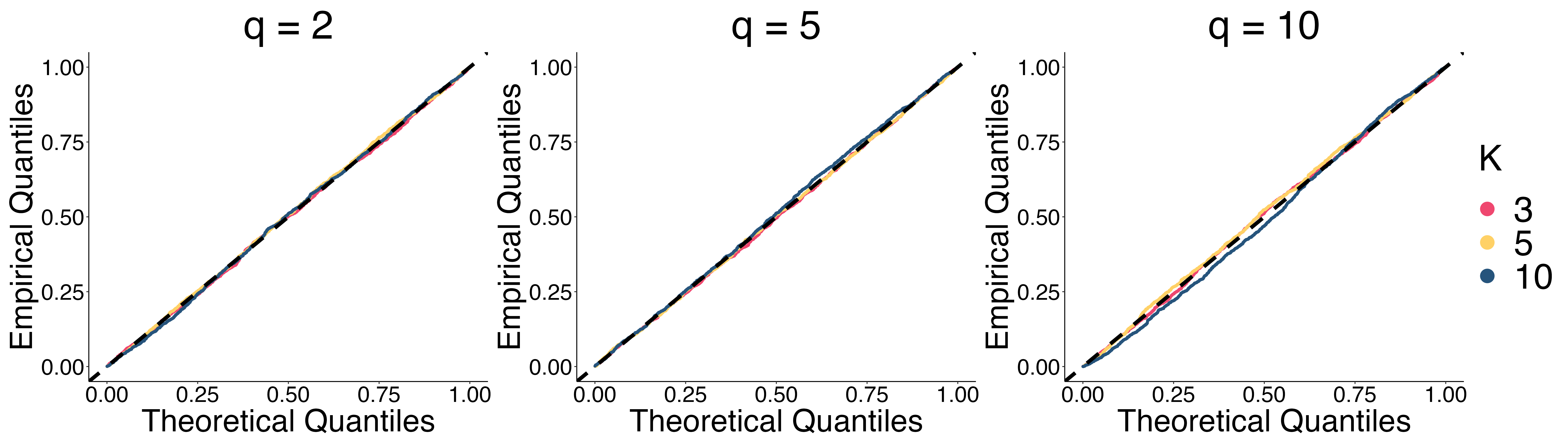}}}
    \caption{QQ plots of 
samples of $p_{\sigma}$ generated under the null hypothesis in \eqref{global_null} with a pre-specified $\mathcal{V}.$}\label{figure_3}
\end{figure}

\paragraph{Empirical power} 
Figure \ref{figure_5} presents the empirical powers of $\phi_{\sigma}$ and $\phi_{\sigma,\mathrm{Bon}},$ where the performance of the latter is compared for the three different choices of $\mathcal{V}.$ The figure illustrates that the relative performances of the tests vary depending on the signal strength: the proposed test $\phi_\sigma$ achieves a higher empirical power than $\phi_{\sigma,\mathrm{Bon}}$ in the presence of weak signals, while the opposite is true in the presence of strong signals. 

We speculate that the intuition for the higher empirical power of $\phi_\sigma$ in the presence of weak signals lies in the construction of the test statistic that combines the signals across all pairs of clusters of interest by considering the span of all of the vectors  $v_{k,k'}$s for $(k,k')$ in $\mathcal{V}.$ On the other hand, the presence of at least one strong signal is conducive to $\phi_{\sigma,\mathrm{Bon}},$ whose test statistic is associated with only the strongest signal, potentially accounting for its superior performance for large values of $\delta.$

\begin{figure}[ht]
    \centering
    {{\includegraphics[width=0.9\textwidth]{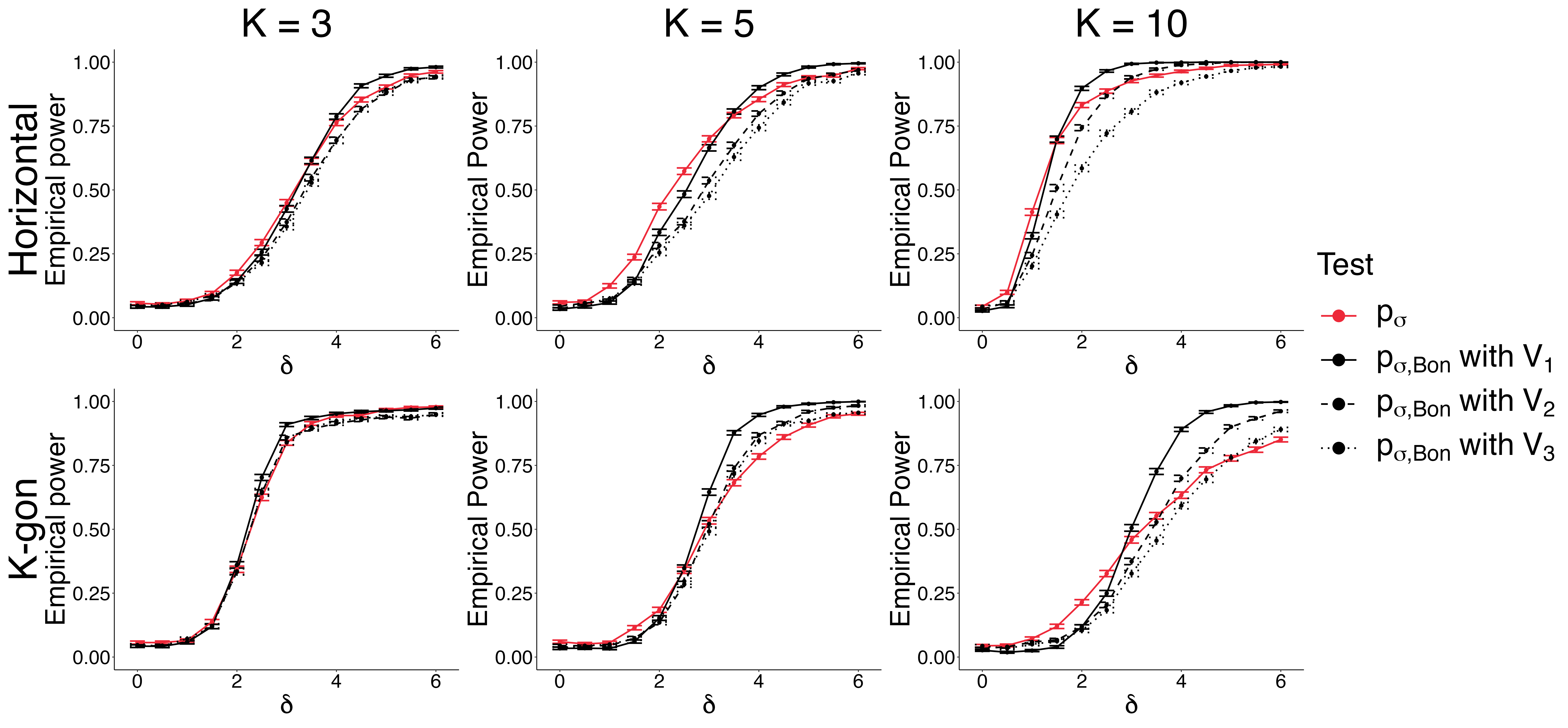}}}
    \caption{Plots of empirical powers of the proposed test based on $p_{\sigma}$ and the baseline testing procedure based on $p_{\sigma,\mathrm{Bon}}$ for the null hypothesis in \eqref{global_null} with pre-specified $\mathcal{V}$s.}\label{figure_5}
\end{figure}

\subsubsection{Data-dependent choice of $\mathcal{V}$} \label{sec:simul_known_dep}

We next consider the setting where $\mathcal{V}$ is chosen in a data-dependent way. 

\paragraph{Type I error} The first and second columns of Figure \ref{figure_6} present QQ plots of empirical p-values obtained from the tests $\phi_{\sigma,J}$ and $\phi_\sigma,$ respectively, under the null hypothesis when $K=20$ clusters are produced by $K$-means clustering. Consistent with Theorem \ref{thm:group_dep}, the figure shows that $p_{\sigma,J}$ is uniformly distributed under both Settings \ref{ex:1} and \ref{ex:2}. On the other hand, $\phi_\sigma$ fails to control the Type I error under Setting \ref{ex:1}, where the test becomes more confident as $g,$ the number of pairs of clusters chosen, decreases. The test controls the Type I error under Setting \ref{ex:2} but becomes more conservative as $g$ decreases. As discussed in Remark \ref{remark_dep}, this observation aligns with the intuition that the null hypothesis is less likely to happen under Setting \ref{ex:1} and more likely to happen under Setting \ref{ex:2}. Note that the QQ plot of a sample of $p_\sigma$ for the case where $g=1$ was also presented in Figure \ref{figure_2}. 

\begin{figure}[ht]
    \centering
    {{\includegraphics[width=0.6\textwidth]{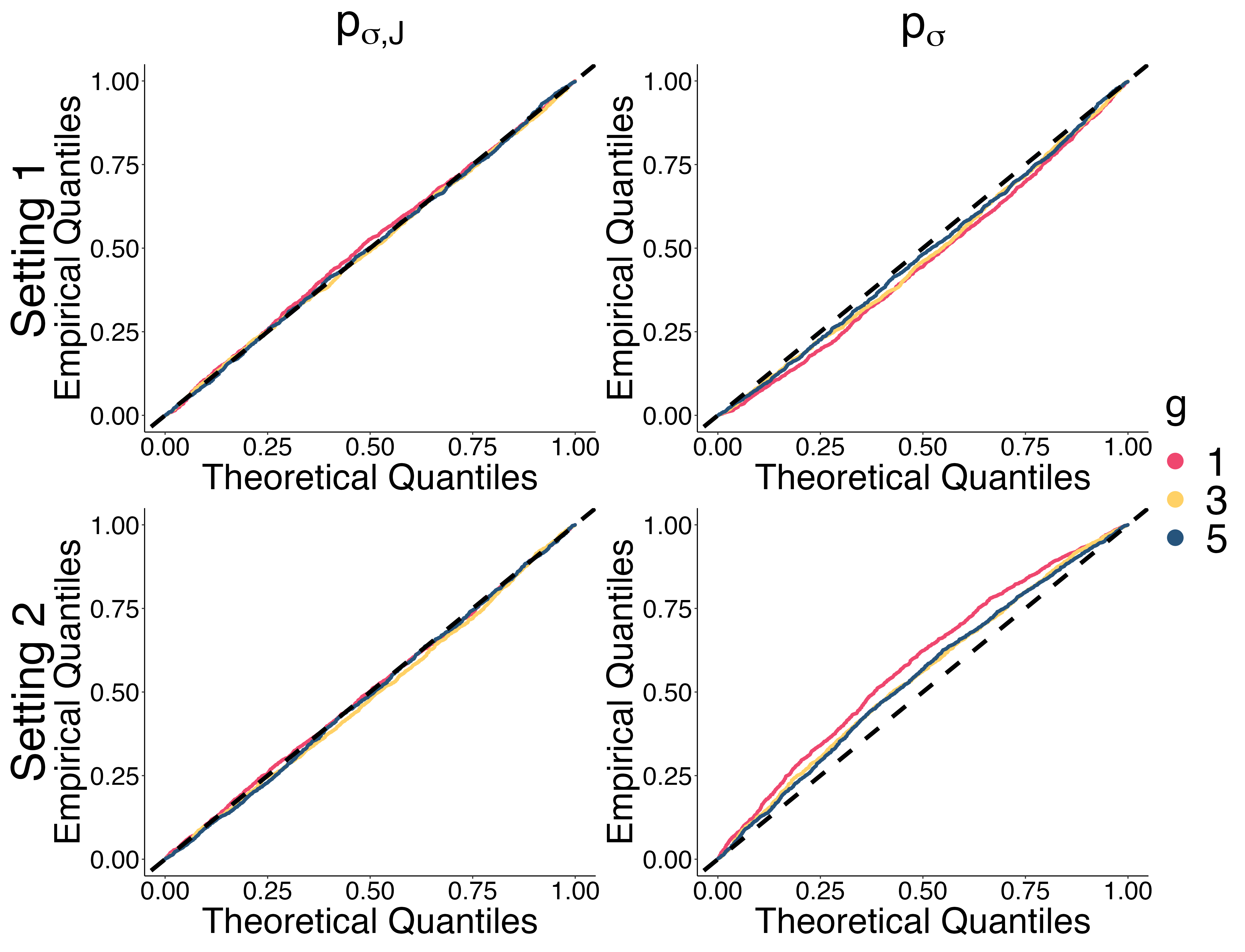}}}
    \caption{QQ plots of a sample of $p_{\sigma,J}$ (left) and a sample of $p_{\sigma}$ (right) generated under the null hypothesis in \eqref{global_null} with data-dependent choices of $\mathcal{V},$ for $K=20.$}\label{figure_6}
\end{figure}

Figure \ref{figure_7} presents analogous results for different values of $K$ when $g=3$ pairs of clusters are chosen to be tested. The figure shows that as $K$ increases, $p_{\sigma}$ becomes more confident under Setting \ref{ex:1} and more conservative under Setting \ref{ex:2}. Both Figures \ref{figure_6} and \ref{figure_7} illustrate that the effects of not accounting for the data-dependence in the choice of $\mathcal{V}$ are not significant, especially when $K$ is small or $g$ is large. 

\begin{figure}[ht]
    \centering
    {{\includegraphics[width=0.6\textwidth]{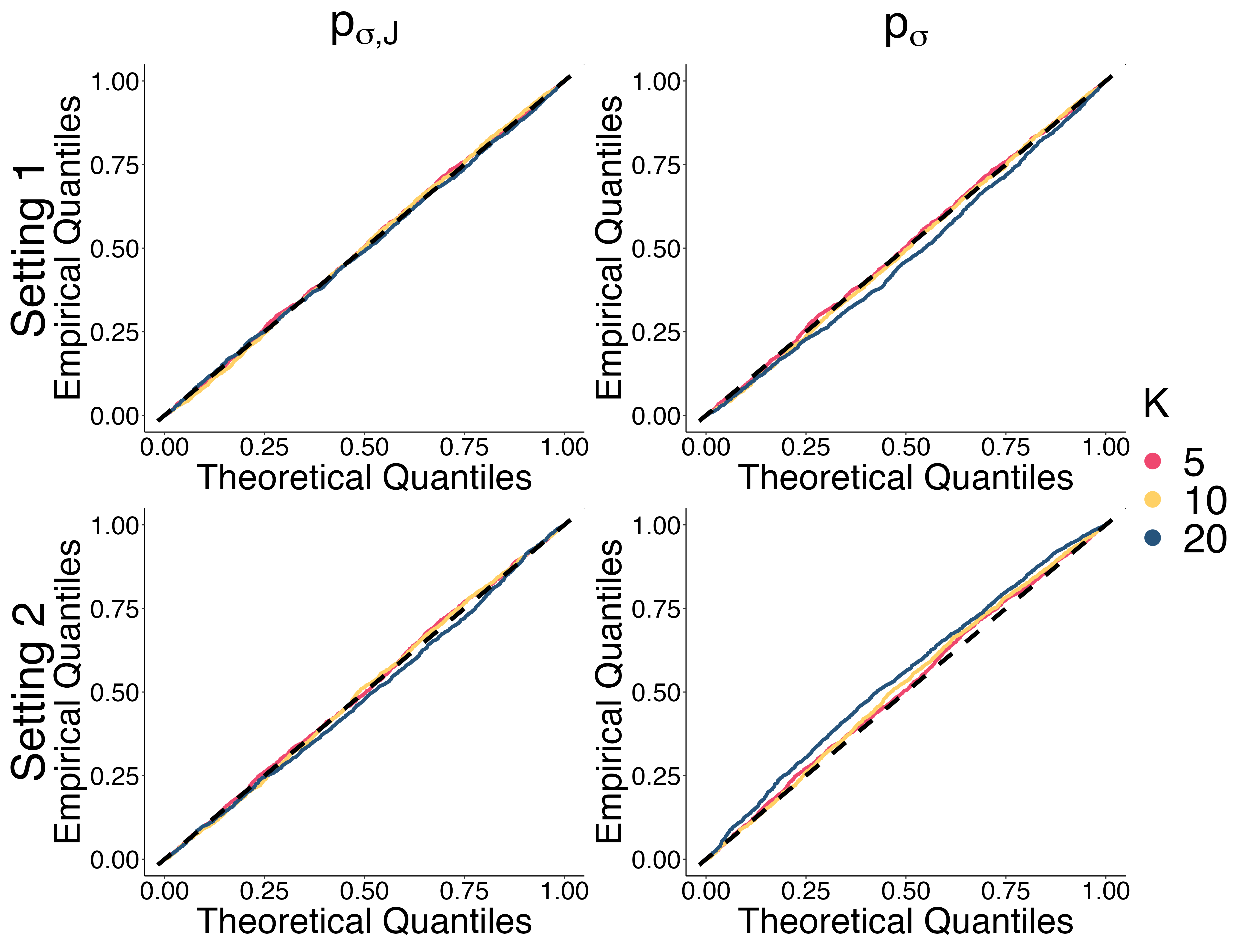}}}
    \caption{QQ plots of a sample of $p_{\sigma,J}$ (left) and a sample of $p_{\sigma}$ (right) generated under the null hypothesis in \eqref{global_null} with data-dependent choices of $\mathcal{V},$ for $g=3.$}\label{figure_7}
\end{figure}
 
\paragraph{Empirical power} Figure \ref{figure_8} presents empirical powers of $\phi_{\sigma,J}$ and $\phi_{\sigma}$ under Settings \ref{ex:1} and \ref{ex:2} with $K=20$ and $g=3.$ As suggested by the Type I error control of $\phi_{\sigma},$ the empirical power of $\phi_{\sigma, J}$ tends to be lower than that of $\phi_\sigma$ under Setting \ref{ex:1}, and the opposite is true under Setting \ref{ex:2}. 

\begin{figure}[ht]
    \centering
{{\includegraphics[width=0.6\textwidth]{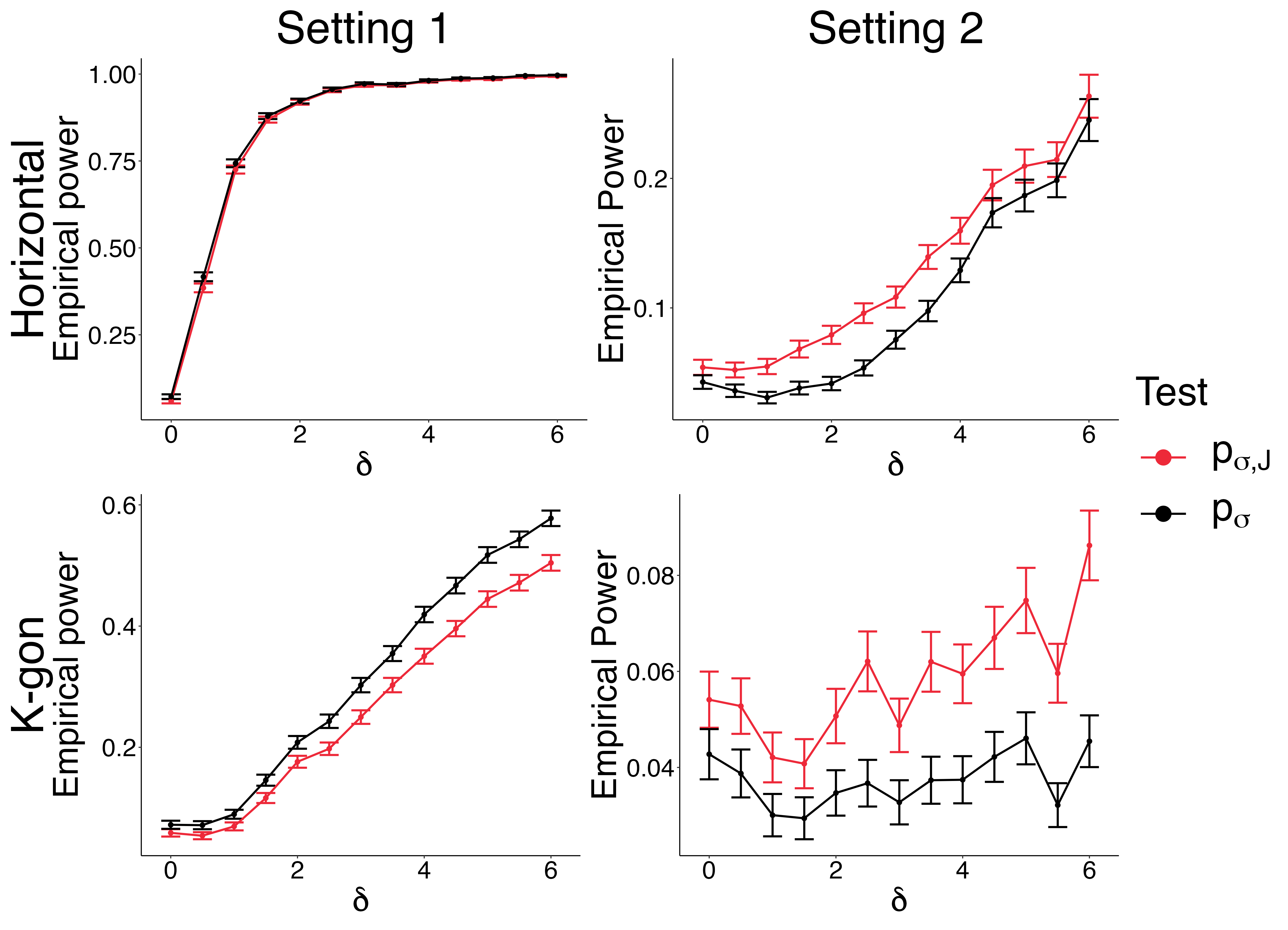}}}
    \caption{Plots of empirical powers of the tests based on $p_{\sigma,J}$ and $p_{\sigma}$ for the null hypothesis in \eqref{global_null} with data-dependent choices of $\mathcal{V},$ for $K=20$ and $g=3.$}\label{figure_8}
\end{figure}

\subsection{Testing for $H_{0,\mathcal{V}}^*$ for unknown $\sigma$}\label{sec:simul_unknown}

We next consider the setting where $\sigma$ is unknown and test for the null hypothesis in \eqref{eq:null_strong}. Unless specified otherwise, we set $q=20$ throughout the experiments in this section. The choice of $q=20$ that is relatively large is due to a computational limitation of $\phi^*,$ which we discuss below.

\paragraph{Computational limitation of $\phi^*$} Since the selective p-value $p^*$ is the right-tail probability of a truncated $F$ distribution, its computation involves taking ratios of probabilities associated with the $F$ distribution, which may be very small in some cases. As a result, computing the probabilities using a built-in function in \verb|R| and then taking ratios can lead to numerical inaccuracies. To address this computational issue, we take an alternative approach for computing the p-value. Specifically, following the approach of \citet[Appendix A.3]{yun2023selective}, we first use the procedure of \citet{li2002approximation} to approximate the truncation set $\mathcal{S}^*,$ which is associated with the $F_{n_1,n_2}$ distribution (where $n_1,n_2>0$) with a set associated with the $\chi^2_{n_1}$ distribution, which we denote as $\mathcal{S}^*_{\chi^2_{n_1}},$ so that $p_J^*\approx 1-F_{\chi_{n_1}^2}(T^*;\mathcal{S}^*_{\chi^2_{n_1}}),$ where $F_{\chi_{n_1}^2}(\cdot;\mathcal{S}^*_{\chi^2_{n_1}})$ represents the CDF of the $\chi_{n_1}^2$ distribution truncated to the set $\mathcal{S}^*_{\chi^2_{n_1}}.$ We then compute the approximation of $p_J^*$ using the function \verb|TChisqRatioApprox| of the package \verb|KmeansInference| of \citet{chen2023selective}.

We speculate that the approximation of the $F$ distribution with the $\chi^2$ distribution is less accurate for values that are far in the right tail of the $F$ distribution. We have observed that the realizations of the set $\mathcal{S}^*$ tend to decrease with $q$ and increase with $n$ and the signal strength $\delta.$ Accordingly, we have noticed that the accuracy of the approximation, as measured by the empirical Type I error control of $\phi^*,$
tends to increase with $q$ and decrease with $n$ and $\delta,$ where the realizations of $p^*$ tend to be underestimated for small values of $q.$ Thus, to ensure that the results presented in this section are accurate both under the null hypothesis and the alternative hypothesis, we set $q=20$ unless specified otherwise. 

\subsubsection{Pre-specified $\mathcal{V}$}\label{sec:simul_unknown_pre}

We test for the global null hypothesis $H_{\mathcal{V}}^*:\mu_i=\mu_{i'}$ for all $i,i'\in [n],$ where we set $\mathcal{V}=\mathcal{V}_1.$
 
\paragraph{Type I error} Figure \ref{figure_9} presents QQ plots of samples of p-values obtained from the test $\phi^*$ under the null hypothesis, which are consistent with Theorem \ref{thm:unknown_var}. Figure \ref{figure_10} presents analogous plots for $\phi_{\hat{\sigma}}$ for $\hat{\sigma}\in \{\hat{\sigma}_{\textrm{sample}},\hat{\sigma}_{\textrm{med}}\},$ illustrating that each test controls the Type I error empirically.

\begin{figure}[ht]
    \centering
    {{\includegraphics[width=0.8\textwidth]{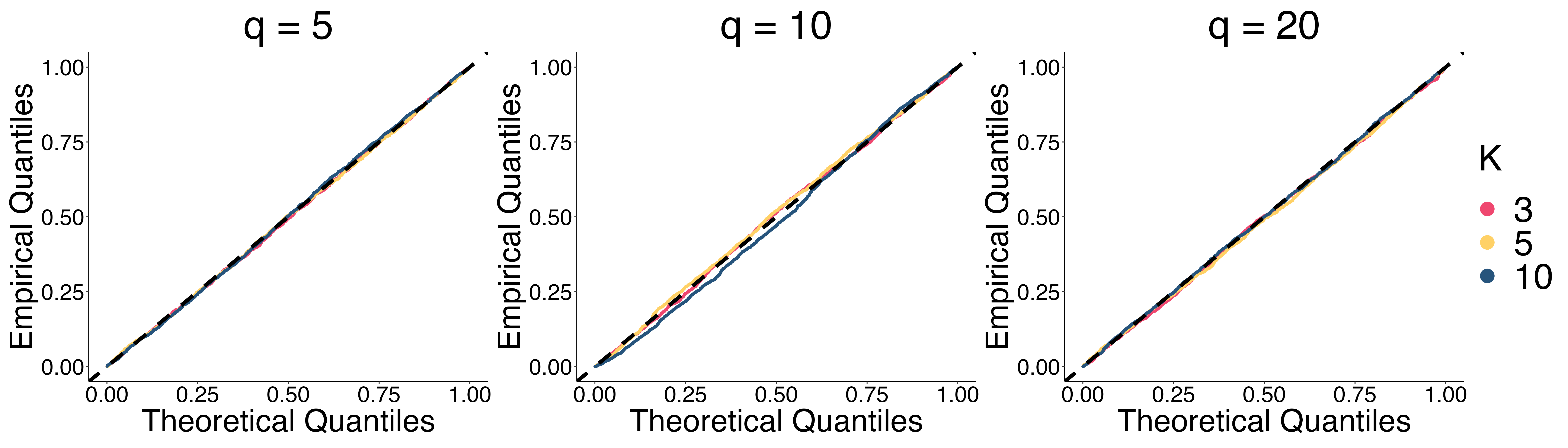}}}
    \caption{QQ plots of samples of $p^*$ generated under the null hypothesis in \eqref{eq:null_strong} with a pre-specified $\mathcal{V}.$}\label{figure_9}
\end{figure}

\begin{figure}[ht]
    \centering
    {{\includegraphics[width=0.6\textwidth]{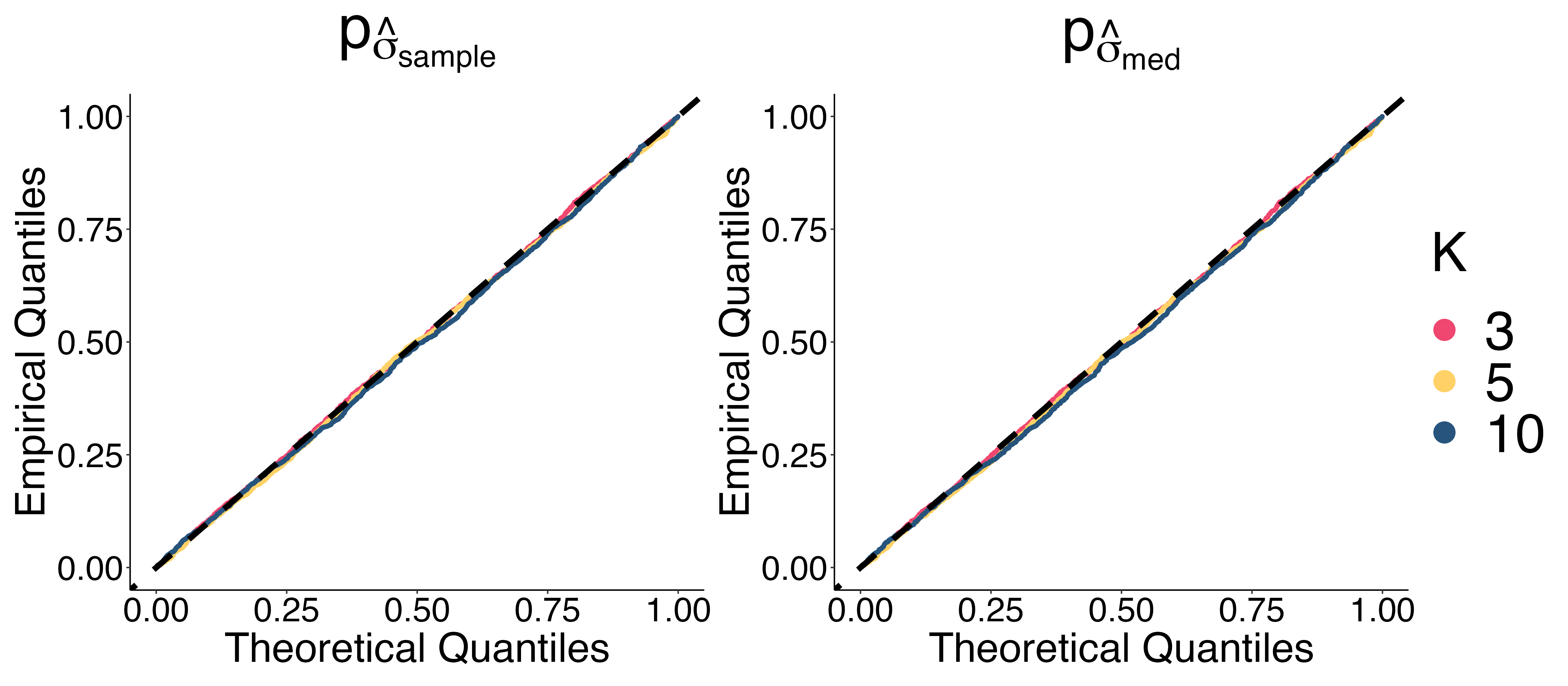}}}
    \caption{QQ plots of a sample of $p_{\hat{\sigma}}$ from each of $\hat{\sigma}\in\{\hat{\sigma}_{\mathrm{sample}},\hat{\sigma}_{\mathrm{med}}\}$ generated under the null hypothesis in \eqref{eq:null_strong} with a pre-specified $\mathcal{V}.$}\label{figure_10}
\end{figure}

\paragraph{Empirical power} Figure \ref{figure_11} presents empirical powers of the tests $\phi^*,$ $\phi_{\hat{\sigma}},$ and $\phi_{\hat{\sigma}, \mathrm{Bon}}$ for $\hat{\sigma}\in\{\hat{\sigma}_{\mathrm{sample}},\hat{\sigma}_{\mathrm{Bon}}\}.$ The figure illustrates that the relative performances of $\phi^*$ and $\phi_{\hat{\sigma}, \mathrm{Bon}}$ for either choice of $\hat{\sigma}\in\{\hat{\sigma}_{\mathrm{sample}},\hat{\sigma}_{\mathrm{Bon}}\}$ is analogous to that illustrated in Figure \ref{figure_5}, where $\phi^*$ achieves a higher empirical power in the presence of weak signals. The figure shows that the empirical powers of $\phi^*$ and $\phi_{\hat{\sigma}}$ for either choice of $\hat{\sigma}$ do not differ much in the presence of weak signals. For strong signal strengths, $\phi^*$ tends to perform marginally better than $\phi_{\hat{\sigma}_\mathrm{sample}},$ but is outperformed by $\phi_{\hat{\sigma}_\mathrm{med}}$ in most cases. 

In this figure, as well as in Figures \ref{figure_11} and \ref{figure_13}, the empirical power of $\phi^*$ may not be trusted for large values of $\delta$ for the Horizontal case as we have observed that the realizations of $S^*$ tend to be large even when $q=20.$ 

\begin{figure}[ht]
    \centering
    {{\includegraphics[width=0.9\textwidth]{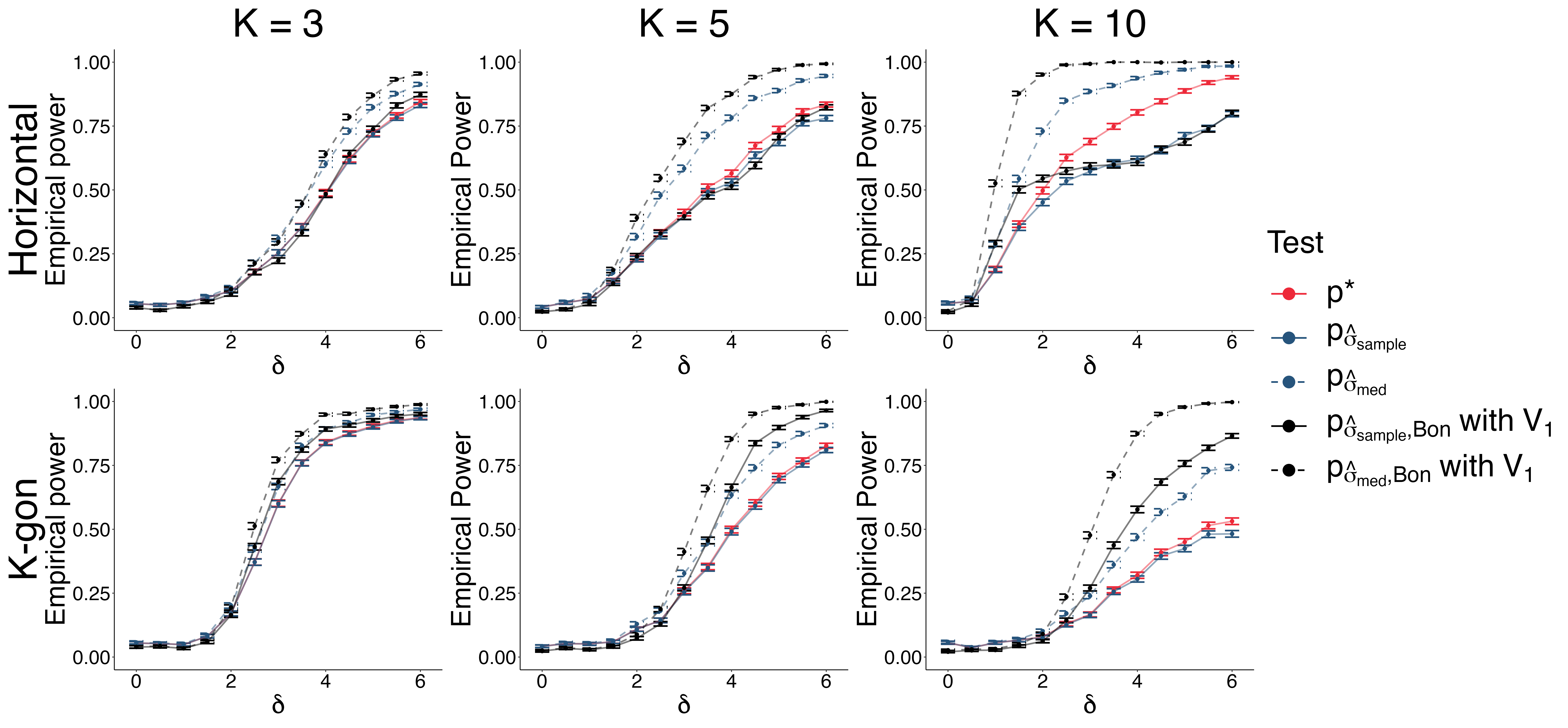}}}
    \caption{Plots of empirical powers of various tests for the null hypothesis in \eqref{eq:null_strong} with a pre-specified $\mathcal{V}.$}\label{figure_11}
\end{figure}

Figure \ref{figure_12} presents the empirical powers of $\phi^*$ and $\phi_{\sigma},$ where the latter assumes $\sigma$ is known, illustrating the loss of power due to not knowing $\sigma.$

\begin{figure}[ht]
    \centering
    {{\includegraphics[width=0.8\textwidth]{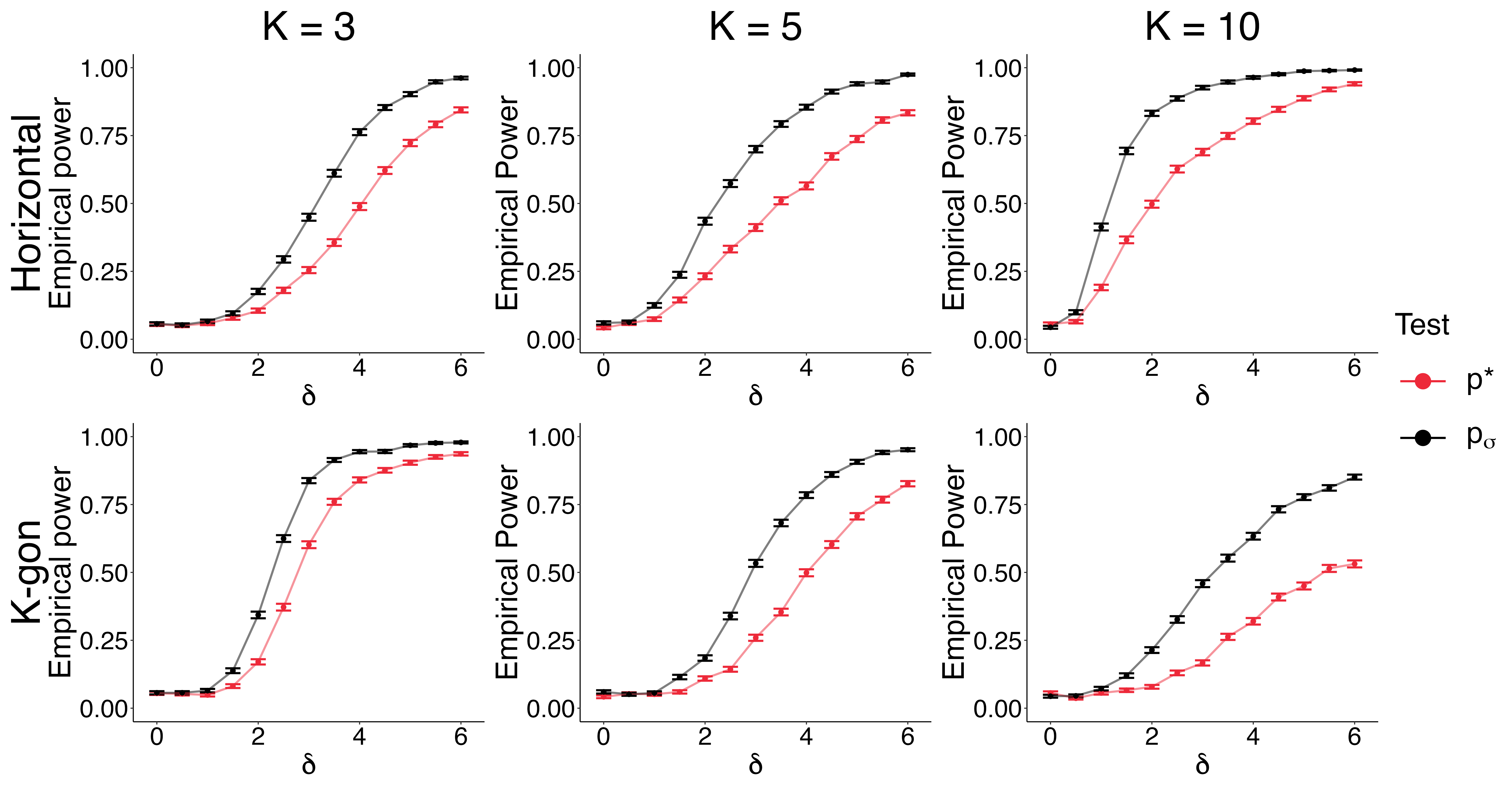}}}
    \caption{Plots of empirical powers of the proposed tests based on $p^*$ and $p_{\sigma}$ for the null hypothesis in \eqref{eq:null_strong} with a pre-specified $\mathcal{V}.$ }\label{figure_12}
\end{figure}

\subsubsection{Data-dependent choice of $\mathcal{V}$}\label{sec:simul_unknown_dep}
We next let $\mathcal{V}$ be chosen in a data-dependent way. We specifically consider the case where it is chosen as in Setting \ref{ex:1} with $K=20$ and $g=3.$ The leftmost plot of Figure \ref{figure_13} illustrates that $\phi^*_J$ controls the Type I error, as implied by Theorem \ref{thm:group_dep}. The middle and the rightmost plots compare the empirical powers of $\phi^*_J$ and $\phi_{\sigma},$ the latter of which assumes  $\sigma$ is known, to illustrate the loss of power due to not knowing $\sigma.$

\begin{figure}[ht]
    \centering
    {{\includegraphics[width=0.8\textwidth]{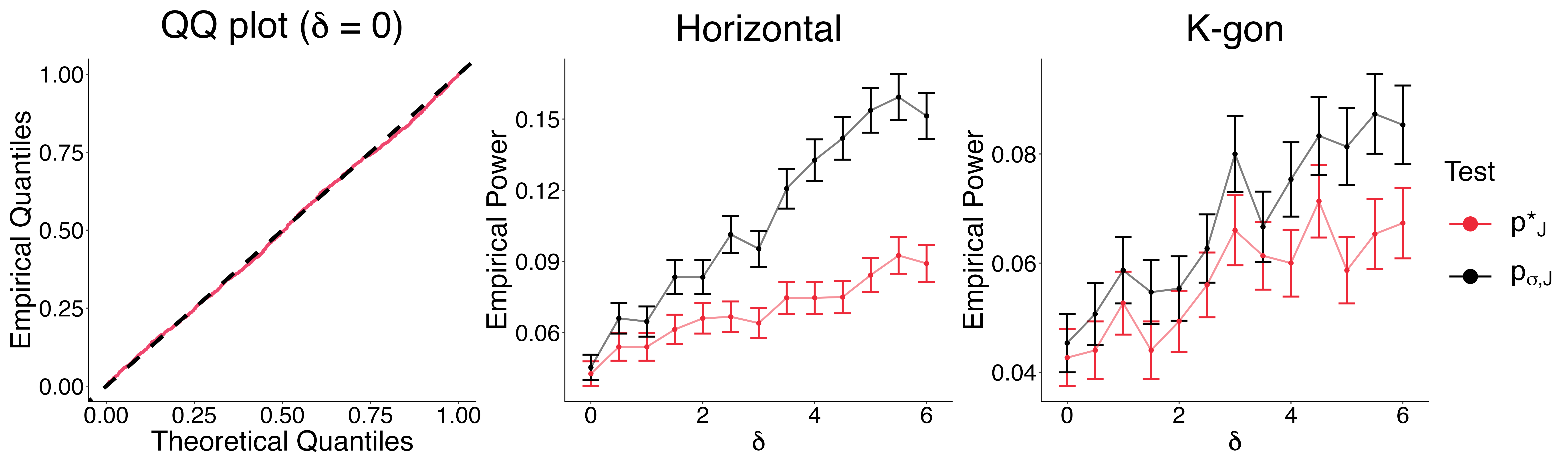}}}
    \caption{QQ plot of a sample of $p^*_J$ generated under the null hypothesis in \eqref{eq:null_strong} with $\mathcal{V}$ chosen according to Setting \ref{ex:2}, and 
    plots of empirical powers of the proposed tests based on $p_J^*$ and $p_{\sigma,J}$ for the same null hypothesis.}\label{figure_13}
\end{figure}

\section{Real data application}\label{sec:data}

We next apply the proposed methods to the penguins data of \citet{horst2020palmerpenguins}, which is also  studied in \citet{gao2022selective}, \citet{chen2023selective}, and \citet{yun2023selective}. The data consists of measurements of different body parts of three species of penguins (Adelie, Gentoo, and Chinstrap) of the male and female sexes. For our purposes, we specifically consider the variables bill depth (mm) and bill length (mm) of female penguins. We test for the null hypothesis in \eqref{eq:null_strong}: Sections \ref{sec:real_pre} and \ref{sec:real_dep} present the cases where $\mathcal{V}$ is pre-specified and chosen in a data-dependent way, respectively.

\subsection{Testing for $H_{0,\mathcal{V}}$ with pre-specified $\mathcal{V}$}\label{sec:real_pre}
For the case where $\mathcal{V}$ is pre-specified, we test for the equality of means of all observations, i.e., we consider the global null hypothesis $H^*_{\mathcal{V}}:\mu_i=\mu_{i'}$ for all $i,i'\in[n],$ where we set $\mathcal{V}=\mathcal{V}_1.$ We apply the tests $\phi_{\hat{\sigma}}$ and $\phi_{\hat{\sigma}, \mathrm{Bon}}$ for $\hat{\sigma}\in\{\hat{\sigma}_{\mathrm{sample}},\hat{\sigma}_{\mathrm{Bon}}\}$ for two different subsets of the penguins data, a subset consisting of observations from the Adelie species only and a subset consisting of observations from both the Adelie and Gentoo species, which we assume to be consistent with the null and the alternative hypotheses, respectively. We run $K$-means clustering on the standardized data with $K=4$ and present an average of 100 realizations for each p-value to account for the randomness in the initialization step of the algorithm. Note that we omit the test $\phi^*$ due to the relatively small number of features that could cause the associated p-value to be underestimated, as discussed in Section \ref{sec:simul_unknown}. 

Figure \ref{figure_14} presents visualizations of the outcomes of an instance of $K$-means clustering run on the two subsets of the data, and Table \ref{table_1} presents an average of 100 p-values associated with each of the tests considered. The table illustrates that the p-values of the proposed tests tend to be lower than those of the baseline testing procedure, both under the null and alternative hypotheses.

\begin{figure}[ht]
    \centering
    {{\includegraphics[width=0.7\textwidth]{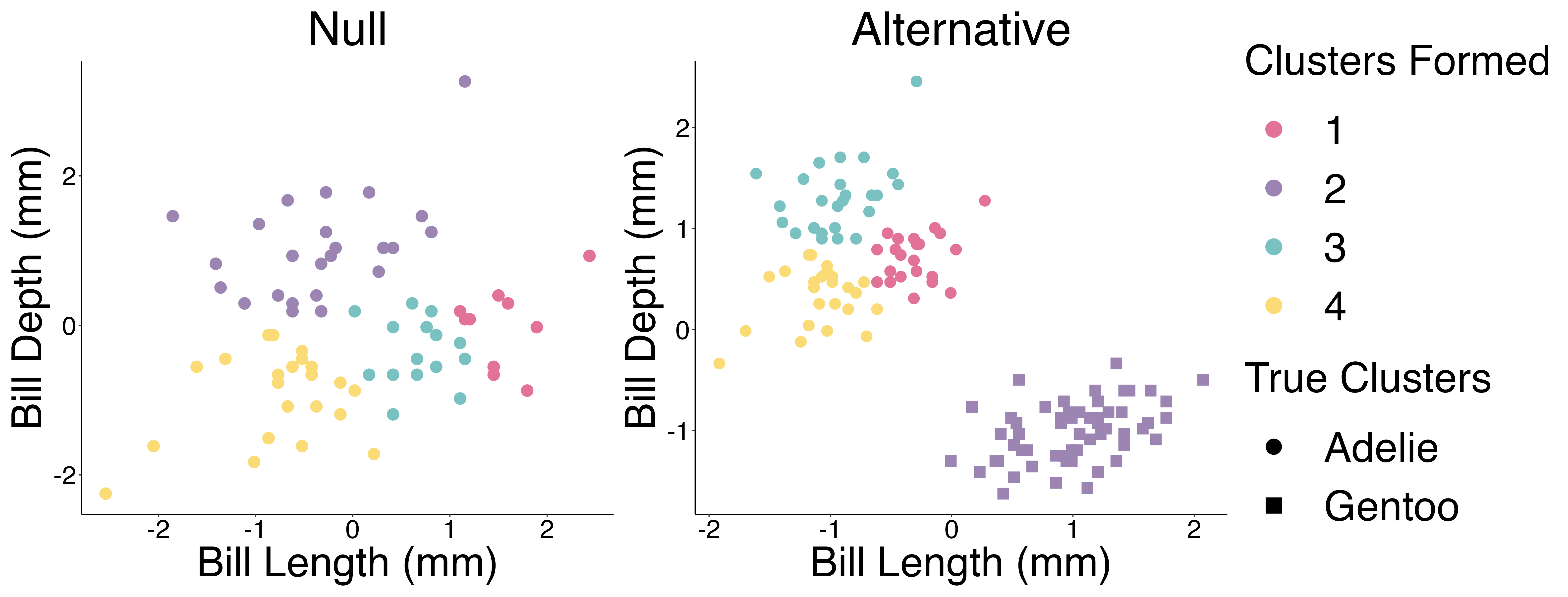}}}
    \caption{Visualizations of subsets of the standardized penguins data that are consistent with the null hypothesis (left) and the alternative hypothesis (right), along with outcomes of an instance of $K$-means clustering run with $K=4.$}\label{figure_14}
\end{figure}

\begin{table}[ht]
\centering
\begin{tabular}{rrrrr}
  \hline
 & $p_{\hat{\sigma}_{\mathrm{sample}}}$ & $p_{\hat{\sigma}_{\mathrm{med}}}$ & $p_{\hat{\sigma}_{\mathrm{sample}, \mathrm{Bon}}}$ & $p_{\hat{\sigma}_{\mathrm{med}, \mathrm{Bon}}}$ \\ 
  \hline\hline
Null & 0.52 & 0.51 & 0.65 & 0.63 \\ 
\hline
  Alternative & 0.26 & 0.33 & 0.47 & 0.56 \\ 
   \hline
\end{tabular}
\caption{Average of 100 p-values associated with various tests for the null hypothesis in \eqref{eq:null_strong} with a pre-specified $\mathcal{V}.$}
\end{table}\label{table_1}

\subsection{Testing for $H_{0,\mathcal{V}}$ with data-dependent choice of $\mathcal{V}$}\label{sec:real_dep}
We next choose $\mathcal{V}$ in a data-dependent way after running $K$-means clustering on the subset of the data that consists of observations from the Adelie and Gentoo species, which coincides with the subset visualized in the right-hand plot of Figure \ref{figure_14}. We apply the tests $\phi_{\hat{\sigma}}$ and $\phi_{\hat{\sigma}, J}$ for $\hat{\sigma}\in\{\hat{\sigma}_{\mathrm{sample}},\hat{\sigma}_{\mathrm{Bon}}\}$---again, we omit the test $\phi^*_J.$ We run $K$-means clustering on the standardized data with $K=10$ and choose $g=2$ pairs of clusters according to each of Setting \ref{ex:1} and Setting \ref{ex:2}. The two pairs that are chosen according to Setting \ref{ex:1} are $(2,3)$ and $(3,8),$ and those chosen according to Setting \ref{ex:2} are $(1,7)$ and $(6,7)$---note that such choices align with the visualizations of the clustering outcomes illustrated in Figure \ref{figure_15}. Table \ref{table_2} presents the p-values of the tests and shows that the tests that do not account for the data-dependence in the choice of $\mathcal{V}$ results in the same p-values as those that do. This observation is consistent with the simulation results of Section \ref{sec:simul_known_dep}, where we have observed that the effect of not accounting for this additional selection event is noticeable only for relatively large values of $K.$

\begin{figure}[ht]
    \centering
    {{\includegraphics[width=0.5\textwidth]{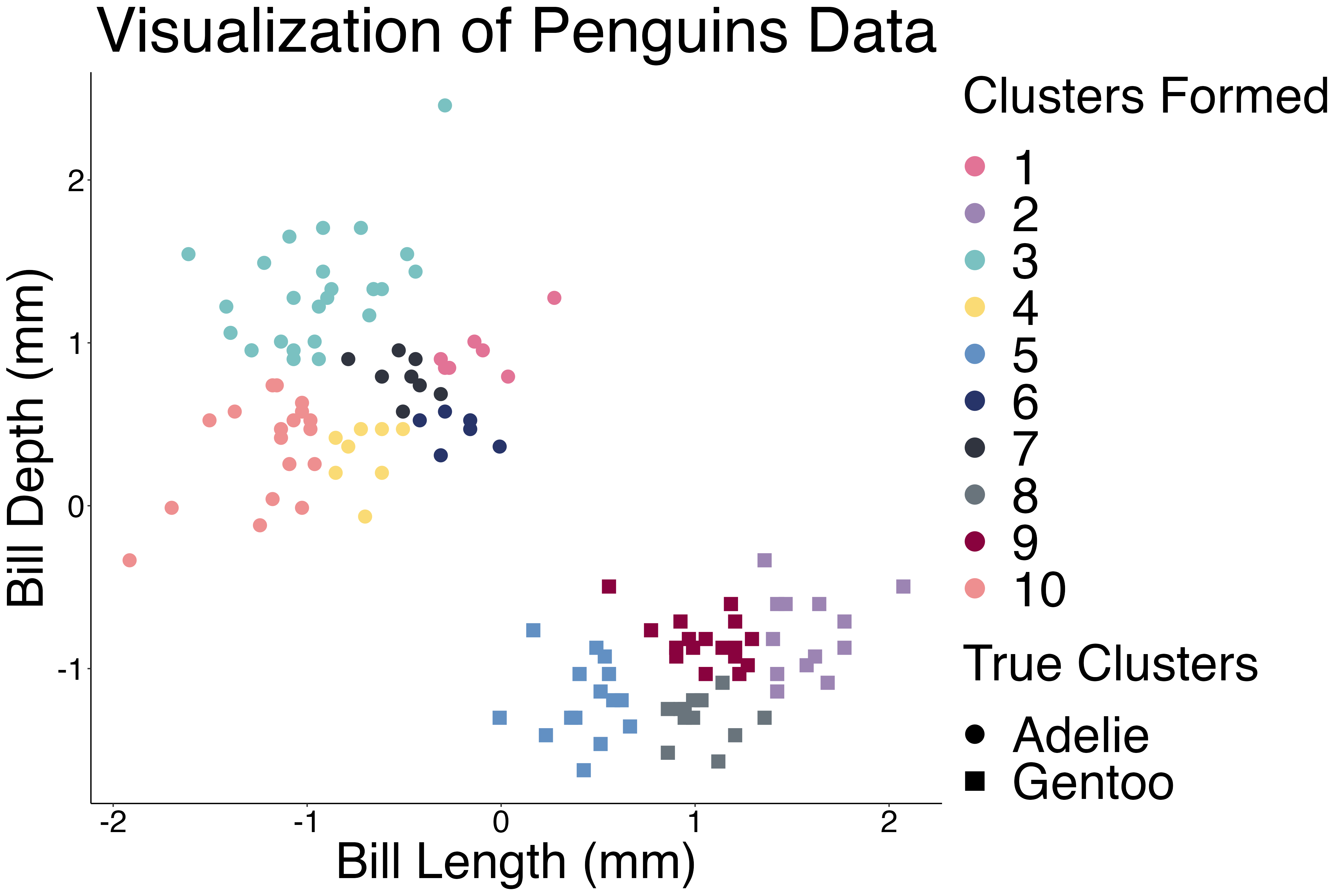}}}
    \caption{Visualizations of a subset of the standardized penguins data with $K$-means clustering run with $K=10.$}\label{figure_15}
\end{figure}

\begin{table}[ht]
\centering
\begin{tabular}{rrrrr}
  \hline
 & $p_{\hat{\sigma}_{\mathrm{sample}}}$ & $p_{\hat{\sigma}_{\mathrm{med}}}$ & $p_{\hat{\sigma}_{\mathrm{sample}}, J}$ & $p_{\hat{\sigma}_{\mathrm{med}},J}$ \\ 
  \hline\hline 
Setting 1 & 0.17 & 0.2 & 0.17 & 0.2 \\ \hline
  Setting 2 & 0.17 & 0.17 & 0.17 & 0.17 \\ 
   \hline
\end{tabular}
\caption{P-values associated with various tests for the null hypothesis in \eqref{eq:null_strong} with data-dependent choices of $\mathcal{V}.$}
\end{table}\label{table_2}

\section{Discussion}\label{sec:disc}
In this work, we have developed tests for multiple pairs of clusters for both known and unknown variance settings, extending the work of \cite{chen2023selective} and \cite{yun2023selective}. We have shown that the proposed tests control the Type I error, and we have derived expressions for the associated p-values that can be computed exactly. We have also presented numerical illustrations of the empirical powers of the proposed tests, which we have compared with that of the baseline testing procedure that combines \citet{chen2023selective}'s test with the Bonferroni correction---empirical results illustrate that the proposed tests tend to have a higher empirical power in the presence of weak signals. We have also briefly discussed a computational limitation of $\phi^*,$ which we speculate lies in the computation of probabilities associated with a truncated distribution. 


As a possible direction for future work, it would be of practical importance to develop tests for the null hypothesis $H_{0,\mathcal{V}}$ for data generated from a more flexible model that captures the complexity of real data. Another direction of computational importance would be to explore ways of computing more accurate selective p-values associated with a truncated $F$ distribution, especially for truncation sets that lie far in the tail of the distribution with non-negligible probability. In this work, we have approximated the $F$ distribution with the $\chi^2$ distribution, which is then approximated using the normal distribution. It could be of interest to study ways of approximating the $F$ distribution directly with the normal distribution. Improvements in the accuracy of such computations would not only address the computational limitation of the proposed test based on $p^*$ but also provide a valuable tool in the literature of selective inference.

\section{Acknowledgements}
Yinqiu He was partially supported by Wisconsin Alumni Research Foundation. We thank the Center for High Throughput Computing (\cite{https://doi.org/10.21231/gnt1-hw21}) for the computational resources that were used for producing the empirical results presented in Section \ref{sec:simul}.

\bibliographystyle{plainnat}
\bibliography{bib}

\appendix

\section{Appendix}
Appendix \ref{app:proofs} contains proofs of the theorems and propositions
presented in this paper. Appendix \ref{app:impl} provides details on the implementations of the proposed tests, which align with the codes available at \url{https://github.com/yjyun97/cluster_inf_multiple}.

\subsection{Proofs}\label{app:proofs}
The proofs of Theorems \ref{thm:main_thm} and \ref{thm:group_dep} can be found in Appendix \ref{pf:main_thm}, and those of Theorems \ref{thm:unknown_var} and \ref{thm:unknownvar_datav} can be found in Appendix \ref{pf:unknown_var}. Additionally, the proofs of Propositions  \ref{prop:quadratic}, \ref{prop:group}, \ref{prop:unknown_cts}, and \ref{V_dep_unknown} can be found in Appendices \ref{pf:quadratic},  \ref{pf:group}, \ref{pf:unknown_cts}, and \ref{pf:V_dep_unknown}, respectively. 

We first introduce additional notations that are used throughout the proofs. We write $Y_i\overset{\indep}{\sim}F$ for $i\in [m]$ and $m\in \mathbb{N}$ to denote that $Y_i$s follow the distribution $F$ independently, and we write $Y\overset{H_0}{\sim}F$ to denote that $Y$ follows the distribution $F$ under the null hypothesis $H_0.$ Likewise, $\overset{H_0}{=}$ and $\overset{H_0}{\propto}$ denote that the equality and the proportionality, respectively, hold under the null hypothesis $H_0.$

Before presenting the proofs, we state and prove Lemmas \ref{lem:mean_null} and \ref{lem:prop_mult_normal}, which are used in the proofs of Theorems \ref{thm:main_thm}, \ref{thm:group_dep}, \ref{thm:unknown_var}, and \ref{thm:unknownvar_datav}.

\begin{lemma}\label{lem:mean_null}For any vector $\mathbf{e}\in \mathcal{E},$ $ \mu^\top \mathbf{e}=\mathbf{0}_q$ under $H_{0,\mathcal{V}}.$
\end{lemma}
\begin{proof}
Since $\mathbf{e}\in \mathcal{E},$ $\mathbf{e}=\sum_{(k,k')\in\mathcal{V}}\lambda_{k,k'}\mathbf{v}_{k,k'}$ for some $\{\lambda_{k,k'}\}_{(k,k')\in\mathcal{V}}\subset \mathbb{R}.$ Then, $\mu^\top \mathbf{e}=\sum_{(k,k')\in\mathcal{V}}\lambda_{k,k'}\mu^\top \mathbf{v}_{k,k'}\overset{\mathcal{H}_{0,\mathcal{V}}}{=}0_q$ since $\mu^\top \mathbf{v}_{k,k'}=0_q$ for all $(k,k')\subset \mathcal{V}$ under $H_{0,\mathcal{V}}.$
\end{proof}

\begin{lemma}\label{lem:prop_mult_normal}
Let $X\in\mathbb{R}^{n\times q}$ be a matrix whose rows are distributed independently as in \eqref{data_gen}. Let $\Pi =\sum_{i=1}^r\mathbf{e}_i\mathbf{e}_i^\top$ for some $r\leq n,$ where $\mathbf{e}_i\in\mathbb{R}^n$ for $i\in [r]$ are orthonormal. Suppose $\mu$ in \eqref{data_gen} satisfies the condition that for each $i\in [r],$ $\mu^\top \mathbf{e}_i=0_q.$ Then, 
\[\|\Pi X\|_F\indep \frac{\Pi X}{\|\Pi X\|_F}.\]
\end{lemma}
\begin{proof}
Define $\mathbf{w}=[\mathbf{e}_1^\top X\cdots \mathbf{e}_{r}^\top X]^\top,$ and define $\tilde{I}_i\in\mathbb{R}^{rq\times q}$ for $i\in [r]$ to be matrices where $\tilde{I}_{i:i+q-1,:}=I_q,$ and the remaining entries are 0; here, $A_{i:i',:}$ for $i\leq i'\in\mathbb{N}$ denotes the entries of the matrix $A$ corresponding to the rows $i$ through $i'$ and all columns. Next, let $f:\mathbb{R}^{r q}\rightarrow \mathbb{R}^{n\times q},f(u)=\sum_{i=1}^{r}\mathbf{e}_iu^\top \tilde{I}_i.$ Note that 
\begin{align}\label{main_thm:first}
    f(\mathbf{w} )=\sum_{i=1}^{r}\mathbf{e}_i\mathbf{e}_i^\top X=\Pi X
\end{align}
and 
\begin{align}\label{main_thm:second}
\|f(\mathbf{w} )\|_F=\sqrt{\|\sum_{i=1}^{r}\mathbf{e}_i\mathbf{e}_i^\top X\|_F^2}=\sqrt{\sum_{i=1}^{r}\|\mathbf{e}_i\mathbf{e}_i^\top X\|^2_F}=\sqrt{\sum_{i=1}^{r}\|\mathbf{e}_i^\top X\|^2_2}=\|\mathbf{w} \|_2.
\end{align}
By the assumption that $\mu^\top \mathbf{e}_i=\mathbf{0}_q$ for all $i\in [r]$ and the distributional assumptions on $X,$ we have $\mathbf{w} \sim \mathcal{N}\left(0,\sigma^2I_{r q}\right).$ Thus, it follows that
\begin{align*}
    \|\mathbf{w} \|_2\indep \frac{\mathbf{w} }{\|\mathbf{w} \|_2}.
\end{align*}
Then, 
\[\|\mathbf{w} \|_2\indep f\left(\frac{\mathbf{w} }{\|\mathbf{w} \|_2}\right)=\frac{1}{\|\mathbf{w} \|_2}f\left(\mathbf{w} \right),\]
where the equality follows by linearity of $f.$ It then follows from \eqref{main_thm:second} that 
\[\|f(\mathbf{w} )\|_F\indep \frac{1}{\|f(\mathbf{w} )\|_F}f\left(\mathbf{w} \right).\] Finally, by \eqref{main_thm:first}, we have \[\norm{\Pi X}_F\indep\frac{\Pi X}{\norm{\Pi X}_F}.\]
\end{proof}

We now present the proofs of Theorem \ref{thm:main_thm}, Theorem \ref{thm:group_dep}, Theorem \ref{thm:unknown_var}, and Theorem \ref{thm:unknownvar_datav}.

\subsubsection{Proofs of Theorems \ref{thm:main_thm} and \ref{thm:group_dep}}\label{pf:main_thm}
The proofs closely follow \citet[Appendix A.1]{yun2023selective} and parts of \citet[Appendix A.1]{gao2022selective}.

We start by presenting the proof of Theorem \ref{thm:group_dep}. We first note that $p_{\sigma,J},$ conditioned on $Z,$ $\mathcal{C}(X),$ and $\mathcal{V}(X),$ is uniformly distributed under $H_{0,\mathcal{V}}$ by the property that $F(Y)$ is uniformly distributed for a random variable $Y$ with CDF $F.$ Then, for any $\alpha\in (0,1),$
\begin{align*}
\mathbb{P}_{H_{0,\mathcal{V}}}\left(p_{\sigma,J}\leq \alpha\mid \mathcal{C}(X),\mathcal{V}(X)\right)&=\mathbb{E}\left[\mathbb{P}_{H_{0,\mathcal{V}}}\left(p_{\sigma,J}\leq \alpha\mid \mathcal{C}(X), \mathcal{V}(X),Z\right)\right]=\mathbb{E}[\alpha]=\alpha,
\end{align*}
as desired. 

We next characterize the distribution function $Q_{\sigma, J}.$ For any  realization $\left(Z^0, \mathcal{C}^0, \mathcal{V}^0\right)$ of $\left(Z, \mathcal{C}(X), \mathcal{V}(X)\right),$ define 
\begin{align*}
    h(t, Z^0, \mathcal{C}^0, \mathcal{V}^0) = f_{T_\sigma \mid Z = Z^0,\;\mathcal{C}(X)=\mathcal{C}^0,\; \mathcal{V}(X)=\mathcal{V}^0}(t)
\end{align*}
for any $t>0.$
Then, $h(\cdot, Z, \mathcal{C}(X), \mathcal{V}(X))$ is the conditional PDF (probability density function) of $T_\sigma$ given $Z,$ $\mathrm{C}(X),$ and $\mathcal{V}(X).$ We aim to show 
\begin{align}\label{main_thm_goal}
    h(t, Z, \mathcal{C}(X), \mathcal{V}(X))\overset{H_{0,\mathcal{V}}}{\propto}
    f_{\chi_{d}}
    (t)\mathbbm{1}_{t\in S_{\sigma,J}},
\end{align}
where $f_{\chi_d}$ is the PDF of the $\chi_{d}$ distribution. To show \eqref{main_thm_goal}, we first state and prove Lemmas \ref{main_thm_goal1} and \ref{main_thm_goal2}.

\begin{lemma}\label{main_thm_goal1}
Suppose the cluster assignments and the set $\mathcal{V}$ are pre-specified, i.e., $\mathcal{C}_k$ for $k\in [K]$ and the set $\mathcal{V}$ are determined independently of the data. Then, $$T_{\sigma}^2 \overset{H_{0,\mathcal{V}}}{\sim} \chi^2_{d}.$$
\end{lemma}

\begin{proof}
    Let $\{\mathbf{e}_l:l\in [\mathrm{dim}(\mathcal{E})]\}$ be an orthonormal basis of $\mathcal{E}.$ Then, $P_{\mathcal{E}}=\sum_{l=1}^{\mathrm{dim}(\mathcal{E})}\mathbf{e}_l\mathbf{e}_l^\top,$ so we have 
\begin{align*}
    T_{\sigma}^2 =\frac{1}{\sigma^2}\norm{P_{\mathcal{E}}X}^2_F=\frac{1}{\sigma^2}\summ{l=1}{\mathrm{dim}(\mathcal{E})}\|\mathbf{e}_l^\top X\|_2^2=\frac{1}{\sigma^2}\summ{l=1}{\mathrm{dim}(\mathcal{E})}\summ{j=1}{q}(\mathbf{e}_l^\top X^j)^2,
\end{align*}
where the second equality follows from the orthogonality of the vectors $\mathbf{e}_l$ for $l\in  [\mathrm{dim}(\mathcal{E})].$ Note that for each $l\in [\mathrm{dim}(\mathcal{E})]$ and $j\in [q],$ $\mathbb{E}[\mathbf{e}_l^\top X^j]=\mathbf{e}_l^\top\mu^j\overset{H_{0,\mathcal{V}}}{=}0$ by Lemma \ref{lem:mean_null}, and $\mathrm{Var}[\mathbf{e}_l^\top X^j]=\sigma^2.$ By the orthogonality across $\mathbf{e}_l$ for $l\in [\mathrm{dim}(\mathcal{E})]$  and the distributional assumptions on $X^j,$ we have 
$\frac{1}{\sigma^2}\sum_{l=1}^{\mathrm{dim}(\mathcal{E})}(\mathbf{e}_l^\top X^j)^2\overset{H_{0,\mathcal{V}}}{\sim}\chi^2_{\mathrm{dim}(\mathcal{E})}.$ Finally, independence across $X^j$ for $j\in[q]$ gives $\frac{1}{\sigma^2}\summ{j=1}{q}\summ{l=1}{\mathrm{dim}(\mathcal{E})}(\mathbf{e}_l^\top X^j)^2\overset{H_{0,\mathcal{V}}}{\sim}\chi^2_{d}.$
\end{proof}

\begin{lemma}\label{main_thm_goal2}
Suppose the cluster assignments and the set $\mathcal{V}$ are pre-specified. Then,
$$T_{\sigma}^2\indep Z$$ under $H_{0,\mathcal{V}}.$ 
\end{lemma}
\begin{proof}
Recall 
$T_{\sigma}^2 \overset{H_{0,\mathcal{V}}}{\sim} \chi^2_{d}$ by Lemma \ref{main_thm_goal1}. We prove  $T_\sigma^2$ and $Z$ are independent under $H_{0,\mathcal{V}}$ by showing that 
\begin{align}\label{main_thm_goal2_goal.1}
T^2_\sigma \mid Z \overset{H_{0,\mathcal{V}}}{\sim} \chi^2_{d}.
\end{align}
Note that for each $j\in [q],$ $X^j\sim \mathcal{N}(\mu^j,\sigma^2I_n),$ so $P_{\mathcal{E}}X^j$ and 
$P_{\mathcal{E}}^\perp X^j$ are independent. By independence across $X^j$ for $j\in[q],$ we further have independence across $\left(P_{\mathcal{E}}X^j,\ \ P_{\mathcal{E}}^\perp X^j\right)$ for $j\in[q].$ As a result, $P_{\mathcal{E}}X$ and $P_{\mathcal{E}}^\perp X$ are independent, which allows \eqref{main_thm_goal2_goal.1} to be equivalently written as 
\begin{align}\label{main_thm_goal2_goal.1.2}
T_\sigma^2 \mid \frac{P_{\mathcal{E}}X}{\|P_{\mathcal{E}}X\|_F}\overset{H_{0,\mathcal{V}}}{\sim} \chi^2_{d}
\end{align}
since for random quantities $Y_1,Y_2,$ and $Y_3,$ we have that $Y_3\indep (Y_1,Y_2)$ implies $ Y_3\indep Y_1\mid Y_2.$ We next show $T_\sigma^2\indep\frac{P_{\mathcal{E}}X}{\norm{P_{\mathcal{E}}X}_F}.$ Let $\{\mathbf{e}_l:l\in [\mathrm{dim}(\mathcal{E})]\}$ be an orthonormal basis of $\mathcal{E}$ and write $P_{\mathcal{E}}=\sum_{l=1}^{\mathrm{dim}(\mathcal{E})} \mathbf{e}_l\mathbf{e}_l^\top,$ where $\mu^\top \mathbf{e}_l\overset{H_{0,\mathcal{V}}}{=}0_q$ for all $i\in [\mathrm{dim}(\mathcal{E})]$ by Lemma \ref{lem:mean_null}. Then, Lemma \ref{lem:prop_mult_normal} gives the desired result. It follows that \eqref{main_thm_goal2_goal.1.2}
can be equivalently written as $T^2_\sigma \overset{H_{0,\mathcal{V}}}{\sim}\chi^2_{d},$ which holds by Lemma \ref{main_thm_goal1}.
\end{proof}
We now show \eqref{main_thm_goal}. In the remainder of this section, we write $H_{0,\mathcal{V}(X)}(\mathcal{C}(X))$ to make explicit the dependence of $H_{0,\mathcal{V}}$ on $\mathcal{C}(X)$ and $\mathcal{V}(X).$ Fix any realization $\left(Z^0, \mathcal{C}^0, \mathcal{V}^0\right)$ of $\left(Z, \mathcal{C}(X), \mathcal{V}(X)\right),$ and define $\tilde{\mathcal{S}}_{\sigma, J}(\mathcal{C}^0, \mathcal{V}^0)=\tilde{\mathcal{S}_\sigma}(\mathcal{C}^0)\cap \tilde{S}_{\sigma, \mathcal{V}}(\mathcal{V}^0),$ where
\begin{align*}
\tilde{\mathcal{S}}_\sigma(\mathcal{C}^0) = \{(\psi,z) \in [0,\infty)\times (\mathbb{R}^{n\times q}\times \mathbb{R}^{n\times q}): \mathcal{C}(\tilde{x}_\sigma (\psi, z))=\mathcal{C}^0\}
\end{align*}
and
\begin{align*}
    \tilde{\mathcal{S}}_{\sigma, \mathcal{V}}(\mathcal{V}^0)=\{(\psi, z) \in [0,\infty)\times (\mathbb{R}^{n\times q}\times \mathbb{R}^{n\times q}): \mathcal{V}(\tilde{x}_\sigma(\psi,z))=\mathcal{V}^0\},
\end{align*}
where $\tilde{x}_\sigma:[0,\infty)\times (\mathbb{R}^{n\times q}\times \mathbb{R}^{n\times q})\rightarrow \mathbb{R}^{n\times q},$
\begin{align*}
	\tilde{x}_\sigma(\psi, (z_1, z_2))=\psi \cdot \sigma z_1+ z_2. 
\end{align*}  
Let $f_{T_\sigma}$ be the PDF of $T_\sigma$ in the setting where the cluster assignments and the set $\mathcal{V}$ are pre-specified; similarly define $f_{T_\sigma, Z}$ as the joint PDF of $T_\sigma$ and $Z,$ and $f_Z$ as the PDF of $Z.$ Note that  
\begin{align}\label{eq:thm3_note}
    \mathcal{C}(X)=\mathcal{C}^0\text{ and }\mathcal{V}(X) = \mathcal{V}^0\iff (T_\sigma, Z)\in \tilde{\mathcal{S}}_{\sigma, J}(\mathcal{C}^0, \mathcal{V}^0).
\end{align}
Then, 
\begin{align}
    h(t, Z^0, \mathcal{C}^0, \mathcal{V}^0) &= f_{T_\sigma \mid Z = Z^0,\;\mathcal{C}(X)=\mathcal{C}^0,\; \mathcal{V}(X)=\mathcal{V}^0}(t)\label{eq:thm3.1}\\
    &=f_{T_\sigma \mid Z = Z^0,\; (T_\sigma, Z)\in \tilde{\mathcal{S}}_{\sigma, J}(\mathcal{C}^0, \mathcal{V}^0)}(t)\label{eq:thm3.2}\\
    &\propto f_{T_\sigma, Z\mid (T_\sigma, Z)\in \tilde{\mathcal{S}}_{\sigma, J}(\mathcal{C}^0, \mathcal{V}^0)}(t, Z^0)\notag\\    &\overset{H_{0,\mathcal{V}^0}(\mathcal{C}^0)}{=} f_{T_\sigma}(t)f_Z(Z^0)\mathbbm{1}_{(t,Z^0)\in \tilde{\mathcal{S}}_{\sigma, J}(\mathcal{C}^0, \mathcal{V}^0)}\label{eq:thm3.3}\\
    &\propto f_{T_\sigma}(t)\mathbbm{1}_{(t,Z^0)\in \tilde{\mathcal{S}}_{\sigma, J}(\mathcal{C}^0, \mathcal{V}^0)}
\end{align}
where \eqref{eq:thm3.1} follows from the definition of the function $h,$ \eqref{eq:thm3.2} follows from the observation in \eqref{eq:thm3_note}, and \eqref{eq:thm3.3} holds by Lemma \ref{main_thm_goal2}. Thus, we have 
\begin{align*}
    h(t, Z, \mathcal{C}(X), \mathcal{V}(X))\overset{H_{0,\mathcal{V}(X)}(\mathcal{C}(X))}{\propto} f_{T_\sigma}(t) \mathbbm{1}_{(t,Z)\in \mathcal{\tilde{S}}_{\sigma,J}(\mathcal{C}(X), \mathcal{V}(X))}=f_{T_\sigma}(t) \mathbbm{1}_{t\in \mathcal{S}_{\sigma,J}},
\end{align*}
where the equality follows from the observation that 
\begin{align*}
    (t,Z)\in \mathcal{\tilde{S}}_{\sigma,J}(\mathcal{C}(X), \mathcal{V}(X))\iff t\in \mathcal{S}_{\sigma,J}.
\end{align*}
By Lemma \ref{main_thm_goal1},  $f_{T_\sigma}$ is the PDF of the $\chi_{d}$ distribution, 
and thus we have shown \eqref{main_thm_goal}.

To prove Theorem \ref{thm:main_thm}, recall that $\mathcal{V}$ is pre-specified in the setting of Theorem \ref{thm:main_thm}. Then, $Q_{\sigma,J}=Q_\sigma,$ and Theorem \ref{thm:main_thm} immediately follows from Theorem \ref{thm:group_dep}.

\subsubsection{Proofs of Theorems \ref{thm:unknown_var} and \ref{thm:unknownvar_datav}}\label{pf:unknown_var}

The proofs closely follow \citet[Appendix A.1]{yun2023selective} and parts of \citet[Appendix A.1]{gao2022selective}. 

We first present the proof of Theorem  \ref{thm:unknownvar_datav}. We first note that $p^*_J,$ conditioned on $Z^*,$ $\mathcal{C}(X),$ and $\mathcal{V}(X),$ is uniformly distributed under $H_{0,\mathcal{V}}^*$ by the property that $F(Y)$ is uniformly distributed for a random variable $Y$ with CDF $F.$ Then, for any $\alpha\in(0,1),$
\begin{align*}
    \mathbb{P}_{H_{0,\mathcal{V}}^*}(p^*_J\leq \alpha \mid \mathcal{C}(X), \mathcal{V}(X))=\mathbb{E}\left[\mathbb{P}_{H_{0,\mathcal{V}}^*}(p_J^*\leq \alpha \mid \mathcal{C}(X), \mathcal{V}(X), Z^*)\right]=\mathbb{E}[\alpha]=\alpha,
\end{align*}
as desired.

We next characterize the distribution function $Q_J^*.$ For any  realization $\left(Z^0, \mathcal{C}^0, \mathcal{V}^0\right)$ of $\left(Z^*, \mathcal{C}(X), \mathcal{V}(X)\right),$ define 
\begin{align*}
    h^*(t, Z^0, \mathcal{C}^0, \mathcal{V}^0) = f_{T^* \mid Z^* = Z^0,\;\mathcal{C}(X)=\mathcal{C}^0,\; \mathcal{V}(X)=\mathcal{V}^0}(t)
\end{align*}
for any $t>0.$
Then, $h^*(\cdot, Z^*, \mathcal{C}(X), \mathcal{V}(X))$ is the conditional PDF of $T^*$ given $Z^*,$ $\mathrm{C}(X),$ and $\mathcal{V}(X).$ We aim to show
\begin{align}\label{thm_unknown_var_goal}
    h^*(t, Z^*, \mathcal{C}(X), \mathcal{V}(X))\overset{H^*_{0,\mathcal{V}}}{\propto}
    f_{F_{d, d^*}}
    (t)\mathbbm{1}_{t\in S^*_{J}},
\end{align}
where $f_{F_{d, d^*}}$ is the PDF of the $F_{d, d^*}$ distribution. To show \eqref{thm_unknown_var_goal}, we first state and prove Lemmas \ref{lem:thm_unknown_var_lem0} and \ref{lem:thm_unknown_var_lem0.5}, which are used in the proofs of Lemmas \ref{lem:thm_unknown_var_lem1} and \ref{lem:thm_unknown_var_lem2}.

\begin{lemma}\label{lem:thm_unknown_var_lem0}
Let $P_{\mathcal{E}}$ and $P_{1}$ be as defined in Section \ref{sec:ftest}. Then, $P_{\mathcal{E}}$ and $P_{1}$ are orthogonal. 
\end{lemma}
\begin{proof}
We show that $\mathrm{Col}\left(P_{\mathcal{E}}\right)$ and $\mathrm{Col}\left(P_{1}\right)$ are orthogonal, where $\mathrm{Col}(A)$ for a matrix $A$ denotes its column space. Note 
\[\mathrm{Col}\left(P_{\mathcal{E}}\right)=\mathrm{span}\{\mathbf{v}_{k,k'}:(k,k')\in \mathcal{V}\}\]
and recall that for any $k,k'\in[K],$ $\mathbf{v}_{k,k'}=\frac{1}{|\mathcal{C}_k|}\mathbf{1}_{\mathcal{C}_k}-\frac{1}{|\mathcal{C}_{k'}|}\mathbf{1}_{\mathcal{C}_{k'}}.$ Therefore, we have 
\[\mathrm{Col}(P_{\mathcal{E}})\subset \mathrm{span}\{\mathbf{1}_{\mathcal{C}_k}:k\in \mathcal{K}\},\] 
where $\mathcal{K}$ is as defined in Section \ref{sec:ftest}. Thus, to show $\mathrm{Col}\left(P_{\mathcal{E}}\right)\perp\mathrm{Col}\left(P_{1}\right),$ it is enough to show \[\mathrm{span}\{\mathbf{1}_{\mathcal{C}_k}:k\in \mathcal{K}\}\perp\mathrm{Col}\left(P_{1}\right),\] which immediately follows from the orthogonality between 
\begin{align*}
    \sum_{k\in\mathcal{K}}\frac{\mathbf{1}_{\mathcal{C}_k}\mathbf{1}_{\mathcal{C}_k}^\top}{|\mathcal{C}_k|}\hspace{1em}\text{and}\hspace{1em}\sum_{k\in \mathcal{K}}I_{\mathcal{C}_k}-\sum_{k\in \mathcal{K}}\frac{\mathbf{1}_{\mathcal{C}_k}\mathbf{1}_{\mathcal{C}_k}^\top}{|\mathcal{C}_k|}.
\end{align*}
\end{proof}

\begin{lemma}\label{lem:thm_unknown_var_lem0.5}
For each $k\in \mathcal{K},$ $\left(I_{\mathcal{C}_k}-\frac{\mathbf{1}_{\mathcal{C}_k}\mathbf{1}_{\mathcal{C}_k}^\top}{|\mathcal{C}_k|}\right)\mu=0_{n\times q}$ under $H_{0,\mathcal{V}}^*.$
\end{lemma}
\begin{proof}
Fix any $k\in\mathcal{K}.$ Under $H_{0,\mathcal{V}}^*,$ $\mu_i=\mu_0$ for some $\mu_0\in\mathbb{R}^q$ for all $i\in \mathcal{C}_k.$ We thus have $I_{\mathcal{C}_k}\mu\overset{H_{0,\mathcal{V}}^*}{=}\mathbf{1}_{\mathcal{C}_k}\mu_0^\top.$ Likewise, $\mathbf{1}_{\mathcal{C}_k}^\top \mu\overset{H_{0,\mathcal{V}}^*}{=}|\mathcal{C}_k|\mu_0^\top,$ giving $\frac{\mathbf{1}_{\mathcal{C}_k}\mathbf{1}_{\mathcal{C}_k}^\top}{|\mathcal{C}_k|}\mu \overset{H_{0,\mathcal{V}}^*}{=}\mathbf{1}_{\mathcal{C}_k}\mu_0^\top.$ Therefore, 
\begin{align*}
    \left(I_{\mathcal{C}_k}-\frac{\mathbf{1}_{\mathcal{C}_k}\mathbf{1}_{\mathcal{C}_k}^\top}{|\mathcal{C}_k|}\right)\mu=I_{\mathcal{C}_k}\mu-\frac{\mathbf{1}_{\mathcal{C}_k}\mathbf{1}_{\mathcal{C}_k}^\top}{|\mathcal{C}_k|}\mu\overset{H_{0,\mathcal{V}}^*}{=}\mathbf{1}_{\mathcal{C}_k}\mu_0^\top-\mathbf{1}_{\mathcal{C}_k}\mu_0^\top=0_{n\times q}.
\end{align*}
\end{proof}

\begin{lemma}\label{lem:thm_unknown_var_lem1}
Suppose the cluster assignments and the set $\mathcal{V}$ are pre-specified. Then, 
\[T^*\overset{H_{0,\mathcal{V}}^*}{\sim}F_{d, d^*}.\]
\end{lemma}
\begin{proof} We first show $\frac{1}{\sigma^2}\|P_{\mathcal{E}}X\|^2_F\overset{H_{0,\mathcal{V}}^*}{\sim}\chi^2_{d}.$ Let $\{\mathbf{e}_l:l\in \mathrm{dim}(\mathcal{E})\}$ be an orthonormal basis of $\mathcal{E}$ and write $P_{\mathcal{E}}= \sum_{l=1}^{\mathrm{dim}(\mathcal{E})}\mathbf{e}_l\mathbf{e}_l^\top.$ Then, 
\begin{align*}
\|P_{\mathcal{E}}X\|^2_F=\|\sum_{l=1}^{\mathrm{dim}(\mathcal{E})}\mathbf{e}_l\mathbf{e}_l^\top X\|^2_F=\sum_{l=1}^{\mathrm{dim}(\mathcal{E})}\|\mathbf{e}_l^\top X\|^2_2=\sum_{l=1}^{\mathrm{dim}(\mathcal{E})}\sum_{j=1}^q(\mathbf{e}_l^\top X^j)^2.
\end{align*}
Fix any $l\in \mathrm{dim}(\mathcal{E})$ and $j\in [q].$ We have $\mathbb{E}[\mathbf{e}_l^\top X^j]\overset{H_{0,\mathcal{V}}^*}{=}0$ and $\mathrm{Var}[\mathbf{e}_l^\top X^j]=\sigma^2,$ where the former follows by Lemma \ref{lem:mean_null} since $H_{0,\mathcal{V}}^*$ implies $H_{0,\mathcal{V}}.$ Then, independence across $\mathbf{e}_l$ for $l\in [\mathrm{dim}(\mathcal{E})]$ and the distributional assumptions on $X^j$ imply 
$\frac{1}{\sigma^2}\sum_{l=1}^{\mathrm{dim}(\mathcal{E})}(\mathbf{e}_l^\top X^j)^2\overset{H_{0,\mathcal{V}}^*}{\sim}\chi^2_{\mathrm{dim}(\mathcal{E})}.$ Finally, independence across $X^j$ for $j\in [q]$ gives $\frac{1}{\sigma^2}\sum_{j=1}^q\sum_{l=1}^{\mathrm{dim}(\mathcal{E})}(\mathbf{e}_l^\top X^j)^2\overset{H_{0,\mathcal{V}}^*}{\sim}\chi^2_{d}.$

We next show $\frac{1}{\sigma^2}\|P_{1}X\|^2_F\overset{H_{0,\mathcal{V}}^*}{\sim}\chi^2_{d^*}.$  
\begin{align*}
    \|P_{1}X\|^2_F=\sum_{k\in \mathcal{K}}\|\left(I_{\mathcal{C}_k}-\frac{\mathbf{1}_{\mathcal{C}_k}\mathbf{1}_{\mathcal{C}_k}^\top}{|\mathcal{C}_k|}\right)X\|^2_F=\sum_{k\in \mathcal{K}}\sum_{j=1}^q\|\left(I_{\mathcal{C}_k}-\frac{\mathbf{1}_{\mathcal{C}_k}\mathbf{1}_{\mathcal{C}_k}^\top}{|\mathcal{C}_k|}\right)X^j\|^2_2,
\end{align*}
where the first equality holds by the orthogonality across $I_{\mathcal{C}_k}$ for $k\in \mathcal{K}$ and the fact that $\mathrm{Col}\left(I_{\mathcal{C}_k}-\frac{\mathbf{1}_{\mathcal{C}_k}\mathbf{1}_{\mathcal{C}_k}^\top}{|\mathcal{C}_k|}\right)\subset \mathrm{Col}\left(I_{\mathcal{C}_k}\right)$ for each $k\in \mathcal{K}.$ Fix any $k\in \mathcal{K}$ and $j\in[q],$ and note that 
\begin{align*}
\frac{1}{\sigma^2}\|\left(I_{\mathcal{C}_k}-\frac{\mathbf{1}_{\mathcal{C}_k}\mathbf{1}_{\mathcal{C}_k}^\top}{|\mathcal{C}_k|}\right)X^j\|^2_2=\frac{X^j}{\sigma}^\top \left(I_{\mathcal{C}_k}-\frac{\mathbf{1}_{\mathcal{C}_k}\mathbf{1}_{\mathcal{C}_k}^\top}{|\mathcal{C}_k|}\right)\frac{X^j}{\sigma},
\end{align*}
where $\frac{X^j}{\sigma}\sim \mathcal{N}(\mu^j,I_n),$ and $I_{\mathcal{C}_k}-\frac{\mathbf{1}_{\mathcal{C}_k}\mathbf{1}_{\mathcal{C}_k}^\top}{|\mathcal{C}_k|}$ is idempotent and has rank $|\mathcal{C}_k|-1.$ It follows that 
\begin{align*}
    \frac{1}{\sigma^2}\|\left(I_{\mathcal{C}_k}-\frac{\mathbf{1}_{\mathcal{C}_k}\mathbf{1}_{\mathcal{C}_k}^\top}{|\mathcal{C}_k|}\right)X^j\|^2_2\sim \chi_{|\mathcal{C}_k|-1}^2\left(\frac{1}{\sigma^2}{\mu^j}^\top \left(I_{\mathcal{C}_k}-\frac{\mathbf{1}_{\mathcal{C}_k}\mathbf{1}_{\mathcal{C}_k}^\top}{|\mathcal{C}_k|}\right)\mu^j\right)\overset{H_{0,\mathcal{V}}^*}{=}\chi_{|\mathcal{C}_k|-1}^2
\end{align*}
since $\left(I_{\mathcal{C}_k}-\frac{\mathbf{1}_{\mathcal{C}_k}\mathbf{1}_{\mathcal{C}_k}^\top}{|\mathcal{C}_k|}\right)\mu^j\overset{H_{0,\mathcal{V}}^*}{=}0_n$ by Lemma \ref{lem:thm_unknown_var_lem0.5}.
Note that 
$\left(I_{\mathcal{C}_k}-\frac{\mathbf{1}_{\mathcal{C}_k}\mathbf{1}_{\mathcal{C}_k}^\top}{|\mathcal{C}_k|}\right)X^j$ for $k\in\mathcal{K}$ are independent by orthogonality across $\left(I_{\mathcal{C}_k}-\frac{\mathbf{1}_{\mathcal{C}_k}\mathbf{1}_{\mathcal{C}_k}^\top}{|\mathcal{C}_k|}\right)$ for $k\in\mathcal{K}$ and the distributional assumptions on $X^j.$ It follows that 
\begin{align*}
\frac{1}{\sigma^2}\sum_{k=1}^K\|\left(I_{\mathcal{C}_k}-\frac{\mathbf{1}_{\mathcal{C}_k}\mathbf{1}_{\mathcal{C}_k}^\top}{|\mathcal{C}_k|}\right)X^j\|^2_2\overset{H_{0,\mathcal{V}}^*}{\sim} \chi^2_{\sum_{k\in \mathcal{K}}(|\mathcal{C}_k|-1)}.
\end{align*}
Then, independence across $X^j$ for $j\in [q]$ gives
\begin{align*}
    \frac{1}{\sigma^2}\sum_{j=1}^q\sum_{k=1}^K\|\left(I_{\mathcal{C}_k}-\frac{\mathbf{1}_{\mathcal{C}_k}\mathbf{1}_{\mathcal{C}_k}^\top}{|\mathcal{C}_k|}\right)X^j\|^2_2\overset{H_{0,\mathcal{V}}^*}{\sim}\chi^2_{d^*}.
\end{align*}

Finally, we show that $\|P_{\mathcal{E}}X\|^2_F$ and $\|P_{1}X\|^2_F$ are independent. Note that
$P_{\mathcal{E}}$ and $P_{1}$ are orthogonal by Lemma  \ref{lem:thm_unknown_var_lem0}, so for each $j\in[q],$ the distributional assumptions on $X^j$ imply independence between $P_{\mathcal{E}}X^j$ and $P_{1}X^j.$ Then, independence across $X^j$ for $j\in[q]$ implies independence across $\left(P_{\mathcal{E}}X^j,P_{1}X^j\right)$ for $j\in[q],$ which then implies that $\|P_{\mathcal{E}}X\|^2_F$ and $\|P_{1}X\|^2_F$ are independent. Thus, by definition of the $F$ distribution, we have 
\begin{align*}
    \frac{\frac{1}{\sigma^2}\|P_{\mathcal{E}}X\|^2_F/d}{\frac{1}{\sigma^2}\|P_{1}X\|^2_F/d^*}\overset{H_{0,\mathcal{V}}^*}{\sim}F_{d,d^*}.
\end{align*}
\end{proof}

\begin{lemma}\label{lem:thm_unknown_var_lem2}
Suppose the cluster assignments and the set $\mathcal{V}$ are pre-specified. Then, \[T^*\indep Z^*\] under $H_{0,\mathcal{V}}^*.$
\end{lemma}
\begin{proof}
By Lemma \ref{lem:thm_unknown_var_lem1}, we have 
\[T^*\overset{H_{0,\mathcal{V}}^*}{\sim}F_{d, d^*}.\] 
We prove $T^*$ and $Z^*$ are independent under $H_{0,\mathcal{V}}^*$ by showing that
\begin{align}\label{thm_unknown_var_lem2_goal}
    T^*\mid Z^*\overset{H_{0,\mathcal{V}}^*}{\sim}F_{d, d^*}.
\end{align}
To show \eqref{thm_unknown_var_lem2_goal}, we repeatedly use the fact that for random quantities $Y_1,Y_2,$ and $Y_3,$ 
\begin{align}\label{thm_unknown_var_lem2_fact}
    Y_3\indep (Y_1,Y_2)\implies Y_3\indep Y_1\:|\:Y_2.
\end{align}
For each $j\in [q],$ the distributional assumptions on $X^j$ and the orthogonality across $P_{\mathcal{E}},$ $P_{1},$ and $P_{2}$ give independence across $P_{\mathcal{E}}X^j,$ $P_{1}X^j,$ and $P_{2}X^j.$ Furthermore, the independence across $X^j$ for $j\in[q]$ implies independence across $(P_{\mathcal{E}}X^j,P_{1}X^j,P_{2}X^j)$ for $j\in[q].$ As a result, $P_{\mathcal{E}}X,$ $P_{1}X,$ and $P_{2}X$ are independent, so \eqref{thm_unknown_var_lem2_goal} can be equivalently written as 
\begin{align}\label{thm_unknown_var_lem2_goal.2}
T^*\ \  |\ \ \left(\frac{P_{\mathcal{E}}X}{\|P_{\mathcal{E}}X\|_F},\ \  \frac{P_{1}X}{\|P_{1}X\|_F},\ \  \|P_{\mathcal{E}}X\|_F^2+\|P_{1}X\|_F^2\right)\overset{H_{0,\mathcal{V}}^*}{\sim}F_{d, d^*},
\end{align}
where we use the fact in \eqref{thm_unknown_var_lem2_fact}. Next, note that independence between $P_{\mathcal{E}}X$ and $P_{1}X$ gives

\begin{align}\label{lem_thm_unknown_var_lem2_second.1}
\left(\|P_{\mathcal{E}}X\|_F^2,\ \  \frac{P_{\mathcal{E}}X}{\|P_{\mathcal{E}}X\|_F}\right)\indep \left(\|P_{1}X\|_F^2,\ \  \frac{P_{1}X}{\|P_{1}X\|_F}\right).
\end{align}
Further note that there exist orthonormal bases $\{\mathbf{e}_l\}_{l=1}^{\mathrm{dim}(\mathcal{E})}$ and $\{\mathbf{e}'_l\}_{l=1}^{\sum_{k\in \mathcal{K}}|\mathcal{C}_k|-|\mathcal{K}|}$ 
of the column spaces of $P_{\mathcal{E}}$ and $P_{1},$ respectively, so we can write $P_{\mathcal{E}}=\sum_{l=1}^{\mathrm{dim}(\mathcal{E})}\mathbf{e}_l\mathbf{e}_l^\top$ and $P_{1}=\sum_{l=1}^{\sum_{k\in \mathcal{K}}|\mathcal{C}_k|-|\mathcal{K}|}\mathbf{e}'_l{\mathbf{e}'_l}^\top.$ For each $l\in [\mathrm{dim}(\mathcal{E})],$ $\mu^\top \mathbf{e}_l\overset{H_{0,\mathcal{V}}^*}{=}\mathbf{0}_q$ by Lemma \ref{lem:mean_null} since $H_{0,\mathcal{V}}^*$ implies $H_{0,\mathcal{V}}.$ Furthermore, $\mu^\top {\mathbf{e}'_l}\overset{H_{0,\mathcal{V}}^*}{=}\mathbf{0}_q$ for each $l\in [m-K]$ by Lemma \ref{lem:thm_unknown_var_lem0.5} since $\mathbf{e}'_l=\sum_{k\in\mathcal{K}}\left(I_{\mathcal{C}_k}-\frac{\mathbf{1}_{\mathcal{C}_k}\mathbf{1}_{\mathcal{C}_k}^\top}{|\mathcal{C}_k|}\right)a$ for some $a\in \mathbb{R}^n.$ Thus, by Lemma \ref{lem:prop_mult_normal}, we have 
\begin{align}\label{lem_thm_unknown_var_lem2_second.2}
    \|P_{\mathcal{E}}X\|_F^2\indep \frac{P_{\mathcal{E}}X}{\|P_{\mathcal{E}}X\|_F}\hspace{1em}\text{and}\hspace{1em}\|P_{1}X\|_F^2\indep \frac{P_{1}X}{\|P_{1}X\|_F}
\end{align}
under $H_{0,\mathcal{V}}^*.$ Then, \eqref{lem_thm_unknown_var_lem2_second.1} and \eqref{lem_thm_unknown_var_lem2_second.2} imply
\begin{align*}
\left(\|P_{\mathcal{E}}X\|_F^2,\ \ \|P_{1}X\|_F^2 \right)\indep \left(
\frac{P_{\mathcal{E}}X}{\|P_{\mathcal{E}}X\|_F},\ \  \frac{P_{1}X}{\|P_{1}X\|_F}\right),
\end{align*}
allowing \eqref{thm_unknown_var_lem2_goal.2} to be equivalently written as 
\begin{align}\label{thm_unknown_var_lem2_goal.3}
T^*\ \ |\ \  \|P_{\mathcal{E}}X\|_F^2+\|P_{1}X\|_F^2\overset{H_{0,\mathcal{V}}^*}{\sim}F_{d, d^*}.
\end{align}
Finally, $T^*\indep\|P_{\mathcal{E}}X\|_F^2+\|P_{1}X\|_F^2 $ under $H_{0,\mathcal{V}}^*$ since $T^*=\frac{\frac{1}{\sigma^2}\|P_{\mathcal{E}}X\|^2_F/d}{\frac{1}{\sigma^2}\|P_{1}X\|^2_F/d^*},$ where $\frac{1}{\sigma^2}\|P_{\mathcal{E}}X\|^2_F\overset{H_{0,\mathcal{V}}^*}{\sim }\chi^2_{d},$ $\frac{1}{\sigma^2}\|P_{1}X\|^2_F\overset{H_{0,\mathcal{V}}^*}{\sim }\chi^2_{d^*},$ and $\|P_{\mathcal{E}}X\|^2_F\indep \frac{1}{\sigma^2}\|P_{1}X\|^2_F$ under $H_{0,\mathcal{V}}^*.$ Therefore, 
\eqref{thm_unknown_var_lem2_goal.3} is equivalent to 
\begin{align}\label{thm_unknown_var_lem2_goal.4}
T^*\overset{H_{0,\mathcal{V}}^*}{\sim}F_{d, d^*},
\end{align}
which holds by Lemma \ref{lem:thm_unknown_var_lem1}.
\end{proof}

We now show \eqref{thm_unknown_var_goal}. In the remainder of this section, we write $H^*_{0,\mathcal{V}(X)}(\mathcal{C}(X))$ to make explicit the dependence of $H^*_{0,\mathcal{V}}$ on $\mathcal{C}(X)$ and $\mathcal{V}(X).$ Fix any realization $\left(Z^0, \mathcal{C}^0, \mathcal{V}^0\right)$ of $\left(Z^*, \mathcal{C}(X), \mathcal{V}(X)\right),$ and define $\tilde{\mathcal{S}}^*_{J}(\mathcal{C}^0, \mathcal{V}^0)=\tilde{\mathcal{S}^*}(\mathcal{C}^0)\cap \tilde{S}_{\mathcal{V}}^*(\mathcal{V}^0),$ where
\begin{align*}
\tilde{\mathcal{S}}^*(\mathcal{C}^0) = \{(\psi,z) \in [0,\infty)\times (\mathbb{R}^{n\times q}\times \mathbb{R}^{n\times q}\times (0,\infty)\times \mathbb{R}^{n\times q}): \mathcal{C}(\tilde{x}^*(\psi,z))=\mathcal{C}^0\}
\end{align*}
and
\begin{align*}
    \tilde{\mathcal{S}}^*_{\mathcal{V}}(\mathcal{V}^0)=\{(\psi,z) \in [0,\infty)\times (\mathbb{R}^{n\times q}\times \mathbb{R}^{n\times q}\times (0,\infty)\times \mathbb{R}^{n\times q}):\mathcal{V}(\tilde{x}^*(\psi,z))=\mathcal{V}^0\},
\end{align*}
where $\tilde{x}^*:[0,\infty)\times (\mathbb{R}^{n\times q}\times \mathbb{R}^{n\times q}\times (0,\infty)\times \mathbb{R}^{n\times q})\rightarrow \mathbb{R}^{n\times q},$
\[\tilde{x}^*(\psi, (z_1,z_2,z_3,z_4))=\sqrt{z_3}\left(\sqrt{\frac{\psi}{\psi+r^*}}z_1+\sqrt{\frac{r^*}{\psi+r^*}}z_2\right)+z_4.\] 
Let $f_{T^*}$ be the PDF of $T^*$ in the setting where the cluster assignments and the set $\mathcal{V}$ are pre-specified; similarly define $f^*_{T^*, Z^*}$ as the joint PDF of $T^*$ and $Z^*,$ and $f_{Z^*}$ as the PDF of $Z^*.$ Note that  
\begin{align}\label{eq:thm7_note}
    \mathcal{C}(X)=\mathcal{C}^0\text{ and }\mathcal{V}(X) = \mathcal{V}^0\iff (T^*, Z^*)\in \tilde{\mathcal{S}}^*_{J}(\mathcal{C}^0, \mathcal{V}^0).
\end{align}
Then, 
\begin{align}
    h^*(t, Z^0, \mathcal{C}^0, \mathcal{V}^0) &= f_{T^* \mid Z^* = Z^0,\;\mathcal{C}(X)=\mathcal{C}^0,\; \mathcal{V}(X)=\mathcal{V}^0}(t)\label{eq:thm7.1}\\
    &=f_{T^*\mid Z^* = Z^0,\; (T^*, Z^*)\in \tilde{\mathcal{S}}^*_{J}(\mathcal{C}^0, \mathcal{V}^0)}(t)\label{eq:thm7.2}\\
    &\propto f_{T^*, Z^*\mid (T^*, Z^*)\in \tilde{\mathcal{S}}^*_{J}(\mathcal{C}^0, \mathcal{V}^0)}(t, Z^0)\notag\\
    &\overset{H^*_{0,\mathcal{V}^0}(\mathcal{C}^0)}{=}  f_{T^*}(t)f_{Z^*}(Z^0)\mathbbm{1}_{(t,Z^0)\in \tilde{\mathcal{S}}^*_{J}(\mathcal{C}^0, \mathcal{V}^0)}\label{eq:thm7.3}\\
    &\propto f_{T^*}(t)\mathbbm{1}_{(t,Z^0)\in \tilde{\mathcal{S}}^*_{J}(\mathcal{C}^0, \mathcal{V}^0)}
\end{align}
where \eqref{eq:thm7.1} follows from the definition of the function $h^*,$ \eqref{eq:thm7.2} follows from the observation in \eqref{eq:thm7_note}, and \eqref{eq:thm7.3} holds by Lemma \ref{lem:thm_unknown_var_lem2}. Thus, we have 
\begin{align*}
    h^*(t, Z^*, \mathcal{C}(X), \mathcal{V}(X))\overset{H^*_{0,\mathcal{V}(X)}(\mathcal{C}(X))}{\propto} f_{T^*}(t) \mathbbm{1}_{(t,Z^*)\in \mathcal{\tilde{S}}^*_{J}(\mathcal{C}(X), \mathcal{V}(X))}=f_{T^*}(t) \mathbbm{1}_{t\in \mathcal{S}^*_{J}},
\end{align*}
where the equality follows from the observation that 
\begin{align*}
    (t,Z^*)\in \mathcal{\tilde{S}}^*_{J}(\mathcal{C}(X), \mathcal{V}(X))\iff t\in \mathcal{S}^*_{J}.
\end{align*}
By Lemma \ref{lem:thm_unknown_var_lem1},  we have that $f_{T^*}$ is the PDF of the $F_{d, d^*}$ distribution, and thus \eqref{thm_unknown_var_goal} holds.

To prove Theorem \ref{thm:unknown_var}, recall that $\mathcal{V}$ is pre-specified in the setting of Theorem \ref{thm:unknown_var}. Then, $Q_{J}^*=Q^*,$ and Theorem \ref{thm:unknown_var} immediately follows from Theorem \ref{thm:unknownvar_datav}.

\subsubsection{Proof of Proposition \ref{prop:quadratic}}\label{pf:quadratic}
Recall $x_\sigma(\psi)=\psi\cdot \frac{\sigma P_{\mathcal{E}}X}{\|P_{\mathcal{E}}X\|_F}+P_{\mathcal{E}}^\perp X=\psi D + E,$ and fix any $i,i'\in [n].$ 
\begin{align*}
\norm{[x_\sigma(\psi)]_i-[x_\sigma(\psi)]_{i'}}_2^2&=\|\psi D_i+E_i- (\psi D_{i'}+E_{i'})\|^2_2\\
&= \|(D_i-D_{i'})\psi + E_i-E_{i'}\|_2^2\\
&= \|\mathbf{d}_{ii'}\psi+\mathbf{e}_{ii'}\|_2^2\\
&= \|\mathbf{d}_{ii'}\|_2^2\psi^2 + 2\langle \mathbf{d}_{ii'}, \mathbf{e}_{ii'}\rangle\psi + \|\mathbf{e}_{ii'}\|^2_2\\
&=\lambda_{ii',1}\psi^2 + \lambda_{ii',2}\psi + \lambda_{ii',3}.
\end{align*}
Fix any $i\in [n], l\in[K],$ and $j\in [J].$
\begin{align*}
\|[x_\sigma(\psi)]_i-M_l^{(j-1)}(x_\sigma(\psi))\|_2^2&=\|{\psi D_i + E_i- \sum_{s=1}^nw_{l,s}^{(j-1)}\left(\psi D_{s} + E_{s}\right)}\|_2^2\\
&=\| (D_i-\sum_{s=1}^nw_{l,s}^{(j-1)}D_{s})\psi + E_i- \sum_{s=1}^nw_{l,s}^{(j-1)}E_{s}\|_2^2\\
&= \|(D_i-M_l^{(j-1)}(D))\psi+E_i-M_l^{(j-1)}(E)\|_2^2\\
&=\|\mathbf{d}_{ilj}\psi+\mathbf{e}_{ilj}\|_2^2\\
&= \|\mathbf{d}_{ilj}\|_2^2\psi^2 + 2\langle \mathbf{d}_{ilj}, \mathbf{e}_{ilj}\rangle\psi + \|\mathbf{e}_{ilj}\|^2_2\\
&=\lambda_{ilj,1}\psi^2 + \lambda_{ilj,2}\psi + \lambda_{ilj,3}.
\end{align*}

\subsubsection{Proof of Proposition \ref{prop:group}}\label{pf:group}
Recall $x_\sigma(\psi)=\psi\frac{\sigma P_{\mathcal{E}}X}{\|P_{\mathcal{E}}X\|_F}+P_{\mathcal{E}}^\perp X=\psi D + E,$ and fix any $\mathbf{v}\in \mathcal{V}.$ 
\begin{align*}
    \|[x_{\sigma}(\psi)]^\top \mathbf{v} \|^2&=\|(\psi D + E)^\top \mathbf{v}\|^2_2\\
    &=\|\psi D^\top \mathbf{v} + E^\top \mathbf{v}\|^2_2\\
    &= \|\psi \mathbf{d}_{\mathbf{v}} + \mathbf{e}_{\mathbf{v}}\|^2_2\\
    &=\|\mathbf{d}_{\mathbf{v}}\|_2^2\psi^2+2\langle \mathbf{d}_{\mathbf{v}}, \mathbf{e}_{\mathbf{v}} \rangle\psi + \|\mathbf{e}_{\mathbf{v}}\|^2_2\\
    &=\lambda_{{\mathbf{v}},1}\psi^2 + \lambda_{{\mathbf{v}},2} \psi + \lambda_{{\mathbf{v}},3}. 
\end{align*}

\subsubsection{Proof of Proposition \ref{prop:unknown_cts}}\label{pf:unknown_cts}
For the simplicity of notations, let  $u_\psi=\sqrt{\frac{\psi}{\psi+r^*}}.$ Fix any $i,i',i''\in \mathbb{N}.$ We first show 
\begin{align*}
    \|[x^*(\psi)]_i-[x^*(\psi)]_{i'}\|^2_2\leq \|[x^*(\psi)]_i-[x^*(\psi)]_{i''}\|^2_2\iff f_{ii'}(\psi)\leq f_{ii''}(\psi).
\end{align*}
Since $\|P_{\mathcal{E}}X\|_F^2+\|P_{1}X\|_F^2>0,$ we have 
\begin{align*}
    &\frac{1}{\|P_{\mathcal{E}}X\|_F^2+\|P_{1}X\|_F^2}\|[x^*(\psi)]_i-[x^*(\psi)]_{i'}\|^2_2\\
    &=\|(u_\psi A_i + \sqrt{1-u_\psi^2} B_i  +\mathcal{C}_i) -(u_\psi A_{i'} + \sqrt{1-u_\psi^2} B_{i'}  +\mathcal{C}_{i'})\|_2^2\\
    &=\|(A_i-A_{i'})u_\psi + (B_i-B_{i'})\sqrt{1-u^2_\psi}+\mathcal{C}_i-\mathcal{C}_{i'}\|^2_2\\
    &= \|\mathbf{a}_{ii'}u_\psi + \mathbf{b}_{ii'}\sqrt{1-u^2_\psi}+\mathbf{c}_{ii'}\|^2_2\\
    &=\|\mathbf{a}_{ii'}\|_2^2u_\psi^2 +
    \|\mathbf{b}_{ii'}\|_2^2(1-u^2_\psi) + 
    \|\mathbf{c}_{ii'}\|^2_2 +2\langle \mathbf{a}_{ii'}, \mathbf{b}_{ii'}\rangle u_\psi\sqrt{1-u^2_\psi} + 2\langle \mathbf{a}_{ii'},\mathbf{c}_{ii'}\rangle u_\psi\\
    &\qquad + 2\langle \mathbf{b}_{ii'}, \mathbf{c}_{ii'}\rangle\sqrt{1-u_\psi^2}\\
    &= (\|\mathbf{a}_{ii'}\|_2^2-\|\mathbf{b}_{ii'}\|_2^2)u_\psi^2 + 2\langle \mathbf{a}_{ii'}, \mathbf{b}_{ii'}\rangle u_\psi\sqrt{1-u^2_\psi}+2\langle \mathbf{a}_{ii'},\mathbf{c}_{ii'} \rangle u_\psi + 2\langle \mathbf{b}_{ii'}, \mathbf{c}_{ii'}\rangle\sqrt{1-u_\psi^2}\\
    &\qquad +\|\mathbf{b}_{ii'}\|^2_2 + \|\mathbf{c}_{ii'}\|^2_2.
\end{align*}
Substituting $u_\psi=\sqrt{\frac{\psi}{\psi+r^*}},$ we have 
\begin{align*}
&\frac{1}{\|P_{\mathcal{E}}X\|_F^2+\|P_{1}X\|_F^2}\|[x^*(\psi)]_i-[x^*(\psi)]_{i'}\|^2_2\\
&=(\|\mathbf{a}_{ii'}\|_2^2-\|\mathbf{b}_{ii'}\|_2^2)\frac{\psi}{\psi+r^*} + 2\langle \mathbf{a}_{ii'}, \mathbf{b}_{ii'}\rangle \sqrt{\frac{\psi}{\psi+r^*}}\sqrt{\frac{r^*}{\psi+r^*}}+2\langle \mathbf{a}_{ii'},\mathbf{c}_{ii'} \rangle \sqrt{\frac{\psi}{\psi+r^*}}\\
&\qquad+ 2\langle \mathbf{b}_{ii'}, \mathbf{c}_{ii'}\rangle\sqrt{\frac{r^*}{\psi+r^*}} +\|\mathbf{b}_{ii'}\|^2_2 + \|\mathbf{c}_{ii'}\|^2_2.
\end{align*}
Thus, 
\begin{align*}
&\frac{\psi+r^*}{\|P_{\mathcal{E}}X\|_F^2+\|P_{1}X\|_F^2}\|[x'(\psi)]_i-[x'(\psi)]_{i'}\|^2_2\\
&=(\|\mathbf{a}_{ii'}\|_2^2-\|\mathbf{b}_{ii'}\|_2^2)\psi+ 2\langle \mathbf{a}_{ii'}, \mathbf{b}_{ii'}\rangle \sqrt{r^*}\sqrt{\psi}+2\langle \mathbf{a}_{ii'},\mathbf{c}_{ii'} \rangle \sqrt{\psi}\sqrt{\psi+r^*}\\
&\qquad+ 2\langle \mathbf{b}_{ii'}, \mathbf{c}_{ii'}\rangle\sqrt{r^*}\sqrt{\psi+r^*} +(\|\mathbf{b}_{ii'}\|^2_2 + \|\mathbf{c}_{ii'}\|^2_2)(\psi+r^*)\\
&=(\|\mathbf{a}_{ii'}\|_2^2+\|\mathbf{c}_{ii'}\|_2^2)\psi+ 2\langle \mathbf{a}_{ii'}, \mathbf{b}_{ii'}\rangle \sqrt{r^*}\sqrt{\psi}+2\langle \mathbf{a}_{ii'},\mathbf{c}_{ii'} \rangle \sqrt{\psi}\sqrt{\psi+r^*}\\
&\qquad+ 2\langle \mathbf{b}_{ii'}, \mathbf{c}_{ii'}\rangle\sqrt{r^*}\sqrt{\psi+r^*} +(\|\mathbf{b}_{ii'}\|^2_2 + \|\mathbf{c}_{ii'}\|^2_2)r^*\\
&=\lambda_{ii',1}\psi + \lambda_{ii',2}\sqrt{\psi}+\lambda_{ii',3}\sqrt{\psi}\sqrt{\psi+r^*}+\lambda_{ii',4}\sqrt{\psi+r^*}+\lambda_{ii',5}.
\end{align*}
Similarly, we have 
\begin{multline*}
\frac{\psi+r^*}{\|P_{\mathcal{E}}X\|_F^2+\|P_{1}X\|_F^2}\|[x^*(\psi)]_i-[x^*(\psi)]_{i''}\|^2_2=\\
\lambda_{ii'',1}\psi + \lambda_{ii'',2}\sqrt{\psi}+\lambda_{ii'',3}\sqrt{\psi}\sqrt{\psi+r^*}+\lambda_{ii'',4}\sqrt{\psi+r^*}+\lambda_{ii'',5}.
\end{multline*}
Then, since $\psi+r^*>0,$ it follows that 
\begin{align*}
\|[x^*(\psi)]_i-[x^*(\psi)]_{i'}\|^2_2\leq \|[x^*(\psi)]_i-[x^*(\psi)]_{i''}\|^2_2\iff f_{ii'}(\psi)\leq f_{ii''}(\psi).
\end{align*}
Fix any $l,l'\in [K].$ We next show 
\begin{align*}
\|[x^*(\psi)]_i- M_l^{(j-1)}(x^*(\psi))\|^2_2\leq \|[x^*(\psi)]_i- M_{l'}^{(j-1)}(x^*(\psi))\|^2_2\iff f_{ijl}(\psi)\leq f_{ijl'}(\psi).
\end{align*}
Again, by $\|P_{\mathcal{E}}X\|_F^2+\|P_{1}X\|_F^2>0,$ we have  
\begin{align*}
  &\frac{1}{\|P_{\mathcal{E}}X\|_F^2+\|P_{1}X\|_F^2}\|[x^*(\psi)]_i- M_l^{(j-1)}(x^*(\psi))\|^2_2\\
  &=\frac{1}{\|P_{\mathcal{E}}X\|_F^2+\|P_{1}X\|_F^2}\|{[x^*(\psi)]_i- \sum_{s=1}^n w_{l,s}^{(j-1)}  [x^*(\psi)]_s }\|^2_2 \\
  &=\|(u_\psi A_i + \sqrt{1-u_\psi^2} B_i  +\mathcal{C}_i) - \sum_{s=1}^n w_{l,s}^{(j-1)} (u_\psi A_s + \sqrt{1-u_\psi^2} B_s  +\mathcal{C}_s)\|^2_2\\
  &=\|(A_i-\sum_{s=1}^n w_{l,s}^{(j-1)}A_s)u_\psi + (B_i-\sum_{s=1}^n w_{l,s}^{(j-1)}B_s)\sqrt{1-u_\psi^2}\\
  &\qquad +\mathcal{C}_i-\sum_{s=1}^n w_{l,s}^{(j-1)}\mathcal{C}_s\|^2_2\\
  &= \|(A_i-M_l^{(j-1)}(A))u_\psi +(B_i-M_l^{(j-1)}(B))\sqrt{1-u_\psi^2}+\mathcal{C}_i-M_l^{(j-1)}(C)\|_2^2\\
  &=\|\mathbf{a}_{ijl}u_\psi + \mathbf{b}_{ijl}\sqrt{1-u_\psi^2} +\mathbf{c}_{ijl}\|^2_2\\
  &=\|\mathbf{a}_{ijl}\|_2^2u_\psi^2 +
    \|\mathbf{b}_{ijl}\|_2^2(1-u^2_\psi) + 
    \|\mathbf{c}_{ijl}\|^2_2 +2\langle \mathbf{a}_{ijl}, \mathbf{b}_{ijl}\rangle u_\psi\sqrt{1-u^2_\psi} + 2\langle \mathbf{a}_{ijl},\mathbf{c}_{ijl}\rangle u_\psi\\
    &\qquad + 2\langle \mathbf{b}_{ijl}, \mathbf{c}_{ijl}\rangle\sqrt{1-u_\psi^2}\\
    &= (\|\mathbf{a}_{ijl}\|_2^2-\|\mathbf{b}_{ijl}\|_2^2)u_\psi^2 + 2\langle \mathbf{a}_{ijl}, \mathbf{b}_{ijl}\rangle u_\psi\sqrt{1-u^2_\psi}+2\langle \mathbf{a}_{ijl},\mathbf{c}_{ijl} \rangle u_\psi + 2\langle \mathbf{b}_{ijl}, \mathbf{c}_{ijl}\rangle\sqrt{1-u_\psi^2}\\
    &\qquad +\|\mathbf{b}_{ijl}\|^2_2 + \|\mathbf{c}_{ijl}\|^2_2.
\end{align*}
Substituting $u_\psi=\sqrt{\frac{\psi}{\psi+r^*}},$ we have 
\begin{align*}
&\frac{1}{\|P_{\mathcal{E}}X\|_F^2+\|P_{1}X\|_F^2}\|[x^*(\psi)]_i- M_l^{(j-1)}(x^*(\psi))\|^2_2\\
&=(\|\mathbf{a}_{ijl}\|_2^2-\|\mathbf{b}_{ijl}\|_2^2)\frac{\psi}{\psi+r^*} + 2\langle \mathbf{a}_{ijl}, \mathbf{b}_{ijl}\rangle \sqrt{\frac{\psi}{\psi+r^*}}\sqrt{\frac{r^*}{\psi+r^*}}+2\langle \mathbf{a}_{ijl},\mathbf{c}_{ijl} \rangle \sqrt{\frac{\psi}{\psi+r^*}}\\
&\qquad+ 2\langle \mathbf{b}_{ijl}, \mathbf{c}_{ijl}\rangle\sqrt{\frac{r^*}{\psi+r^*}} +\|\mathbf{b}_{ijl}\|^2_2 + \|\mathbf{c}_{ijl}\|^2_2.
\end{align*}
Thus, 
\begin{align*}
&\frac{\psi+r^*}{\|P_{\mathcal{E}}X\|_F^2+\|P_{1}X\|_F^2}\|[x^*(\psi)]_i- M_l^{(j-1)}(x^*(\psi))\|^2_2\\
&=(\|\mathbf{a}_{ijl}\|_2^2-\|\mathbf{b}_{ijl}\|_2^2)\psi+ 2\langle \mathbf{a}_{ijl}, \mathbf{b}_{ijl}\rangle \sqrt{r^*}\sqrt{\psi}+2\langle \mathbf{a}_{ijl},\mathbf{c}_{ijl} \rangle \sqrt{\psi}\sqrt{\psi+r^*}\\
&\qquad+ 2\langle \mathbf{b}_{ijl}, \mathbf{c}_{ijl}\rangle\sqrt{r^*}\sqrt{\psi+r^*} +(\|\mathbf{b}_{ijl}\|^2_2 + \|\mathbf{c}_{ijl}\|^2_2)(\psi+r^*)\\
&=(\|\mathbf{a}_{ijl}\|_2^2+\|\mathbf{c}_{ijl}\|_2^2)\psi+ 2\langle \mathbf{a}_{ijl}, \mathbf{b}_{ijl}\rangle \sqrt{r^*}\sqrt{\psi}+2\langle \mathbf{a}_{ijl},\mathbf{c}_{ijl} \rangle \sqrt{\psi}\sqrt{\psi+r^*}\\
&\qquad+ 2\langle \mathbf{b}_{ijl}, \mathbf{c}_{ijl}\rangle\sqrt{r^*}\sqrt{\psi+r^*} +(\|\mathbf{b}_{ijl}\|^2_2 + \|\mathbf{c}_{ijl}\|^2_2)r^*\\
&=\lambda_{ijl,1}\psi + \lambda_{ijl,2}\sqrt{\psi}+\lambda_{ijl,3}\sqrt{\psi}\sqrt{\psi+r^*}+\lambda_{ijl,4}\sqrt{\psi+r^*}+\lambda_{ijl,5}.
\end{align*}
Similarly, we have 
\begin{multline*}
\frac{\psi+r^*}{\|P_{\mathcal{E}}X\|_F^2+\|P_{1}X\|_F^2}\|[x^*(\psi)]_i- M_{l'}^{(j-1)}(x^*(\psi))\|^2_2=\\\lambda_{ijl',1}\psi + \lambda_{ijl',2}\sqrt{\psi}+\lambda_{ijl',3}\sqrt{\psi}\sqrt{\psi+r^*}+\lambda_{ijl',4}\sqrt{\psi+r^*}+\lambda_{ijl',5}.
\end{multline*}
Then, since $\psi+r^*>0,$ it follows that
\begin{align*}
\|[x^*(\psi)]_i- M_l^{(j-1)}(x^*(\psi))\|^2_2\leq \|[x^*(\psi)]_i- M_{l'}^{(j-1)}(x^*(\psi))\|^2_2\iff f_{ijl}(\psi)\leq f_{ijl'}(\psi).
\end{align*}

\subsubsection{Proof of Proposition \ref{V_dep_unknown}}\label{pf:V_dep_unknown}
For the simplicity of notations, let  $u_\psi=\sqrt{\frac{\psi}{\psi+r^*}}.$ Fix any $\mathbf{v},\mathbf{v}'\in \mathbb{R}^n.$ We first show 
\begin{align*}
    \|[x^*(\psi)]^\top \mathbf{v}\|_2^2\leq \|[x^*(\psi)]^\top \mathbf{v}'\|_2^2\iff h_{\mathbf{v}}(\psi)\leq h_{\mathbf{v}'}(\psi).
\end{align*}
Since $\|P_{\mathcal{E}}X\|_F^2+\|P_{1}X\|_F^2>0,$ we have 
\begin{align*}
    &\frac{1}{\|P_{\mathcal{E}}X\|_F^2+\|P_{1}X\|_F^2}\|[x^*(\psi)]^\top \mathbf{v}\|_2^2\\
    &=\|(u_\psi A + \sqrt{1-u_\psi^2} B  +C)^\top \mathbf{v}\|^2_2\\
    &= \|A^\top \mathbf{v} u_\psi + B^\top \mathbf{v} \sqrt{1-u_\psi^2}+C^\top \mathbf{v}\|^2_2\\
    &= \|\mathbf{a}_{\mathbf{v}}u_\psi + \mathbf{b}_{\mathbf{v}}\sqrt{1-u_\psi^2}+\mathbf{c}_{\mathbf{v}}\|^2_2\\
    &=\|\mathbf{a}_{\mathbf{v}}\|_2^2u_\psi^2 + \|\mathbf{b}_{\mathbf{v}}\|_2^2(1-u_\psi^2)+\|\mathbf{c}_{\mathbf{v}}\|_2^2+2\langle \mathbf{a}_{\mathbf{v}},\mathbf{b}_{\mathbf{v}}\rangle u_\psi \sqrt{1-u_\psi^2}+2\langle \mathbf{a}_{\mathbf{v}},\mathbf{c}_{\mathbf{v}}\rangle u_\psi + 2\langle \mathbf{b}_{\mathbf{v}}, \mathbf{c}_{\mathbf{v}}\rangle\sqrt{1-u_\psi^2}\\
    &= (\|\mathbf{a}_{\mathbf{v}}\|_2^2-\|\mathbf{b}_{\mathbf{v}}\|_2^2)u_\psi^2 +2\langle \mathbf{a}_{\mathbf{v}},\mathbf{b}_{\mathbf{v}}\rangle u_\psi \sqrt{1-u_\psi^2}+2\langle \mathbf{a}_{\mathbf{v}},\mathbf{c}_{\mathbf{v}}\rangle u_\psi + 2\langle \mathbf{b}_{\mathbf{v}}, \mathbf{c}_{\mathbf{v}}\rangle\sqrt{1-u_\psi^2}+\|\mathbf{b}_{\mathbf{v}}\|_2^2+\|\mathbf{c}_{\mathbf{v}}\|^2_2
\end{align*}
Substituting $u_\psi=\sqrt{\frac{\psi}{\psi+r^*}},$ we have 
\begin{align*}
&\frac{1}{\|P_{\mathcal{E}}X\|_F^2+\|P_{1}X\|_F^2}\|[x^*(\psi)]^\top \mathbf{v}\|_2^2\\
&=(\|\mathbf{a}_{\mathbf{v}}\|_2^2-\|\mathbf{b}_{\mathbf{v}}\|_2^2)\frac{\psi}{\psi+r^*} +2\langle \mathbf{a}_{\mathbf{v}},\mathbf{b}_{\mathbf{v}}\rangle \sqrt{\frac{\psi}{\psi+r^*}} \sqrt{\frac{r^*}{\psi+r^*}}+2\langle \mathbf{a}_{\mathbf{v}},\mathbf{c}_{\mathbf{v}}\rangle \sqrt{\frac{\psi}{\psi+r^*}} \\
    &\qquad + 2\langle \mathbf{b}_{\mathbf{v}}, \mathbf{c}_{\mathbf{v}}\rangle\sqrt{\frac{r^*}{\psi+r^*}}+\|\mathbf{b}_{\mathbf{v}}\|_2^2+\|\mathbf{c}_{\mathbf{v}}\|^2_2.
\end{align*}
Thus, 
\begin{align}
&\frac{\psi+r^*}{\|P_{\mathcal{E}}X\|_F^2+\|P_{1}X\|_F^2}\|[x^*(\psi)]^\top \mathbf{v}\|_2^2\notag\\
&=(\|\mathbf{a}_{\mathbf{v}}\|_2^2-\|\mathbf{b}_{\mathbf{v}}\|_2^2)\psi +2\langle \mathbf{a}_{\mathbf{v}},\mathbf{b}_{\mathbf{v}}\rangle \sqrt{r^*}\sqrt{\psi} +2\langle \mathbf{a}_{\mathbf{v}},\mathbf{c}_{\mathbf{v}}\rangle \sqrt{\psi}\sqrt{\psi+r^*} \notag\\
&\qquad + 2\langle \mathbf{b}_{\mathbf{v}}, \mathbf{c}_{\mathbf{v}}\rangle\sqrt{r^*}\sqrt{\psi+r^*}+(\|\mathbf{b}_{\mathbf{v}}\|_2^2+\|\mathbf{c}_{\mathbf{v}}\|^2_2)(\psi+r^*)\notag\\
&=(\|\mathbf{a}_{\mathbf{v}}\|_2^2+\|\mathbf{c}_{\mathbf{v}}\|_2^2)\psi +2\langle \mathbf{a}_{\mathbf{v}},\mathbf{b}_{\mathbf{v}}\rangle \sqrt{r^*}\sqrt{\psi} +2\langle \mathbf{a}_{\mathbf{v}},\mathbf{c}_{\mathbf{v}}\rangle \sqrt{\psi}\sqrt{\psi+r^*} \label{eq:prop_v_dep}\notag \\
&\qquad + 2\langle \mathbf{b}_{\mathbf{v}}, \mathbf{c}_{\mathbf{v}}\rangle\sqrt{r^*}\sqrt{\psi+r^*} +(\|\mathbf{b}_{\mathbf{v}}\|_2^2+\|\mathbf{c}_{\mathbf{v}}\|_2^2)r^*.\notag\\
&=\lambda_{v,1}\psi + \lambda_{v,2}\sqrt{\psi}+\lambda_{v,3}\sqrt{\psi}\sqrt{\psi+r^*}+\lambda_{v,4}\sqrt{\psi+r^*}+\lambda_{v,5}.
\end{align}
Similarly, we have 
\begin{multline*}
\frac{\psi+r^*}{\|P_{\mathcal{E}}X\|_F^2+\|P_{1}X\|_F^2}\|[x^*(\psi)]^\top \mathbf{v}'\|_2^2=
\lambda_{\mathbf{v}',1}\psi + \lambda_{\mathbf{v}',2}\sqrt{\psi}+\lambda_{\mathbf{v}',3}\sqrt{\psi}\sqrt{\psi+r^*}+\lambda_{\mathbf{v}',4}\sqrt{\psi+r^*}+\lambda_{\mathbf{v}',5}.
\end{multline*}
Then, since $\psi+r^*>0,$ it follows that
\begin{align*}
\|[x^*(\psi)]^\top \mathbf{v}\|_2^2\leq \|[x^*(\psi)]^\top \mathbf{v}'\|_2^2\iff h_{\mathbf{v}}(\psi)\leq h_{\mathbf{v}'}(\psi).
\end{align*}
Fix any $\mathbf{v}\in\mathbb{R}^n$ and $t>0.$ We next show 
\begin{align*}
    \|[x^*(\psi)]^\top \mathbf{v}\|_2^2\leq t\iff h_{\mathbf{v}}(\psi)\leq h(\psi).
\end{align*}
By \eqref{eq:prop_v_dep}, we have 
\[\frac{\psi+r^*}{\|P_{\mathcal{E}}X\|_F^2+\|P_{1}X\|_F^2}\|[x^*(\psi)]^\top \mathbf{v}\|_2^2=
h_{\mathbf{v}}(\psi).\]
Then, 
\begin{align*}
&\|[x^*(\psi)]^\top \mathbf{v}\|_2^2\leq t\\
&\iff \frac{\psi+r^*}{\|P_{\mathcal{E}}X\|_F^2+\|P_{1}X\|_F^2}\|[x^*(\psi)]^\top \mathbf{v}\|_2^2\leq \frac{t}{\|P_{\mathcal{E}}X\|_F^2+\|P_{1}X\|_F^2}\psi + \frac{r^*t}{\|P_{\mathcal{E}}X\|_F^2+\|P_{1}X\|_F^2}\\
&\iff h_{\mathbf{v}}(\psi)\leq h(\psi).
\end{align*}

\subsection{Notes on implementation}\label{app:impl}

All of the implementation is done in \verb|R|. Parts of the implementation of the proposed methods have been adapted from the source code of the package \verb|KmeansInference| of \citet{chen2023selective} or call its functions---more details can be found in the codes. In the rest of this section, we elaborate on various aspects of the implementation of the proposed tests.

\subsubsection{The proposed test of Section \ref{sec:proposed}}
 
\paragraph{Computation of $\mathrm{dim}(\mathcal{E})$}

A challenge in computing $\mathrm{dim}(\mathcal{E})$ given a set of vectors $\{\mathbf{v}_{k,k'}:(k,k')\in \mathcal{V}\}$ 
lies in finite precision. For instance, it may be the case that two linearly dependent vectors are numerically linearly independent due to finite precision. To address this phenomenon, we compute $\mathrm{dim}(\mathcal{E})$ differently depending on the value of $|\mathcal{V}|.$ 
\begin{itemize}
    \item If $|\mathcal{V}|=\binom{K}{2},$ then we know that $\mathrm{dim}(\mathcal{E})=K-1$ since the set $\{\mathbf{v}_{k,k+1}:k\in[K-1]\}$ forms a basis for $\mathrm{dim}(\mathcal{E})$ by Lemma \ref{lem:span} below. 
    \item If $|\mathcal{V}|<\binom{K}{2},$ we resort to numerical methods for computing $\mathrm{dim}(\mathcal{E}).$ We use the function \verb|fast.svd| from the package \verb|corpcor| for computing the condensed singular value decomposition (SVD) that gives the condensed SVD  $U\Sigma W^\top$ of $V,$ where $V$ is the matrix whose columns consist of the vectors in $\{\mathbf{v}_{k,k'}:(k,k')\in \mathcal{V}\}.$ In the case where $V$ does not have full rank, it is likely to be the case that $V$ still has full numerical rank due to finite precision. To account for finite precision, we set the rank of $V,$ i.e., $\mathrm{dim}(\mathcal{E}),$ to be the number of diagonal entries of $\Sigma,$ which corresponds to the number of ``nonzero" singular values of $V,$ where the default threshold value of the function \verb|fast.svd| is used for determining if a singular value is to be treated as zero. 
\end{itemize}

\paragraph{Computation of $P_{\mathcal{E}}$} To be consistent with the computation of $\mathrm{dim}(\mathcal{E})$ discussed above, we compute $P_{\mathcal{E}}$ differently depending on the value of $|\mathcal{V}|.$
\begin{itemize}
    \item If $|\mathcal{V}|=\binom{K}{2},$ we compute the projection matrix $P_{\mathcal{E}}$ by setting $P_{\mathcal{E}}=\widetilde{U}\widetilde{U}^\top,$ where $\widetilde{U}$ is an orthogonal matrix whose columns are generated by the Gram-Schmidt orthogonalization procedure applied to the columns of $\begin{bmatrix}
    v_{1,2}&\cdots&v_{K-1,K}\end{bmatrix},$ using the function \verb|gramSchmidt| from the package \verb|pracma|. 
    \item If $|\mathcal{V}|<\binom{K}{2},$ we compute $P_{\mathcal{E}}$ by setting $P_{\mathcal{E}}=UU^\top,$ where $U$ is as defined above. 
\end{itemize}
Here, we state and prove Lemma \ref{lem:span}, which is referred to in the discussion above. 
\begin{lemma}\label{lem:span}
   Suppose $|\mathcal{V}|=\binom{K}{2}.$ Then, the set $\left\{v_{k, k+1}:k\in [K-1]\right\}$ forms a basis of $\mathcal{E}.$
\end{lemma}
\begin{proof}\label{pf:span} Note that if $|\mathcal{V}|=\binom{K}{2},$ then $\mathcal{E}=\mathrm{span}\{v_{k,k'}:k<k'\in [K]\}.$
  For any $k,k'\in [K]$ such that $k<k',$ note that 
\begin{align*}
    \mathbf{v}_{k,k'}&=\frac{1}{|\mathcal{C}_k|}\mathbf{1}_{\mathcal{C}_k}-\frac{1}{|\mathcal{C}_{k'}|}\mathbf{1}_{\mathcal{C}_{k'}}\\
    &=\frac{1}{|\mathcal{C}_k|}\mathbf{1}_{\mathcal{C}_k}-\frac{1}{|\mathcal{C}_{k+1}|}\mathbf{1}_{\mathcal{C}_{k+1}}+\cdots+\frac{1}{|\mathcal{C}_{k'-1}|}\mathbf{1}_{\mathcal{C}_{k'-1}}-\frac{1}{|\mathcal{C}_{k'}|}\mathbf{1}_{\mathcal{C}_{k'}}\\
    &=\mathbf{v}_{k,k+1}+\dots +\mathbf{v}_{k'-1,k'}.
\end{align*}
Thus, we have 
$\mathcal{E}\subset \mathrm{span}\left\{\mathbf{v}_{k,k+1}:k\in [K-1]\right\}.$ The other direction trivially holds.

We next show that $\mathbf{v}_{k,k+1}$ for $k\in [K-1]$ are linearly independent. Suppose $\summ{k=1}{K-1}\lambda_k \mathbf{v}_{k,k+1}=0.$ Note
\begin{align*}
\overset{K-1}{\sum_{k=1}}\lambda_k \mathbf{v}_{k,k+1}&= \overset{K-1}{\sum_{k=1}}\lambda_k\left(\frac{1}{|\mathcal{C}_k|}\mathbf{1}_{\mathcal{C}_k}-\frac{1}{|\mathcal{C}_{k+1}|}\mathbf{1}_{\mathcal{C}_{k+1}}\right)\\
&=\lambda_1\frac{1}{|\mathcal{C}_1|}\mathbf{1}_{\mathcal{C}_1}+\sum_{k=1}^{K-2}(\lambda_{k+1}-\lambda_k)\frac{1}{|\mathcal{C}_{k+1}|}\mathbf{1}_{\mathcal{C}_{k+1}}-\lambda_{K-1}\frac{1}{|\mathcal{C}_K|}\mathbf{1}_{\mathcal{C}_K}.
\end{align*}
Define $u\coloneqq\lambda_1\frac{1}{|\mathcal{C}_1|}\mathbf{1}_{\mathcal{C}_1}+\sum_{k=1}^{K-2}(\lambda_{k+1}-\lambda_k)\frac{1}{|\mathcal{C}_{k+1}|}\mathbf{1}_{\mathcal{C}_{k+1}}-\lambda_{K-1}\frac{1}{|\mathcal{C}_K|}\mathbf{1}_{\mathcal{C}_K},$ and let $u_j$ denote $j$th entry of $u.$ Since $u_j=\lambda_1\frac{1}{|\mathcal{C}_1|}=0$ for $j\in C_1,$ we have $\lambda_1=0.$ Then, since ${u_j=(\lambda_2-\lambda_1)\frac{1}{|\mathcal{C}_2|}=\lambda_2\frac{1}{|\mathcal{C}_2|}=0}$ for $j\in \mathcal{C}_2,$ we have $\lambda_2=0.$ The same argument holds for $\lambda_3,...,\lambda_{K-2}.$ Finally, since $u_j=\lambda_{K-1}\frac{1}{|\mathcal{C}_K|}=0$ for $j\in \mathcal{C}_K,$ we have $\lambda_{K-1}=0.$ Thus, $\lambda_k=0$ for all $k\in [K-1],$ and it follows that $\mathbf{v}_{k,k+1}$ for $k\in [K-1]$ are linearly independent.  
\end{proof}

\subsubsection{The proposed test of Section \ref{sec:ftest}}\label{sec:imple_5.2}

By Proposition \ref{prop:unknown_cts}, $S^*$ is a set of solutions to a system of inequalities, where each inequality takes the form 
\begin{align}\label{S*_computation}
    \lambda_1\psi + \lambda_2\sqrt{\psi} + \lambda_3 \sqrt{\psi}\sqrt{\psi+r^*}+\lambda_4\sqrt{\psi+r^*}+\lambda_5\leq 0.
\end{align}
In this section, we elaborate on the implementation for finding the set of solutions to \eqref{S*_computation}. Since $u\mapsto u^2$ on $[0, \infty)$ is a bijective map, the set of values of $\sqrt{\psi}$ satisfying the inequality is bijective with the set of values of $\psi$ satisfying the same inequality. Thus, we aim to find the sub-level set of $f$ at 0, where $f:[0, \infty)\rightarrow \mathbb{R},$ 
\begin{align}\label{S*_computation_rewritten}
    f(y)=\lambda_1y^2 + \lambda_2 y + \lambda_3 y\sqrt{y^2+r^*}+\lambda_4\sqrt{y^2+r^*}+\lambda_5.
\end{align}
Define $\mathcal{A}_f=\{y\in [0, \infty):f(y)=0\}$ to be the set of non-negative real roots of $f.$ Since $f$ is continuous, intermediate value theorem implies that for any $y_1<y_2\in [0, \infty),$ $\mathrm{sign}(f(y_1))=-\mathrm{sign}(f(y_2))$ only if there exists $y\in [y_1,y_2]$ such that $y\in \mathcal{A}_f,$ where $\mathrm{sign}(a)$ for $a\in\mathbb{R}$ is equal to 1 if $a>0,$ -1 if $a<0,$ and 0 if $a=0.$ Once we compute $\mathcal{A}_f,$ we find the set of solutions to the inequality $f(y)\leq 0$ by checking the sign of $f$ on each interval partitioned by its roots, i.e., the elements of $\mathcal{A}_f.$ To compute $\mathcal{A}_f,$ write
\begin{align*}
    f(y)=0\iff (\lambda_3y+\lambda_4)\sqrt{y^2+r^*}=-\lambda_1y^2-\lambda_2y-\lambda_5
\end{align*}
and define $f_1:[0, \infty)\rightarrow\mathbb{R},$\[f_1(y)=(\lambda_3y+\lambda_4)\sqrt{y^2+r^*}\] and $f_2:[0, \infty)\rightarrow \mathbb{R},$
\[f_2(y)=-\lambda_1y^2-\lambda_2y-\lambda_5.\] 
Note 
\begin{align}
    &f_1^2(y)=f_2^2(y)\notag\\
    &\iff(\lambda_3y+\lambda_4)^2(y^2+r^*)=(\lambda_1y^2+\lambda_2y+\lambda_5)^2\notag\\
    &\iff(\lambda_3^2y^2+2\lambda_3\lambda_4y+\lambda_4^2)(y^2+r^*)=\lambda_1^2y^4+\lambda_2^2y^2+\lambda_5^2+2\lambda_1\lambda_2y^3+2\lambda_2\lambda_5y+2\lambda_1\lambda_5y^2\notag\\
    &\iff\lambda_3^2y^4+2\lambda_3\lambda_4y^3+\lambda_4^2y^2+\lambda_3^2r^*y^2+2\lambda_3\lambda_4 r^* y + \lambda_4^2 r^*=\lambda_1^2y^4+\lambda_2^2y^2+\lambda_5^2+2\lambda_1\lambda_2y^3\notag\\
    &\hspace{27em}+2\lambda_2\lambda_5y+2\lambda_1\lambda_5y^2\notag\\
    &\iff 0=(\lambda_3^2-\lambda_1^2)y^4+2(\lambda_3\lambda_4-\lambda_1\lambda_2)y^3+(\lambda_4^2+\lambda_3^2r^*-\lambda_2^2-2\lambda_1\lambda_5)y^2\label{quartic}\\
    &\qquad\qquad\qquad +2(\lambda_3\lambda_4r^*-\lambda_2\lambda_5)y+\lambda_4^2r^*-\lambda_5^2\notag,
\end{align}
which is a quartic equation in $y.$ We then solve for $\mathcal{A}_f$ by checking the condition $f_1(y)=f_2(y)$ for each $y\in \{y\in[0,\infty):f_1^2(y)=f_2^2(y)\}.$

We use the \verb|R| base function \verb|polyroot| to solve for the roots of the quartic equality in \eqref{quartic}. We now briefly discuss how we address finite precision in our implementation in the rest of the procedure. For determining whether a root is real, i.e., if its imaginary part is 0, and for checking if the condition $f_1(y)=f_2(y)$ is satisfied, we use a relatively large tolerance level of 1---note that having additional values in $\mathcal{A}_f$ leads to a finer partition of the interval, which does not alter the output of the procedure for computing the set of solutions to \eqref{S*_computation_rewritten}. To determine if there is multiplicity for a root, we use the function \verb|almost.unique| from the package \verb|bazar| with a stringent threshold of $10^{-30}$ to avoid losing any root of $f.$

\subsubsection{Simulations of Section \ref{sec:simul}}
Note that the tests proposed throughout this paper, as well as those of \citet{chen2023selective}, assume that the $K$-means algorithm outputs $K$ clusters at each iteration of the algorithm, an assumption implicitly made in the characterization of the truncation sets $\mathcal{S}_{\sigma, k,k'}$ in \eqref{eq:charac}, $\mathcal{S}_{\sigma}$ in \eqref{eq:charac2}, and $\mathcal{S}^*$ in \eqref{eq:charac3}. Thus, the tests cannot be applied with theoretical guarantees if $K$-means clustering outputs a different number of clusters at any iteration of the algorithm. 

However, there are instances in practice where the algorithm produces a different number of clusters at the final step than that at the initial step. We omit such cases in our simulations, having the corresponding functions return \verb|NA| for the p-value, and only the outputs that are not \verb|NA| are reflected in the QQ plots and the computation of the empirical powers. We have observed that the number of instances of this phenomenon is small compared to the number of p-values generated.

\end{document}